%% file: cowgill-contracting-for-decision-relevant-beliefs.tex
\documentclass[12pt]{article}
\input{includes/preamble}

\newcommand{\PaperAppendixReadingGuide}{%
\section*{Reading Guide}
\phantomsection\label{app:reading-guide}

The first two routes verify the main design results. They separate the
general test for an interior benchmark from the exact solution when the
benchmark reaches the stake bounds.

\paragraph{Pricing a cutoff and testing the interior benchmark.}
The \hyperref[app:proof-belief-only]{belief-only proof} shows why exact
truthfulness forces the decision to ignore the report when the same two
stakes can occur at every forecast. The
\hyperref[lem:withinbeliefenv]{stake envelope} and
\hyperref[prop:convex-utility-characterization]{joint reporting
characterization} supply the two incentive arguments. The
\hyperref[app:proof-threshold-implementability]{threshold proof} combines
them into two convexity conditions. The
\hyperref[app:full-curvature-characterization]{curvature proof} uses
those conditions to derive and attain the price in
Theorem~\ref{thm:finite-rent-curvature}. The
\hyperref[app:threshold-attainment-zero-rent]{zero-rent and existence
result} and the \hyperref[app:proof-pointwise-optimality]{benchmark proof}
then establish when the firm can keep the interior public-forecast rule.

\paragraph{Solving the uniform example with stake bounds.}
The \hyperref[app:private-interest-family]{exact uniform solution}
finds the optimal contract and verifies it with the
\hyperref[app:global-optimality-certificate]{global certificate}. The
\hyperref[app:distribution-robustness]{distribution comparison} uses the
uniform solution to establish the small-stakes profit-cost rate under
bounded changes in the type law.

\paragraph{Which tool answers a question about randomization?}
\hyperref[app:threshold-sufficiency]{Threshold sufficiency} replaces
lotteries by thresholds when conditional virtual surplus is concave in
the cutoff at a fixed forecast. The
\hyperref[app:randomization-bounds]{concave-envelope bound} limits the
possible gains from randomization. The
\hyperref[app:global-optimality-certificate]{global certificate} verifies
a particular cutoff by showing it attains a bound on unrestricted profit.
The \hyperref[app:regional-threshold-optimality]{regional test} gives a
necessary and sufficient test for a finite-piece-affine cutoff with its
least-cost floor under additional regularity and tie conditions.

\paragraph{Further applications.}
The \hyperref[app:approximate-response]{approximate-choice analysis}
compares two uniform menus when managers can fall short of a best
response. The \hyperref[app:partial-conjugate-extensions]{allocation
extensions} characterize other forms of implementable rules. For
computation, start with \hyperref[app:constructive-cutoff-selection]%
{constructing an optimal cutoff}, then see the
\hyperref[app:quadratic-cutoff-loss]{quadratic-loss analysis} and
\hyperref[app:clipped-affine-construction]{clipped-affine construction}.
The \hyperref[app:dependent-capacity-expansion]{dependent capacity
example} shows when the optimal adjustment eliminates rent.
}
\patchcmd{\appendixprelim}{\noindent{\large\textbf{Table of Contents}}}
  {\noindent For a route through the main proofs and their dependencies,
   start with the \hyperref[app:reading-guide]{Reading Guide}
   on page~\pageref{app:reading-guide}.\par\vspace{3mm}
   \noindent{\large\textbf{Table of Contents}}}{}{}
\patchcmd{\appendixprelim}{\printonlineappendixcontents}
  {\setlength{\parskip}{0pt}\setstretch{1.07}\footnotesize
   \noindent\hyperref[app:reading-guide]{\textbf{Reading Guide}\dotfill
     \textbf{\pageref*{app:reading-guide}}}\par
   \printonlineappendixcontents\normalsize\setstretch{1}\setlength{\parskip}{7.2pt}}{}{}
\apptocmd{\appendixprelim}{\clearpage}{}{}
\apptocmd{\appendixprelim}{\PaperAppendixReadingGuide\clearpage}{}{}

\hypersetup{
  pdftitle={Contracting for Decision-Relevant Beliefs},
  pdfauthor={Bo Cowgill},
  pdfkeywords={belief elicitation, mechanism design, multidimensional screening,
    state-contingent contracts, implementation rent, decision markets}
}

\begin{document}
\onehalfspacing
\setlength{\emergencystretch}{1em}

\title{\Huge Contracting for Decision-Relevant Beliefs\thanks{The author
thanks Fabrizio Dell'Acqua, Joshua Gans, Kris Gulati, Rubi Hudson, Brian Jabarian, Scott Kominers,
Barton Lee, Don Moore, Chad Syverson, Zikai Xu, and Eric Zitzewitz for constructive discussions
and feedback, and participants at the Institutions and Innovation Conference
(Columbia), Manifest Conference (Berkeley), and the Stigler Fellows Conference
(Chicago Booth) for helpful comments. I thank the Kauffman Foundation Emerging
Scholars Program, the Columbia Center for Political Economy, the Institute for
Humane Studies (grant no.\ IHS018366), the NET Institute, the Stellar Development
Foundation, and the Columbia-Ethereum Center for Blockchain Protocol Design. I
also acknowledge the Russell Sage Foundation, where I was a Visiting Scholar
during part of this project. Refine.ink and GPT\mbox{-}5.6 Pro were used to check
the paper for consistency and clarity.}}

\author{Bo Cowgill\thanks{Assistant Professor, University of Toronto, IZA, and CESifo. \href{mailto:bo@bocowgill.org}{bo@bocowgill.org}.}}

\date{\today\phantom{-}(First Version: August 2026) \\ \ \\ \emph{Working Paper}\normalsize}


\maketitle\thispagestyle{empty}
\begin{abstract}
\begin{singlespace}
A manager may overstate demand to win approval for an expansion she
stands to gain from. How should the firm decide which expansions to
approve when honest reporting can be costly? We study threshold rules
that set a minimum stake for approval at each forecast. We characterize
the firm’s profit-maximizing approval rule in this class, accounting for
the rents it must leave managers to elicit truthful forecasts. In some
cases, the firm declines expansions with positive expected profits to
reduce the cost of eliciting truthful forecasts.

\vspace{5mm}
\vspace{1.5mm} 
\end{singlespace}
\end{abstract}

\pagebreak
\section{Introduction}
\label{sec:introduction}

A division manager who wants a larger budget may overstate demand to
win approval for a capacity expansion. She privately knows both the probability
of high demand and how much she stands to gain or lose from expanding.
The firm can commit to an approval and bonus policy before she reports,
then settle payments against an index of industry orders recorded after
the decision. How should the firm decide which expansions to approve when
honest reporting can be costly?

We motivate our results by why joint reporting is valuable in this environment. If the same two distinct
stakes are possible at every forecast, exact truthfulness in a mechanism
that asks only for the forecast forces the expansion choice to be
unresponsive to the forecast, even with state-contingent payments
(Theorem~\ref{thm:belief-only-impossibility}).

Our main results are about \emph{cutoff rules}: 
rules that approve expansion if the manager's stake meets or  exceeds a minimum stake (at each forecast). Our first result is  
that the elicitation cost of any proposed cutoff 
policy is a function of its curvature --- how the slope of the cutoff changes with the forecast. Upward bends in the 
cutoff require implementation rents, while downward 
bends do not (Theorem~\ref
{thm:finite-rent-curvature}). 

This curvature result characterizes implementation costs for any proposed cutoff rule. It can therefore be used to select the best cutoff rule for the firm's expansion decision --- the cutoff rule that maximizes the benefits of expansion minus the  elicitation rents.
As a benchmark, suppose the firm observes the manager's forecast.
This reduces the problem from two dimensions of private information
to one, leaving only her stake private.\footnote{%
\citet[Propositions 2--4]{farinhaluz2013surplus} derives a related
revenue bound from a relaxed problem in which bidders can misreport
their values but not their beliefs.}
The firm can then choose the best cutoff separately for each forecast.
We call the resulting rule
the \emph{pointwise cutoff}. 

Theorem~\ref{thm:pointwise-cutoff-optimality} shows 
that when the pointwise benchmark is concave, this 
benchmark remains profit-maximizing even when the 
forecast is private. However, when the benchmark is 
not concave, the firm may have to leave the manager 
extra rent to implement the pointwise benchmark. It 
can sometimes do better by abandoning the benchmark. 
The firm chooses a cutoff by weighing what it saves 
on forecast rents
against the profit it gives up by changing which 
expansions go ahead
(see \eqref{eq:rent-adjusted-concavification}).
Theorem~\ref{thm:constructive-cutoff-selection} constructs an optimal
cutoff as the limit of a sequence of cutoffs obtained by solving
finite problems.

When the manager has little at stake, one might expect the firm to
lose comparably little to forecast rents and distorted approvals.
If the firm knows her stake, it can offset her private gain or loss
from expansion through her pay. She then gains nothing by misstating
demand, and the firm can use her truthful forecast without paying
forecast rent or distorting approvals. In the capacity model, starting
with both dimensions private, observing either the stake or the forecast
raises optimal expected profit by order \(h^{2/3}\) as \(h\) tends to zero,
so the gain vanishes more slowly than the width \(h\) of the stake range
(Proposition~\ref{prop:distribution-robust-loss} and Appendix
Corollary~\ref{cor:alternative-public-benchmark}).
This \(h^{2/3}\) rate holds under exact truthful reporting --- even when the
firm can randomize approvals --- over a class of joint
distributions that allows stakes and forecasts to be dependent.\footnote{The class of joint distributions requires that the joint
density of \((p,\theta/h)\) must be bounded above and away from zero
by constants independent of \(h\)
(Proposition~\ref{prop:distribution-robust-loss}).}

\paragraph{Related literature.}
We contribute to the literature on interested experts. This literature
studies how evidence, advocacy, competition, cheap talk, and compensation
can make their advice useful
\citep{milgrom1986relying,dewatripont1999advocates,
wolinsky1993competition,pesendorfer2003second,crawford1982strategic,inderst2012pay}.

We use mechanism design to study how the cost of truthful reporting
shapes a firm's preferred approval rule. Classical research in mechanism
design studies the elicitation of preferences
\citep{myerson1981optimal,riley1983optimal} or forecasts
\citep{winkler1967quantification,gneiting2007strictly,
karni2009mechanism}---in each case assuming away the other dimension. 

We study the joint elicitation of both dimensions. A related line of work allows both dimensions to vary, but treats one as known
to the designer. \citet[Propositions 2--3]{farinhaluz2013surplus}
screens private values in a relaxed problem with known beliefs,
while \citet[Theorem~3]{boutilier2012eliciting} characterizes truthful
forecast compensation when the expert's decision preferences are known.

Other work allows both dimensions to vary privately, but studies
mechanisms that solicit only one. \citet{bergemann2005robust} study
when preference reports suffice despite private beliefs, while
\citet*{wuthrich2026belief} study forecast elicitation with private,
unreported stakes in the decision. \citet{bose2007optimal} study a related elicitation problem under
a different restriction: beliefs uniquely identify the agent's
preference type.

Our work studies the joint elicitation of both dimensions, allowing both to vary separately. \citet[Propositions~3--4]{farinhaluz2013surplus} characterizes
a revenue bound with both dimensions private and gives a
distributional condition under which it is attained. \citet*{gerardi2009aggregation} jointly elicit information and preferences
from multiple experts without transfers or ex post verification. 
\citet*{cai2013designing} and \citet*{bergemann2025data} study joint
elicitation for efficient decisions. 

Our firm's objective is profit after screening rents. To the best of our knowledge, we are the first to characterize
the profit-maximizing approval cutoff when an expert privately
knows both her stake in the decision and the probability that
the project succeeds. We show how the shape of the approval cutoff 
determines the
additional cost of eliciting forecasts and when this cost makes
changing expansion decisions profitable. We also show that the profit cost of private forecasts can vanish more slowly than the range of managerial stakes. 

Section~\ref{sec:setup} introduces the capacity expansion contract.
Section~\ref{sec:screening} separates the two sources of rent.
Section~\ref{sec:single-threshold} characterizes the implementability
of cutoff rules.
Section~\ref{sec:threshold-curvature} derives the minimum forecast
rent for a fixed cutoff.
We then turn to the optimal cutoff choice. Section~\ref{sec:pointwise-benchmark} determines when the public-forecast
rule remains optimal. Section~\ref{sec:optimal-threshold} chooses the
adjusted rule. Section~\ref{sec:private-interest-family} develops the
small-stakes comparison. The online
appendix begins with a \hyperref[app:reading-guide]{guide} to the proofs.

\section{Setup}
\label{sec:setup}

A principal chooses whether to take an action
\(a\in\mathcal A=\{0,1\}\). The state is \(\omega\in\Omega=\{0,1\}\).
Before transfers, the principal earns \(v(\omega)\) from taking the
action and zero otherwise, with \(v(1)>0>v(0)\). The single agent's payoff is
\(\theta a\), where \(\theta\in\Theta=[\underline\theta,\bar\theta]
\subseteq\mathbb R\) is her private stake in the action. A positive
stake is a private benefit and a negative stake is a private cost.
The agent's preferences are quasilinear. A positive transfer is paid by
the agent to the principal.

Let \(\mathcal I\) denote the agent's private information, including
her stake. Her forecast is the posterior probability
\[
p:=\Pr(\omega=1\mid\mathcal I).
\]
Her payoff-relevant type is \((\theta,p)\in\Theta\times\mathcal P\),
where the forecast ranges over the nondegenerate open interval
\[
\mathcal P=(\underline p,\bar p),
\qquad
0\le \underline p<\bar p\le 1.
\]
Write \(\overline{\mathcal P}=[\underline p,\bar p]\) for its closure.\footnote{Closed or half-closed belief domains can be accommodated by additionally requiring the convex functions used in the implementability results to admit finite one-sided supporting slopes at every included boundary point. We use an open belief domain to avoid these inessential endpoint qualifications.}
The forecast already incorporates anything the agent's stake tells her
about the state. Iterated expectations give\footnote{Because \(\omega\) is
binary, this posterior identity is equivalent to calibration together with
\(\theta\perp\!\!\!\perp\omega\mid p\). Conditional probabilities are
understood as a version satisfying the displayed identity throughout the
modeled type domain.}
\begin{align}
\Pr(\omega=1\mid\theta,p)=p.
\label{eq:posterior-consistency}
\end{align}
Thus a type \((\theta,p)\) evaluates the two state payments using
probabilities \(p\) and \(1-p\). The stake and forecast can be dependent.
We write \(F(\cdot\mid p)\) for the conditional distribution of the stake
given the forecast.

\begin{assump}[State contractibility]
\label{assump:state-contractibility}
The realized state \(\omega\) is publicly observed (or contractible) after the decision, regardless of which action is chosen, so transfers may depend on it.
\end{assump}

\paragraph{Timing.} The principal commits to a decision and payment
rule. Nature draws \((\omega,\theta,p)\), and the agent privately observes
\((\theta,p)\). She can decline for a payoff of zero or submit a report.
The principal then chooses the action according to the announced rule.
The state is observed and transfers are made.

\vspace{3mm}
\begin{mdframed}[style=runningexamplebox]
\vspace{-5mm}
\subsection{Running Example}
\label{subsec:broad-stakes-capacity}
\phantomsection\label{ex:capacity-contract}
\paragraph{Setup.}To illustrate our model, we include a numeric example. Suppose that a manufacturing firm earns one unit from expanding
when demand is high ($v(1)=1$) and loses one when demand is low ($v(0)=-1$). Without expansion
its incremental payoff is zero. High demand means that a verifiable market index exceeds a
preset level. The index is recorded regardless of whether the firm
expands. 

The manager's stake is the value to her of the expansion. Let her forecast be uniform on \((0,1)\) and her stake be
independently uniform on \([-1,1]\), a stake range of \(h=2\). The
following contract illustrates how the firm can use the forecast.

\paragraph{A truthful contract.} Ask the manager to report both her forecast and her stake. Expand when their
sum is at least one:
\[
x^0(\hat\theta,\hat p)=\mathbf 1\{\hat\theta+\hat p\ge1\}.
\]
The manager pays the firm one unit if expansion is followed by low demand,
and zero otherwise.

Given her true stake and forecast, the manager's expected payoff from
expansion is
\[
\theta-(1-p)=\theta+p-1.
\]
Without expansion, her payoff is zero. Truthful reports therefore induce
expansion exactly when she weakly prefers it. Changing either or both
reports cannot improve her payoff. Her truthful payoff is nonnegative, so
she is also willing to participate.

Reports leading to the same decision give the same payoff. Thus
truthful joint reporting is optimal even when several reports select the
same option.

The expected price of expansion is \(1-p\). A more optimistic manager
faces a lower expected price, and two managers with the same forecast
can make different choices if their stakes differ. The lowest-stake
manager earns zero at every forecast. The firm attains its public-forecast
profit of \(1/6\) with zero implementation rent.
\end{mdframed}

We later return to this example in Section~\ref{sec:pointwise-benchmark} 
to establish optimality. Appendix~\ref{app:broad-stakes-contract} depicts the decision regions and state payments. Section~\ref{early-capacity-expansion} later returns to
the same project with smaller stakes, where the firm changes its rule
and leaves positive rent.
\subsection{Limits of Belief-Only Elicitation}
\label{subsec:belief-only-impossibility}

\Needspace{8\baselineskip}
\begin{definition}[Belief-only mechanism]
\label{def:belief-only-mechanism}
A \emph{belief-only mechanism} is a pair $(x^B,t^B)$, where
$x^B:\mathcal P\to[0,1]$ gives the probability of choosing $a=1$ after each
belief report and $t^B:\mathcal P\times\Omega\to\mathbb R$ gives the
corresponding state-contingent transfer.
\end{definition}
A scoring rule \citep{winkler1967quantification} can be embedded in a
belief-only mechanism through its state-contingent transfer component.

A belief-only mechanism is incentive compatible if, for every
$p,\hat p\in\mathcal P$ and
every $\theta\in\operatorname{supp}F(\cdot\mid p)$,
\begin{align}
\theta x^B(p)
-\bigl[p\,t^B(p,1)+(1-p)t^B(p,0)\bigr]
\geq
\theta x^B(\hat p)
-\bigl[p\,t^B(\hat p,1)+(1-p)t^B(\hat p,0)\bigr].
\label{eq:belief-only-ic}
\end{align}

Because a belief-only mechanism does not solicit a preference report,
this is its full incentive-compatibility requirement.

\Needspace{10\baselineskip}
\begin{theorem}[Belief-only IC forces decision blindness]
\label{thm:belief-only-impossibility}
Suppose there exist two distinct preference types
\(\theta_L,\theta_H\in\Theta\) such that
\[
\theta_L,\theta_H
\in
\operatorname{supp}F(\cdot\mid p)
\qquad
\text{for every }p\in\mathcal P.
\]
Then every incentive-compatible belief-only mechanism has an action rule
\(x^B\) that is constant on \(\mathcal P\).
\end{theorem}

At a given forecast report, both stakes face the same difference between
the two state payments. This difference pins down the slope of each
stake's truthful utility as a function of the forecast, wherever the
convex utilities are differentiable. The utility curves therefore differ
by a constant. At the same forecast, the two types make the same payments,
so their utility gap is the stake difference times the probability of
taking the action. The action probability must therefore be constant.

\begin{appendixproof}[Proof of Theorem \ref{thm:belief-only-impossibility}]
\phantomsection\label{app:proof-belief-only}
For \(i\in\{L,H\}\), define the state-contingent transfer difference
\[
\beta^B(q):=t^B(q,0)-t^B(q,1)
\]
and the truthful utility of type \((\theta_i,p)\) by
\[
U_i(p)
:=
\theta_i x^B(p)-t^B(p,0)+p\beta^B(p).
\]
Because \(\theta_i\) belongs to the conditional support at every belief,
incentive compatibility gives
\begin{align}
U_i(p)
&=
\sup_{q\in\mathcal P}
\left\{
\theta_i x^B(q)-t^B(q,0)+p\beta^B(q)
\right\}.
\label{eq:belief-only-convex-envelope}
\end{align}
Thus \(U_i\) is a finite convex function on \(\mathcal P\).

For any \(p,z\in\mathcal P\), evaluating the supremum in
\eqref{eq:belief-only-convex-envelope} at report \(q=p\) yields
\[
U_i(z)
\ge
U_i(p)+\beta^B(p)(z-p).
\]
Consequently,
\[
\beta^B(p)
\in
\partial U_L(p)\cap\partial U_H(p)
\qquad
\text{for every }p\in\mathcal P.
\]
Moreover, the transfer terms cancel between the two truthful utilities:
\begin{align}
U_H(p)-U_L(p)
&=
(\theta_H-\theta_L)x^B(p).
\label{eq:belief-only-utility-difference}
\end{align}

Finite convex functions on an open interval are locally Lipschitz and hence
locally absolutely continuous. Equation
\eqref{eq:belief-only-utility-difference} therefore makes \(x^B\) locally
absolutely continuous. The two convex utilities are differentiable on a common
set \(D\subseteq\mathcal P\) of full Lebesgue measure. At every \(p\in D\),
each subdifferential is the singleton containing the corresponding derivative,
so the common-subgradient relation gives
\[
U_L'(p)=\beta^B(p)=U_H'(p).
\]
Differentiating \eqref{eq:belief-only-utility-difference} on \(D\) therefore
gives \(x^{B\prime}(p)=0\).

Finally, fix any \(p_0<p_1\) in \(\mathcal P\). Local absolute continuity and
the fundamental theorem for absolutely continuous functions imply
\[
x^B(p_1)-x^B(p_0)
=
\int_{p_0}^{p_1}x^{B\prime}(p)\,dp
=0.
\]
Thus \(x^B\) is constant on \(\mathcal P\).
\end{appendixproof}

\begin{toappendix}
\subsection{The Running Example Capacity Contract}
\label{app:broad-stakes-contract}

Figure~\ref{fig:threshold-mechanism} shows the decision regions and
payments from the running example in Section \ref{subsec:broad-stakes-capacity}. 

\begin{figure}[htbp]
\centering
\caption{A truthful contract for capacity expansion}
\label{fig:threshold-mechanism}
\begin{minipage}[t]{0.48\textwidth}
\vspace{0pt}
\centering
{\small\textbf{Panel A: Expansion decision}\par}
\smallskip
\begin{tikzpicture}[x=5.2cm,y=2cm,every node/.style={font=\small}]
  \fill[blue!9] (0,1)--(1,1)--(1,0)--cycle;
  \draw[->,black!75] (0,0)--(1.08,0);
  \draw[->,black!75] (0,-1)--(0,1.15);
  \draw[very thick,blue!70!black] (0,1)--(1,0)
    node[midway,below,sloped,inner sep=3pt] {$1-\hat p$};
  \node at (0.76,0.77) {Expand};
  \node at (0.47,-0.48) {Do not expand};
  \node[left] at (0,0) {$0$};
  \node[below] at (1,0) {$1$};
  \node[left] at (0,1) {$1$};
  \node[left] at (0,-1) {$-1$};
  \node[below] at (0.50,-1.10) {reported belief $\hat p$};
  \node[rotate=90] at (-0.20,0.05) {reported stake $\hat\theta$};
\end{tikzpicture}
\end{minipage}\hfill
\begin{minipage}[t]{0.48\textwidth}
\vspace{0pt}
\centering
{\small\textbf{Panel B: State-contingent transfers}\par}
\smallskip
\small
\begin{tabular}{@{}lr@{}}
\toprule
Decision and outcome & Manager pays firm\\
\midrule
Expand; low demand & $1$\\
Expand; high demand & $0$\\
Do not expand & $0$\\
\bottomrule
\end{tabular}
\end{minipage}
\par\medskip
\begin{minipage}{\textwidth}
\small\begin{singlespace}
\textbf{Notes:} Reports within either decision region induce the same decision and payments. The manager pays one unit only when expansion is followed by low demand.  
\end{singlespace}
\end{minipage}
\end{figure}

\FloatBarrier
\end{toappendix}

Theorem~\ref{thm:belief-only-impossibility} separates rewarding forecast
accuracy from using a forecast to choose an action. A scoring rule that
is proper when reports affect only the score may fail when the expert
also cares about the action \citep{winkler1968good,wolfers2004prediction}.
When her action preference is known, compensation can offset that motive
\citep{boutilier2012eliciting}.

With private preferences, \citet{morgan2003analysis} obtain informative
categorical recommendations in a cheap-talk model, and
\citet{li2008mandatory} show how disclosure of an expert's bias can
affect what she credibly communicates. Our theorem concerns a mechanism
that requires truthful reporting of the full posterior.
\citet[Theorem~1]{wuthrich2026belief} bound worst-case departures from
a target decision under a finite payment budget. In our setting,
exact truthfulness forces a belief-only rule to ignore the forecast
without a uniform transfer bound and with two arbitrarily close stakes
that can occur at every belief. In the capacity expansion example,
every managerial stake can occur at every belief.

\subsection{Direct Mechanisms}
\label{subsec:direct-mechanisms}

We now move beyond belief-only mechanisms --- allowing the principal to condition decisions and transfers on reports of \emph{both} dimensions. Let \(x:\Theta\times\mathcal P\to[0,1]\) be the allocation rule, so
that \(x(\hat\theta,\hat p)\) is the probability that the principal chooses
\(a=1\) after reports \((\hat\theta,\hat p)\). Let
\(t:\Theta\times\mathcal P\times\Omega\to\mathbb{R}\) be the transfer rule, so
that \(t(\hat\theta,\hat p,\omega)\) is the transfer in state
\(\omega\).\footnote{With quasilinear preferences, allowing transfers also to
depend on the realized action would not change the analysis:
\(t(\hat\theta,\hat p,\omega)\) can be interpreted as the expected transfer
conditional on the report and state, averaging over the mechanism's action
randomization. Throughout, mechanism rules are measurable.}

For a true type \((\theta,p)\) reporting \((\hat\theta,\hat p)\), write
\begin{align*}
T(\hat\theta,\hat p;p)
&=p\,t(\hat\theta,\hat p,1)+(1-p)t(\hat\theta,\hat p,0),\\
u(\hat\theta,\hat p;\theta,p)
&=\theta x(\hat\theta,\hat p)-T(\hat\theta,\hat p;p).
\end{align*}
Under truthful reporting, write \(T(\theta,p):=T(\theta,p;p)\).
Truthful utility is then
\[
U(\theta,p):=\theta x(\theta,p)-T(\theta,p).
\]

\subsection{The Principal's Problem}

For the principal's optimization problem, we assume that the joint distribution of \((\omega,\theta,p)\) is common knowledge. In particular, the principal knows the marginal distribution \(G\) of beliefs and the conditional distributions \(F(\cdot\mid p)\). Let \(H\) denote the resulting joint distribution of private types, so that
\[
dH(\theta,p)=dF(\theta\mid p)\,dG(p).
\]

Define the principal's expected gross payoff from taking the action at belief
\(p\) by
\begin{align}
V(p):=p\,v(1)+(1-p)v(0).
\label{eq:principal-action-value}
\end{align}
By \eqref{eq:posterior-consistency}, the principal's expected payoff
conditional on \((\theta,p)\), under truthful reporting, is
\[
V(p)x(\theta,p)+T(\theta,p).
\]
The principal evaluates a mechanism by integrating this interim payoff over
\(H\), restricting attention to mechanisms with integrable truthful-path
interim transfers.\footnote{Formally, we require
\(\int_{\Theta\times\mathcal P}|T(\theta,p)|\,dH(\theta,p)<\infty\).
Since \(\theta x\) and \(Vx\) are bounded, this also makes truthful utility and
the principal's interim payoff integrable; the pointwise incentive-compatibility
results do not require it. We do not impose the stronger condition
\(\mathbb E[|t(\theta,p,\omega)|]<\infty\). We use this interim
integrability criterion throughout. It does not by itself imply absolute
integrability of realized state-contingent transfers, a stronger requirement
most conveniently expressed using the transfer tilt introduced in Appendix Definition
\ref{def:state-contingent-tilt}.}
\begin{assump}[Conditional full support]
\label{assump:conditional-full-support}
The conditional support of the preference type is the common interval
\(\Theta\):
\[
\operatorname{supp}F(\cdot\mid p)=\Theta
\qquad
\forall p\in\mathcal P.
\]
\end{assump}
In particular, this assumption implies the common-type condition in Theorem
\ref{thm:belief-only-impossibility}.
Accordingly, the BIC and IR constraints below are imposed on the full product
domain \(\Theta\times\mathcal P\).\footnote{This is deliberate: beliefs assigned
zero probability by \(G\) remain modeled types and admissible reports. Under a
support-only formulation, behavior inside support gaps would instead be part
of the choice of an incentive-compatible extension.}
We can now state the principal's constrained problem.
\begin{definition}[Global direct-mechanism problem]
\label{def:globalproblem}
The principal's global problem is
\begin{align}
\sup_{x,t}\quad
& \int_{\Theta\times\mathcal P}
\left[
V(p)x(\theta,p)+T(\theta,p)
\right]
\,dH(\theta,p).
\tag{GP}\label{eq:globalproblem}
\\
\text{s.t.}\quad
& u(\theta,p;\theta,p)
\ge
u(\hat\theta,\hat p;\theta,p)
&&
\forall(\theta,p),(\hat\theta,\hat p)\in\Theta\times\mathcal P,
\tag{BIC}\label{eq:fullproblem-bic}
\\
& u(\theta,p;\theta,p)\ge 0
&&
\forall(\theta,p)\in\Theta\times\mathcal P,
\tag{IR}\label{eq:fullproblem-ir}
\\
& x(\theta,p)\in[0,1]
&&
\forall(\theta,p)\in\Theta\times\mathcal P.
\tag{F}\label{eq:fullproblem-feas}
\end{align}
\end{definition}

Constraint \hyperref[eq:fullproblem-bic]{\textup{(BIC)}}, evaluated using the
interim utility defined above, is global: truthful
reporting must be optimal, though not necessarily uniquely optimal, against
every joint misreport
\((\hat\theta,\hat p)\). By contrast, the earlier belief-only IC condition in
\eqref{eq:belief-only-ic} restricts only the belief report. Constraint
\hyperref[eq:fullproblem-ir]{\textup{(IR)}} imposes truthful-path interim
individual rationality, and \hyperref[eq:fullproblem-feas]{\textup{(F)}}
imposes feasibility.

The menu in the \hyperref[ex:capacity-contract]{running example} shows how the
firm can make expansion depend on a private forecast. We now derive the
cost of doing so.

\Needspace{8\baselineskip}
\section{Screening Beliefs and Stakes}
\label{sec:screening}

In the capacity contract, the same expansion offer appeals to managers with
different combinations of stakes and beliefs. We first hold the belief
fixed to price stake screening, then ask how state-contingent payments can
make both reports truthful.
Unless stated otherwise, Assumptions \ref{assump:state-contractibility}
and \ref{assump:conditional-full-support} remain in force below.

\subsection{Stake Reporting}
\label{sec:preference-margin}

First hold the belief report fixed. Truthful stake reporting requires that
a higher reported stake cannot make the action less likely. The familiar
envelope formula \citep{myerson1981optimal} then gives
\[
U(\theta,p)=R(p)+\int_{\underline\theta}^{\theta}x(s,p)\,ds,
\]
where the \emph{utility floor} \(R(p)=U(\underline\theta,p)\) is the
utility of the lowest-stake type at belief \(p\). Since the allocation is
nonnegative, truthful utility rises weakly with the stake. Participation
therefore requires \(R(p)\ge0\).
Appendix Lemma \ref{lem:withinbeliefenv} gives the precise equivalence.

The envelope also determines expected payments:
\begin{align}
T(\theta,p)
&=\theta x(\theta,p)-R(p)
-\int_{\underline\theta}^{\theta}x(s,p)\,ds.
\label{eq:transfer-envelope}
\end{align}
If beliefs were public, the principal could set the floor to zero at each
belief. Private belief reports may require positive floors to prevent
joint misreports.\footnote{Related sequential-screening analyses condition
an envelope on another private type and then impose incentives across
types \citep{courty2000sequential,bergemann2020scope}. Conditional virtual
values also appear in other mechanism-design settings; for example,
\citet{roughgarden2016optimal} condition on other bidders' signals.}

To evaluate the principal's payoff, we impose the following regularity
condition on the conditional stake distributions.
\begin{assump}[Conditional density regularity]
\label{assump:conditional-density}
For every \(p\in\mathcal P\), \(F(\cdot\mid p)\) is absolutely continuous and
admits a density \(f(\cdot\mid p)\) that is strictly positive on \(\Theta\).
\end{assump}

Under this assumption, the conditional virtual value accounts for ordinary
stake-screening rents:
\begin{align}
\phi(\theta,p)&:=\theta-\frac{1-F(\theta\mid p)}{f(\theta\mid p)},
\label{eq:conditional-virtual-value}\\
W(\theta,p)&=V(p)+\phi(\theta,p).
\nonumber
\end{align}
Thus \(W\) is the principal's expected gross payoff from the action plus
the expert's virtual value. The principal's expected payoff becomes
\[
\Pi(x,R)=\int_{\Theta\times\mathcal P}W(\theta,p)x(\theta,p)\,dH(\theta,p)
-\int_{\mathcal P}R(p)\,dG(p).
\]
Appendix Proposition \ref{prop:prefvirtualsurplus} derives this identity.
The density assumption is used for this payoff calculation; the reporting
conditions themselves do not require it.

For the running example \(h=2\) contract, uniform managerial stakes give
\(W(\theta,p)=2(\theta+p-1)\). Maximizing virtual surplus therefore delivers
the expansion rule \(x^0\) in the \hyperref[ex:capacity-contract]{running example}.
Section~\ref{sec:pointwise-benchmark} establishes the contract's global
optimality.

\subsection{Belief Reporting}
\label{sec:belief-margin}

Stake-reporting incentives determine expected payments up to a utility
floor. We next ask whether those expected payments 
can be divided across states to make belief reports truthful as well.

At a fixed report, define the transfer tilt
\(\beta(\theta,p)=t(\theta,p,0)-t(\theta,p,1)\). A positive tilt means
that the expert pays more in state 0 than in state 1. Holding the truthful
expected payment \(T(\theta,p)\) fixed, a type with true belief \(q\)
evaluates this report at
\begin{align}
T(\theta,p;q)=T(\theta,p)-\beta(\theta,p)(q-p).
\label{eq:offpath-transfer-beta}
\end{align}
In the capacity example, the expansion offer charges one after low demand
and zero after high demand. At a forecast of \(1/2\), it costs the same in
expectation as a charge of \(0.5\) in either state. A manager who puts
probability \(0.6\) on high demand expects to pay \(0.4\) under the
state-contingent offer and \(0.5\) under the flat charge.

For any fixed report, expected utility is affine in the expert's true stake
and belief. Its slope in the stake is the probability of action \(x\), and
its slope in the belief is the transfer tilt \(\beta\). The expert can
choose any report, so the highest utility she can obtain is the upper
envelope of these affine payoffs. Truthfulness requires each report's payoff
to touch that envelope at the type it is intended for. This is the
supporting-plane condition of \citet{rochet1987necessary}. Together with
participation, it requires a nonnegative convex utility function with
\[
\bigl(x(\theta,p),\beta(\theta,p)\bigr)\in\partial U(\theta,p).
\]
The supporting slopes allow for kinks in utility. Both slopes come from
the same utility function because the expert can change both reports at once.
As in a proper scoring rule \citep{winkler1967quantification,gneiting2007strictly},
state-contingent pay lets the principal change how different beliefs value a
report. Appendix Proposition~\ref{prop:convex-utility-characterization}
gives the equivalence and transfer formulas.\footnote{The maintained contract
class requires integrable truthful-path interim transfers. The recovered
state-contingent transfers are absolutely integrable on the truthful path
if and only if
\(\int p(1-p)|\beta(\theta,p)|\,dH(\theta,p)<\infty\).
The running-example contracts satisfy this stronger condition. Neither
condition imposes limited liability or a uniform bound on transfer
exposure; such restrictions could change implementation cost.}

For the running example, the tilt is one when
expansion occurs and zero otherwise, so \(\beta^0=x^0\). Its truthful utility is
\(U^0(\theta,p)=\max\{0,\theta+p-1\}\), a nonnegative convex function with
the required supporting slopes. This incentive argument holds for any
distribution of managerial stakes and beliefs on the modeled domain.

\subsection{Implementation Cost}
\label{sec:global-reduced-problem}

For a fixed allocation, virtual surplus is fixed. The remaining choice is
how cheaply to satisfy reporting incentives. Write
\[
\operatorname{VS}(x)=
\int_{\Theta\times\mathcal P}W(\theta,p)x(\theta,p)\,dH(\theta,p).
\]
Let \(\operatorname{Rent}(x)\) be the infimum of expected lowest-stake
utility over truthful, individually rational implementations of \(x\)
in the maintained contract class. Set it to \(+\infty\) if there is
no such implementation. The principal's
problem is therefore
\begin{align}
\sup_x\left\{\operatorname{VS}(x)-\operatorname{Rent}(x)\right\}.
\label{eq:global-benefit-cost}
\end{align}
Virtual surplus already deducts ordinary stake-screening rents.
Implementation rent measures the additional cost of truthful belief
reporting. Lowering the utility floor at one belief can make a different
report more attractive, so the firm must choose the entire floor schedule
jointly.
Appendix Theorem \ref{thm:global-reduction} establishes the exact
reformulation, and Appendix Definition \ref{def:rent-cost} gives the formal
rent definition.

The running example leaves the lowest-stake manager zero
utility at every belief, so its implementation rent is 
zero. Managers
whose stakes exceed the expansion cutoff still receive positive utility.
In the \(h=2\) example, the firm's expected payoff is
\(\tfrac12\mathbb E[p^2]=1/6\).

We now fix a threshold rule and calculate the lowest expected floor needed
to implement it.

\begin{toappendix}
\subsection{Payoff Notation and Stake Reporting}

\begin{definition}[Interim transfers and utility]
\label{def:interim-payoffs}
For an agent whose true belief is $p$ and who reports $(\hat\theta,\hat p)$, define
\[
T(\hat\theta,\hat p;p)
:=
p\,t(\hat\theta,\hat p,1)
+
(1-p)\,t(\hat\theta,\hat p,0).
\]
On the truthful path, write $T(\theta,p)
:= T(\theta,p;p)$. The interim expected utility of a type $(\theta,p)$ who
reports $(\hat\theta,\hat p)$ is then\[
u(\hat\theta,\hat p;\theta,p)
:=
\theta x(\hat\theta,\hat p)
-
T(\hat\theta,\hat p;p).
\]
On the truthful path, write $U(\theta,p)
:=
u(\theta,p;\theta,p)
=
\theta x(\theta,p)-T(\theta,p)$.
\end{definition}

\begin{definition}[Same-belief IC]
\label{def:samebeliefic}
A pair \((x,T)\) satisfies same-belief incentive compatibility if, for every
belief \(p\in\mathcal P\) and every pair \(\theta,\hat\theta\in\Theta\),
\begin{align}
U(\theta,p)
\ge
U(\hat\theta,p)
+
(\theta-\hat\theta)x(\hat\theta,p).
\label{eq:samebeliefic}
\end{align}
\end{definition}

\begin{restatable}[Same-belief envelope]{lemma}{withinbeliefenv}
\label{lem:withinbeliefenv}
A pair \((x,T)\) satisfies same-belief IC if and only if, for every
\(p\in\mathcal P\), \(x(\cdot,p)\) is nondecreasing and
\begin{align}
U(\theta,p)
=
R(p)
+
\int_{\underline\theta}^{\theta}x(s,p)\,ds,
\qquad \forall\theta\in\Theta,
\label{eq:withinbelief-envelope}
\end{align}
where \(R(p):=U(\underline\theta,p)\).
\end{restatable}

\begin{proof}[Proof of Lemma \ref{lem:withinbeliefenv}]
Fix \(p\in\mathcal P\), and write \(x(\theta):=x(\theta,p)\) and
\(U(\theta):=U(\theta,p)\). Suppose first that
same-belief IC holds. For any \(\theta>\hat\theta\), applying
\eqref{eq:samebeliefic} in both directions gives
\begin{align}
(\theta-\hat\theta)x(\hat\theta)
\le
U(\theta)-U(\hat\theta)
\le
(\theta-\hat\theta)x(\theta).
\label{eq:withinbelief-two-sided-bound}
\end{align}
The two bounds can be consistent only if
\(x(\hat\theta)\le x(\theta)\). Hence \(x\) is nondecreasing.

It remains to derive the envelope formula without imposing differentiability.
Fix \(a<b\) in \(\Theta\), and let
\(\pi=\{a=\theta_0<\theta_1<\cdots<\theta_n=b\}\) be a partition of
\([a,b]\). Applying \eqref{eq:withinbelief-two-sided-bound} on each
subinterval and summing yields
\begin{align}
\sum_{i=1}^n(\theta_i-\theta_{i-1})x(\theta_{i-1})
\le
U(b)-U(a)
\le
\sum_{i=1}^n(\theta_i-\theta_{i-1})x(\theta_i).
\label{eq:withinbelief-riemann-bounds}
\end{align}
The difference between the upper and lower sums is bounded by
\[
\|\pi\|
\sum_{i=1}^n\bigl[x(\theta_i)-x(\theta_{i-1})\bigr]
=
\|\pi\|\,[x(b)-x(a)],
\]
where \(\|\pi\|:=\max_i(\theta_i-\theta_{i-1})\). Because \(x\) is bounded
and nondecreasing, it is Riemann integrable, and its Riemann and Lebesgue
integrals coincide. Letting the mesh of the partition converge to zero in
\eqref{eq:withinbelief-riemann-bounds} therefore gives
\[
U(b)-U(a)=\int_a^b x(s)\,ds.
\]
For \(\theta>\underline\theta\), take \(a=\underline\theta\) and
\(b=\theta\); at \(\theta=\underline\theta\), the identity follows directly
from the definition of \(R(p)\). Restoring the belief argument proves
\eqref{eq:withinbelief-envelope}.

Conversely, suppose that \(x(\cdot,p)\) is nondecreasing and
\eqref{eq:withinbelief-envelope} holds for every \(p\). Fix
\(\theta,\hat\theta\in\Theta\). If \(\theta\ge\hat\theta\), monotonicity gives
\[
U(\theta,p)-U(\hat\theta,p)
=
\int_{\hat\theta}^{\theta}x(s,p)\,ds
\ge
(\theta-\hat\theta)x(\hat\theta,p).
\]
If \(\theta<\hat\theta\), monotonicity instead gives
\[
U(\theta,p)-U(\hat\theta,p)
=
-\int_{\theta}^{\hat\theta}x(s,p)\,ds
\ge
-(\hat\theta-\theta)x(\hat\theta,p)
=
(\theta-\hat\theta)x(\hat\theta,p).
\]
Thus \eqref{eq:samebeliefic} holds in both cases, which proves same-belief IC.
\end{proof}

\Needspace{12\baselineskip}
\begin{restatable}[Preference-margin virtual-surplus decomposition]{prop}{prefvirtualsurplus}
\label{prop:prefvirtualsurplus}
Under Assumption \ref{assump:conditional-density}, the principal's expected
payoff from any same-belief IC mechanism \((x,T)\) whose truthful-path interim
transfer is \(H\)-integrable is
\begin{align}
\Pi(x,R)
&=
\int_{\mathcal P}\int_{\Theta}
W(\theta,p)x(\theta,p)f(\theta\mid p)\,d\theta\,dG(p)
-
\int_{\mathcal P}R(p)\,dG(p),
\label{eq:payoff-virtual-form}
\\
\text{where}\qquad W(\theta,p)
&:=
V(p)
+
\theta
-
\frac{1-F(\theta\mid p)}{f(\theta\mid p)}.
\label{eq:belief-adjusted-virtual-surplus}
\end{align}
\end{restatable}

\begin{proof}[Proof of Proposition \ref{prop:prefvirtualsurplus}]
By Lemma \ref{lem:withinbeliefenv},
\(T(\theta,p)=S_x(\theta,p)-R(p)\), where
\[
S_x(\theta,p)
:=
\theta x(\theta,p)
-
\int_{\underline\theta}^{\theta}x(s,p)\,ds
\]
is uniformly bounded. Because \(R\) depends only on \(p\), the assumed
\(H\)-integrability of \(T\) implies that \(R\) is \(G\)-integrable. Thus all
of the integrals below are finite. Under truthful reporting, the principal's
expected payoff is
\begin{align*}
\Pi(x,R)
&=
\int_{\mathcal P}\int_{\Theta}
\left[V(p)x(\theta,p)+T(\theta,p)\right]
\,dF(\theta\mid p)\,dG(p)\\
&=
\int_{\mathcal P}\int_{\Theta}
\left[
(V(p)+\theta)x(\theta,p)
-R(p)
-\int_{\underline\theta}^{\theta}x(s,p)\,ds
\right]
\,dF(\theta\mid p)\,dG(p),
\end{align*}
where the second equality substitutes the on-path transfer formula implied by
Lemma \ref{lem:withinbeliefenv}. For each \(p\), Fubini's theorem and Assumption
\ref{assump:conditional-density} give
\begin{align*}
\int_{\Theta}
\left[\int_{\underline\theta}^{\theta}x(s,p)\,ds\right]
dF(\theta\mid p)
&=
\int_{\Theta}x(s,p)\bigl[1-F(s\mid p)\bigr]\,ds\\
&=
\int_{\Theta}
\frac{1-F(\theta\mid p)}{f(\theta\mid p)}
x(\theta,p)f(\theta\mid p)\,d\theta.
\end{align*}
Substituting this identity into the preceding payoff expression and collecting
terms yields \eqref{eq:payoff-virtual-form}, with \(W\) defined by
\eqref{eq:belief-adjusted-virtual-surplus}.
\end{proof}

\Needspace{12\baselineskip}
\subsection{State Payments and Joint Reporting}

\begin{definition}[State-contingent transfer tilt]
\label{def:state-contingent-tilt}
Fix an on-path expected transfer schedule \(T\). For any transfer rule delivering
\(T\), define its \emph{state-contingent tilt} by
\(\beta(\theta,p):=t(\theta,p,0)-t(\theta,p,1)\), the state-0 payment
minus the state-1 payment. There is a one-to-one
correspondence between such transfer rules and tilt functions, with recovery
formulas
\begin{align}
t(\theta,p,1)
&=
T(\theta,p)-(1-p)\beta(\theta,p),
\label{eq:tilt-transfer-1}
\\
t(\theta,p,0)
&=
T(\theta,p)+p\beta(\theta,p).
\label{eq:tilt-transfer-0}
\end{align}
\end{definition}

Substituting the expected-payment formula into global IC gives
\begin{align}\label{eq:global-ic-beta}
U(\theta,p)
&\ge
U(\hat\theta,\hat p)
+
(\theta-\hat\theta)x(\hat\theta,\hat p)
+
\beta(\hat\theta,\hat p)(p-\hat p),
\qquad\forall\,(\theta,p),(\hat\theta,\hat p)\in\Theta\times\mathcal P.
\end{align}
for every true type and report. This is the supporting-plane condition used in the following characterization.

\begin{restatable}[Convex characterization of global implementability]{prop}{convutilitychar}
\label{prop:convex-utility-characterization}
An allocation rule \(x:\Theta\times\mathcal P\to[0,1]\) can be implemented by
a globally BIC and truthful-path interim IR mechanism if and only if there exist
a nonnegative convex function \(U:\Theta\times\mathcal P\to\mathbb R\) and a
measurable tilt \(\beta:\Theta\times\mathcal P\to\mathbb R\) such that
\[
\bigl(x(\theta,p),\beta(\theta,p)\bigr)
\in \partial U(\theta,p)
\qquad
\forall(\theta,p)\in\Theta\times\mathcal P.
\]
Here \(\partial U\) denotes the subdifferential
relative to the type domain.

Given such \((U,\beta)\), an implementing mechanism is obtained by setting
\[
T(\theta,p)=\theta x(\theta,p)-U(\theta,p)
\]
and defining state-contingent transfers by
\eqref{eq:tilt-transfer-1}--\eqref{eq:tilt-transfer-0}.
\end{restatable}

\begin{proof}[Proof of Proposition \ref{prop:convex-utility-characterization}]
Suppose \(x\) is implemented by a globally BIC and truthful-path interim IR
mechanism, and let \(U\) and \(\beta\) denote its truthful utility and tilt.
For a reported type \(r=(\hat\theta,\hat p)\), let
\[
g(r):=\bigl(x(\hat\theta,\hat p),\beta(\hat\theta,\hat p)\bigr).
\]
Writing the true type as \(z=(\theta,p)\), inequality
\eqref{eq:global-ic-beta} is
\[
U(z)\ge U(r)+g(r)\mathbin{\cdot}(z-r)
\qquad \forall z,r\in\Theta\times\mathcal P.
\]
Hence \(g(r)\in\partial U(r)\). Moreover, equality holds when \(r=z\), so
\[
U(z)=\sup_{r\in\Theta\times\mathcal P}
\bigl\{U(r)+g(r)\mathbin{\cdot}(z-r)\bigr\}.
\]
As the supremum of affine functions, \(U\) is convex. Truthful-path interim IR
is equivalent to \(U\ge 0\).

Conversely, if \(U\) is convex and
\((x(r),\beta(r))\in\partial U(r)\), the subgradient inequality is exactly
\eqref{eq:global-ic-beta}, and hence implies global BIC. Nonnegativity implies
truthful-path interim IR. Finally, the stated transfer construction delivers
the prescribed on-path transfer and, by
\eqref{eq:offpath-transfer-beta}, the off-path expected transfers used in
\eqref{eq:global-ic-beta}.
\end{proof}

\subsection{The Reduced Problem and Implementation Rent}

For a proposed allocation and floor, write
\[
U_R(\theta,p):=R(p)+\int_{\underline\theta}^{\theta}x(s,p)\,ds.
\]

\begin{definition}[Reduced feasible set]
\label{def:reduced-feasible-set}
The \emph{reduced feasible set} \(\mathcal K\) consists of all measurable
triples
\[
x:\Theta\times\mathcal P\to[0,1],
\qquad
R:\mathcal P\to\mathbb R,
\qquad
\beta:\Theta\times\mathcal P\to\mathbb R
\]
such that \(x(\cdot,p)\) is nondecreasing for every \(p\), \(R\) is
nonnegative on \(\mathcal P\) and \(G\)-integrable, and
\begin{align}
U_R(\theta,p)
&\ge
U_R(\hat\theta,\hat p)
+(\theta-\hat\theta)x(\hat\theta,\hat p)
+\beta(\hat\theta,\hat p)(p-\hat p)
\label{eq:floor-tilt-ic}
\end{align}
for every true type \((\theta,p)\) and report \((\hat\theta,\hat p)\).
\end{definition}

\begin{restatable}[Global reduction]{theorem}{globalreduction}
\label{thm:global-reduction}
Under Assumptions \ref{assump:state-contractibility},
\ref{assump:conditional-full-support}, and
\ref{assump:conditional-density}, the value of the principal's global
direct-mechanism problem in Definition \ref{def:globalproblem} equals the value
of
\begin{align}
\sup_{(x,R,\beta)\in\mathcal K}\quad
&
\int_{\mathcal P}\int_{\Theta}
W(\theta,p)x(\theta,p)f(\theta\mid p)\,d\theta\,dG(p)
-
\int_{\mathcal P}R(p)\,dG(p).
\tag{GRP}\label{eq:global-reduced-problem}
\end{align}

This equivalence preserves allocations and payoffs: every
\((x,R,\beta)\in\mathcal K\) can be extended to a globally BIC and
truthful-path interim IR mechanism whose principal payoff equals the
corresponding objective value in \eqref{eq:global-reduced-problem}, and every
such mechanism induces a triple in \(\mathcal K\) with the same principal
payoff. Consequently, if the reduced supremum is attained, any optimizer
implements a globally optimal direct mechanism.
\end{restatable}

\begin{proof}[Proof of Theorem \ref{thm:global-reduction}]
By Lemma \ref{lem:withinbeliefenv}, any globally BIC mechanism has a
nondecreasing allocation in \(\theta\) and truthful utility \(U_R\).
Because \(x\ge 0\), truthful-path interim IR is equivalent to \(R\ge 0\).
Substituting \(U_R\) into \eqref{eq:global-ic-beta} shows that global BIC is
equivalent to \eqref{eq:floor-tilt-ic}; conversely, the transfer formulas
\eqref{eq:tilt-transfer-1}--\eqref{eq:tilt-transfer-0} implement any triple
satisfying that inequality. To verify the integrability correspondence, write
\[
A(\theta,p):=\int_{\underline\theta}^{\theta}x(s,p)\,ds.
\]
Because \(0\le x\le1\), we have
\(0\le A(\theta,p)\le\bar\theta-\underline\theta\), and the envelope formula
gives
\[
T_R(\theta,p)=\theta x(\theta,p)-A(\theta,p)-R(p).
\]
The first two terms are uniformly bounded. Hence \(R\in L^1(G)\) implies
\(T_R\in L^1(H)\), so every reduced feasible triple induces a mechanism in the
domain of the principal's problem. Conversely, any globally BIC mechanism
satisfies the same-belief envelope and hence the same transfer formula. If
\(T\in L^1(H)\), boundedness of \(\theta x-A\) implies \(R\in L^1(G)\). Proposition
\ref{prop:prefvirtualsurplus} shows that the two representations give the same
principal payoff. Taking suprema proves the equivalence of the two problems;
the final claim follows whenever the reduced supremum is attained.
\end{proof}

\begin{definition}[Virtual surplus and implementation rent]
\label{def:rent-cost}
For any measurable allocation rule \(x:\Theta\times\mathcal P\to[0,1]\), define
its virtual surplus by
\begin{align}
\operatorname{VS}(x)
:=
\int_{\mathcal P}\int_{\Theta}
W(\theta,p)x(\theta,p)f(\theta\mid p)\,d\theta\,dG(p).
\label{eq:virtual-surplus-functional}
\end{align}
Define the implementation rent of \(x\) by
\begin{align}
\operatorname{Rent}(x)
:=
\inf_{R,\beta}
\left\{
\int_{\mathcal P}R(p)\,dG(p)
:
(x,R,\beta)\in\mathcal K
\right\},
\label{eq:rent-functional}
\end{align}
with the convention that \(\operatorname{Rent}(x)=+\infty\) when the feasible
set in \eqref{eq:rent-functional} is empty.
\end{definition}

\end{toappendix}

\Needspace{9\baselineskip}
\section{Implementability of Cutoff Rules}
\label{sec:single-threshold}

We now fix an approval cutoff and characterize the payments that make joint reporting truthful.
\begin{definition}[Single-threshold allocation rule]
\label{def:single-threshold-rule}
A \emph{single-threshold allocation rule} is an allocation rule \(x_c\) generated
by a measurable \emph{weak cutoff}\footnote{Endpoint cutoffs represent
always-action and never-action rules, apart from the treatment of the cutoff
type. Including them keeps these limiting rules in the admissible class for
the existence proof.}
$c:\mathcal P\to[\underline\theta,\bar\theta]$ according to
\begin{align}
x_c(\theta,p)
:=
\mathbb{1}\{\theta\ge c(p)\}.
\label{eq:single-threshold-x}
\end{align}
\end{definition}

Here we take the cutoff as given and find the least-cost payments that
implement it.\footnote{Appendix Figure
\ref{fig:randomized-allocation-examples} illustrates implementable
randomized rules, and Appendix~\ref{app:randomized-allocation-examples}
constructs their payments. Section~\ref{subsec:threshold-sufficiency}
gives a condition under which a threshold is optimal even among randomized
rules.}

\subsection{Utilities and Expected Payments}
\label{subsec:threshold-reduction}

At each belief, think of the threshold as a posted price 
for action. The
no-action option gives utility \(R(p)\). The action 
option gives the
manager her stake \(\theta\), charges the additional 
expected price
\(c(p)\), and includes the same \(R(p)\). The two options tie at the
cutoff. The stake envelope in Appendix Lemma~\ref{lem:withinbeliefenv}
gives
\begin{align}
U_{R,c}(\theta,p)
&=
R(p)+\int_{\underline\theta}^{\theta}x_c(s,p)\,ds
\nonumber\\
&=
\begin{cases}
R(p), & \theta<c(p),\\[4pt]
R(p)+\theta-c(p), & \theta\ge c(p).
\end{cases}
\label{eq:threshold-utility}
\end{align}
The lowest-stake type receives \(R(p)\), including at beliefs where she
takes the action because the cutoff reaches her stake. With
\((z)_+:=\max\{z,0\}\), we can write the same utility as
\begin{align}
U_{R,c}(\theta,p)
&=
R(p)+(\theta-c(p))_+
=
\max\{R(p),\,\theta+R(p)-c(p)\}.
\label{eq:threshold-utility-max}
\end{align}
The truthful expected payment, \(T=\theta x-U\), is
\begin{align}
T_{R,c}(\theta,p)
=
\begin{cases}
-R(p), & \theta<c(p),\\[4pt]
c(p)-R(p), & \theta\ge c(p).
\end{cases}
\label{eq:threshold-onpath-transfer}
\end{align}
The tilt divides these expected payments across states and determines how
other beliefs value the offers. For the fixed cutoff, the principal seeks
the lowest expected floor among the schedules that admit a truthful,
individually rational implementation. Write
\(\operatorname{Rent}(c):=\operatorname{Rent}(x_c)\).

\subsection{Implementability}
\label{subsec:threshold-inner}

Fix the manager's stake. Each report gives her a payoff that is affine
in her belief, so truthful utility, the payoff from her best report, is
convex in belief. At the extreme stakes, this puts a convexity restriction
on each of the two schedules in \eqref{eq:threshold-utility-max}.

\begin{restatable}[Threshold implementability]{prop}{thresholdimplementability}
\label{prop:threshold-implementability}
Fix a single-threshold cutoff \(c\) as in Definition
\ref{def:single-threshold-rule} and a \(G\)-integrable floor schedule \(R\).
There exists a globally BIC mechanism with allocation \(x_c\) and truthful
utility \(U_{R,c}\) if and only if both \(R\) and \(R-c\) are convex on
\(\mathcal P\). Such an implementation satisfies truthful-path interim IR if
and only if, in addition, \(R\ge0\).
\end{restatable}

At the lowest stake, truthful utility is \(R\). At the highest stake, it
is \(\bar\theta+R-c\). Both must be convex in the belief. Conversely,
when both schedules are convex, so is the maximum of the two options.
The slope of the active option with respect to the belief supplies the
transfer tilt.

\begin{appendixproof}[Proof of Proposition \ref{prop:threshold-implementability}]
\phantomsection\label{app:proof-threshold-implementability}
By Proposition \ref{prop:convex-utility-characterization}, global BIC requires
\(U_{R,c}\) to be convex. For every weak cutoff,
\[
U_{R,c}(\underline\theta,p)=R(p),
\qquad
U_{R,c}(\bar\theta,p)=\bar\theta+R(p)-c(p),
\]
so convexity of \(U_{R,c}\) implies that both \(R\) and \(R-c\) are convex.
Conversely, if both functions are convex, then \(R(p)\) and
\(\theta+R(p)-c(p)\) are jointly convex in \((\theta,p)\). Their maximum in
\eqref{eq:threshold-utility-max} is therefore convex. Below the cutoff, choose
a subgradient \((0,q)\) with \(q\in\partial R(p)\); at and above the cutoff,
choose \((1,q)\) with \(q\in\partial(R-c)(p)\). Thus an implementing tilt
can be selected according to
\begin{align}
\beta(\theta,p)
\in
\begin{cases}
\partial R(p), & \theta<c(p),\\[4pt]
\partial(R-c)(p), & \theta\ge c(p),
\end{cases}.
\label{eq:threshold-beta-subgradient}
\end{align}
Here \(\partial\) denotes the subdifferential on \(\mathcal P\).
Right derivatives provide measurable selections; at a preference endpoint,
the subgradient of truthful utility is understood relative to \(\Theta\). Proposition
\ref{prop:convex-utility-characterization} then gives global BIC. Finally,
\eqref{eq:threshold-utility} shows that truthful-path interim IR is equivalent
to \(R\ge0\).
\end{appendixproof}

The criterion is distribution-free: apart from the integrability needed to
evaluate expected rents, implementability depends only on the geometry of the
cutoff and the floor.

Proposition \ref{prop:threshold-implementability} turns the inner problem
into the following convex program over the floor schedule:
\begin{align}
\operatorname{Rent}(c)
=
\inf_R
\left\{
\int_{\mathcal P}R(p)\,dG(p)
:
R\ge0,
\quad R\text{ and }R-c\text{ are convex on }\mathcal P
\right\},
\tag{TFP}\label{eq:threshold-rent}
\end{align}
where \(R\) is finite at every belief and has finite expected cost.

\paragraph{Tilt Recovery.} Every feasible floor admits an implementing tilt. Where \(R\) and \(c\)
are differentiable, use the belief slope of the corresponding option.\footnote{At
nondifferentiable points, choose
\(\beta\in\partial R(p)\) below the cutoff and
\(\beta\in\partial(R-c)(p)\) at or above it. Right derivatives provide
measurable selections; see the appendix proof of Proposition
\ref{prop:threshold-implementability}.}
\begin{align}
\beta(\theta,p)
=
\begin{cases}
R'(p), & \theta<c(p),\\[4pt]
R'(p)-c'(p), & \theta\ge c(p).
\end{cases}
\label{eq:threshold-beta-differentiable}
\end{align}
The state-contingent transfers follow from the recovery formulas
\eqref{eq:tilt-transfer-1}--\eqref{eq:tilt-transfer-0} in Appendix
Definition~\ref{def:state-contingent-tilt}, using the expected payment
in \eqref{eq:threshold-onpath-transfer}.

The running example contract pays the manager her baseline bonus
when she declines, regardless of the state. In such a contract, the firm
needs to verify the state only after expansion. Appendix~\ref{app:no-action-payments}
gives the exact condition for a cutoff and floor to admit state-independent
payments when the action is not taken.

\section{The Cost of Implementing a Cutoff}
\label{sec:threshold-curvature}

We now solve the floor-minimization problem in
\eqref{eq:threshold-rent}. Its value is the additional rent
required to elicit the forecast, beyond ordinary
stake-screening rents. Our main result is that rent comes 
from upward changes in the cutoff's slope, weighted by 
where they occur in the forecast distribution. 

To illustrate the result, consider a piecewise-linear cutoff with a single upward kink. Let \(\mu:=\int_{\mathcal P}p\,dG(p)\) be the mean belief. Bayes
plausibility makes it equal to the principal's prior,
\(\Pr(\omega=1)\).\footnote{Equation
\eqref{eq:posterior-consistency} and the law of iterated expectations imply
calibration, \(\Pr(\omega=1\mid p)=p\) almost surely. Averaging yields the
stated Bayes-plausibility identity. Since \(p\in\mathcal P\) almost surely
and \(\mathcal P\) is an interval, \(\mu\in\mathcal P\).}
Panel (a) of Figure \ref{fig:kink-utility-floor} shows a cutoff that falls
linearly in belief until it reaches \(\underline\theta\) at \(t>\mu\), then
stays at that bound. It has an \emph{upward kink} because the slope
rises from a negative value to zero.

\begin{figure}[htbp]
\centering
\caption{\textbf{Implementing an upward kink}}
\label{fig:kink-utility-floor}
\begin{tikzpicture}[x=1cm,y=0.7cm,
  curve/.style={very thick,blue!70!black},
  guide/.style={densely dashed,black!35},
  axis/.style={->,black!65},
  every node/.style={font=\footnotesize}]
  \begin{scope}
    \node[font=\small] at (2,5.75) {(a) Fixed cutoff};
    \draw[axis] (0,0) -- (4.2,0) node[right] {$p$};
    \draw[axis] (0,0) -- (0,4.45) node[above] {$c(p)$};
    \draw[guide] (1.2,0) -- (1.2,4.3);
    \draw[guide] (2.4,0) -- (2.4,4.3);
    \draw[curve] (0,2.7) -- (2.4,0.3) -- (4,0.3);
    \fill[blue!70!black] (2.4,0.3) circle (1.5pt);
    \node[below] at (1.2,-0.08) {$\mu$};
    \node[below] at (2.4,-0.08) {$t$};
    \node[left] at (0,0.3) {$\underline\theta$};
    \node[above right] at (0.35,2.35) {slope $-a$};
    \node[above] at (3.25,0.48) {slope $0$};
  \end{scope}
  \begin{scope}[shift={(5.35,0)}]
    \node[font=\small] at (2,5.75) {(b) Least-cost floor};
    \draw[axis] (0,0) -- (4.2,0) node[right] {$p$};
    \draw[axis] (0,0) -- (0,4.45) node[above] {$R^c(p)$};
    \draw[guide] (1.2,0) -- (1.2,4.3);
    \draw[guide] (2.4,0) -- (2.4,4.3);
    \draw[curve] (0,0) -- (2.4,0) -- (4,1.6);
    \fill[blue!70!black] (2.4,0) circle (1.5pt);
    \node[below] at (1.2,-0.08) {$\mu$};
    \node[below] at (2.4,-0.08) {$t$};
    \node[left] at (0,0) {$0$};
    \node[above] at (1.2,0.18) {slope $0$};
    \node[above left] at (3.35,0.95) {slope $a$};
  \end{scope}
  \begin{scope}[shift={(10.7,0)}]
    \node[font=\small] at (2,5.75) {(c) Action utility less stake};
    \draw[axis] (0,0) -- (4.2,0) node[right] {$p$};
    \draw[axis] (0,0) -- (0,4.45) node[above] {$R^c(p)-c(p)$};
    \draw[guide] (1.2,0) -- (1.2,4.3);
    \draw[guide] (2.4,0) -- (2.4,4.3);
    \draw[guide] (0,2.7) -- (2.4,2.7);
    \draw[curve] (0,0.3) -- (4,4.3);
    \node[below] at (1.2,-0.08) {$\mu$};
    \node[below] at (2.4,-0.08) {$t$};
    \node[left] at (0,2.7) {$-\underline\theta$};
    \node[below right] at (2.65,2.95) {slope $a$};
  \end{scope}
\end{tikzpicture}
\par\vspace{2mm}
\begin{minipage}{\textwidth}
\small
\begin{singlespace}
\textbf{Notes.} The cutoff \(c(p)=\underline\theta+a(t-p)_+\) in
panel (a) has an upward kink at \(t>\mu\). The least-cost floor
\(R^c(p)=a(p-t)_+\) in panel (b) has the same kink and is positive only
above \(t\). Panel (c) shows the resulting action utility net of the
manager's stake, \(R^c(p)-c(p)\), which is affine.
\end{singlespace}
\end{minipage}
\end{figure}
\FloatBarrier
With a flat floor, action utility would bend downward at \(t\), violating
the convexity condition in Proposition~\ref{prop:threshold-implementability}.
The matching kink in the floor makes action utility affine while keeping
the floor convex. The least-cost floor shown in panel (b) stays at zero
through \(t\), then rises just enough to keep action utility convex.

\paragraph{Pricing upward slope changes.}
The cost of the kink in Figure~\ref{fig:kink-utility-floor} depends on its
size and its location in the belief distribution. For a possible kink at
\(t\in\mathcal P\), define the \emph{Jensen-gap weight} for this one-kink example as
\begin{align}
J_G(t)
&:=\int_{\mathcal P}(p-t)_+\,dG(p)-(\mu-t)_+.
\label{eq:jensen-kernel}
\end{align}
The weight is the nonnegative gap in Jensen's inequality for the
convex function \(p\mapsto(p-t)_+\). We study this function in our one-kink example because the lowest-stake manager's extra rent is zero up
to \(t\) and rises linearly above it, just like \((p-t)_+\).

For \(t\ge\mu\), this is \(\mathbb E_G[(p-t)_+]\). For \(t<\mu\),
it equals \(\mathbb E_G[(t-p)_+]\) because \(\mathbb E_G[p]=\mu\).
A unit upward kink costs the expected distance beyond \(t\) in the tail
on its side of the mean, counting beliefs outside that tail as zero.
The same weight appears in convex-order analysis
\citep{rothschild1970increasing,gentzkow2016rothschild}.

We now expand this idea to the general case of a piecewise-linear cutoff (i.e., multiple kinks). Consider a continuous piecewise-linear weak cutoff \(c\) with finitely many
interior knots \(t_1<\cdots<t_m\). Write
\[
\Delta_j:=c'_+(t_j)-c'_-(t_j)
\]
for the change in slope at knot \(t_j\), and recall
\((z)_+:=\max\{z,0\}\). Its minimum implementation rent is exactly
\begin{align}
\operatorname{Rent}(c)
&=\sum_{j=1}^{m}J_G(t_j)(\Delta_j)_+.
\label{eq:finite-kink-rent-main}
\end{align}
Each upward change in slope is priced at its belief weight; downward
changes add no cost. The sum is finite and its minimum is attained.
A least-cost floor adds the upward-kink corrections illustrated above,
with each correction chosen to be zero at \(\mu\). An upward kink left of
\(\mu\) uses a downward-sloping floor correction; one right of \(\mu\)
uses an upward-sloping correction. Appendix
Corollary~\ref{cor:piecewise-linear-threshold-inner} gives the explicit
floor and also covers locally finite sets of knots.

The same price applies to a smooth upward bend. If the cutoff and its
first two derivatives extend continuously to
\(\overline{\mathcal P}\), then
\begin{align}
\operatorname{Rent}(c)
&=\int_{\mathcal P}J_G(t)(c''(t))_+\,dt.
\label{eq:smooth-rent-main}
\end{align}
This integrates the positive second derivative against the same weight.
Appendix Corollary~\ref{cor:smooth-threshold-inner} constructs a least-cost
floor. Both the finite-kink and smooth cases admit bounded implementing
state-contingent payments; see Appendix
Corollary~\ref{cor:bounded-transfers-finite-curvature}.

\paragraph{The general characterization.}
Both formulas use the same floor component for one unit of upward slope
change at \(t\). It is zero at the mean belief:
\begin{align}
K_\mu(p,t)
&:=
\begin{cases}
(t-p)_+, & t<\mu,\\
(p-t)_+, & t\ge\mu.
\end{cases}
\label{eq:centered-hinge-kernel}
\end{align}
For either location of a single kink, Bayes plausibility gives
\(J_G(t)=\int_{\mathcal P}K_\mu(p,t)\,dG(p)\).

For several upward kinks, we add the corresponding floor components. For a
smooth upward bend, we integrate them instead. The curvature measure defined
next records where---and by how much---the cutoff's slope rises in either case.
A candidate cutoff \(c:\mathcal P\to[\underline\theta,\bar\theta]\) is
\emph{locally difference-of-convex} (locally DC) if its restriction to every compact
subinterval is the difference of two finite convex functions.\footnote{Difference-of-convex
functions and curvature-measure representations are standard tools of convex
analysis; see \citet{hartman1959functions,dudley1977second}.} Such a cutoff
is continuous, and its second distributional derivative is a signed Radon
measure with locally finite variation. Write
\[
D^2c=(D^2c)_+-(D^2c)_-
\]
for its Jordan decomposition. The positive measure \((D^2c)_+\) records every
upward change in the slope of the cutoff, including both smooth curvature and
kinks.

The following theorem separates the existence of a finite floor at each belief
from the finiteness of its expected cost. It combines kinks, smooth bends,
and other changes in slope in a single formula.

\Needspace{22\baselineskip}
\begin{restatable}[Cutoff implementation cost]{theorem}{finiterentcurvature}
\label{thm:finite-rent-curvature}
Let \(c:\mathcal P\to[\underline\theta,\bar\theta]\) be a cutoff. If \(c\) is not locally DC pointwise,
no finite-valued implementing floor exists. Otherwise,
\begin{align}
\operatorname{Rent}(c)
&=\int_{\mathcal P}J_G(t)\,d(D^2c)_+(t),
\label{eq:curvature-rent-formula}\\
R^c(p)
&=\int_{\mathcal P}K_\mu(p,t)\,d(D^2c)_+(t).
\label{eq:curvature-least-cost-floor}
\end{align}
The floor \(R^c\) is finite at every forecast and has expected
cost \(\operatorname{Rent}(c)\), possibly \(+\infty\).
If this cost is finite, \(R^c\) attains the minimum and every
least-cost floor satisfies \(R(\mu)=0\).
If it is infinite, every implementing floor has infinite expected cost.
\end{restatable}

A jump in the cutoff at an interior belief rules out a finite-valued
floor because both \(R\) and \(R-c\) must be continuous. A kink changes
the cutoff's slope and is compatible with a finite floor. Concave cutoffs 
require no additional forecast rent. For a convex cutoff, the cost is
\(\int_{\mathcal P}c(p)\,dG(p)-c(\mu)\). Adding an affine function to
a cutoff leaves its rent unchanged as long as the resulting cutoff stays
within the stake bounds.

The proof starts from the two convexity constraints. Convexity of
\(R-c\) requires the floor to match every upward bend of the cutoff.
Convexity of \(R\) rules out downward bends in the floor. The cost of
any feasible floor is its value at the mean plus the weighted cost of
its upward bends. Subtracting a supporting line at the mean preserves
both convexity constraints and leaves the floor nonnegative, so the
principal can set the floor to zero there. The construction in
\eqref{eq:curvature-least-cost-floor} attains the lower bound by adding
only the required bends. Appendix~\ref{app:full-curvature-characterization}
gives the proof.

\Needspace{10\baselineskip}
\begin{mdframed}[style=runningexamplebox]
\vspace{-5mm}
\subsection{Running Example: Pricing a Cutoff with Smaller Stakes}
\label{ex:small-stakes-cutoff-cost}

Return to the running example in
Section~\ref{subsec:broad-stakes-capacity}. Its affine cutoff
\(1-p\) requires no additional forecast rent. Now keep the same
project payoffs and uniform forecast distribution, but let the
manager's stake be independently uniform on
\([-0.005,0.005]\), a stake range of \(h=0.01\).

Consider the fixed cutoff
\[
c(p)=\min\{0.005,\max\{-0.005,0.5025-p\}\}.
\]
The cutoff stays flat at the highest stake until \(p=0.4975\),
falls linearly to the lowest stake at \(p=0.5075\), and then
remains flat, creating one downward kink followed by one upward kink.

\par\smallskip
{\centering
\begin{tikzpicture}[x=1.7cm,y=1.35cm,
  every node/.style={font=\small},
  cutoff/.style={very thick,blue!70!black},
  guide/.style={densely dashed,black!30}]
  \draw[guide] (0,0)--(5,0);
  \draw[guide] (1.5,-1.35)--(1.5,1);
  \draw[guide] (3.5,-1.35)--(3.5,-1);
  \draw[->,black!65] (0,-1.35)--(5.25,-1.35);
  \draw[->,black!65] (0,-1.35)--(0,1.45)
    node[above] {$c(p)$};
  \foreach \y/\value in {-1/-0.005,0/0,1/0.005} {
    \draw[black!65] (-0.04,\y)--(0.04,\y);
    \node[left=3pt] at (0,\y) {$\value$};
  }
  \foreach \x/\value in {0/0.490,1.5/0.4975,5/0.515} {
    \draw[black!65] (\x,-1.30)--(\x,-1.40);
    \node[below=3pt] at (\x,-1.35) {$\value$};
  }
  \draw[black!65] (3.5,-1.30)--(3.5,-1.40);
  \node[below=3pt] at (3.5,-1.35) {$t=0.5075$};
  \draw[cutoff] (0,1)--(1.5,1)--(3.5,-1)--(5,-1);
  \fill[blue!70!black] (1.5,1) circle (1.7pt);
  \fill[blue!70!black] (3.5,-1) circle (1.7pt);
  \node[above=5pt,align=center,font=\footnotesize] at (1.5,1)
    {downward kink\\slope $0\to-1$};
  \node[above right=4pt,align=left,font=\footnotesize] at (3.5,-1)
    {upward kink\\slope $-1\to0$};
  \node[below=22pt] at (2.5,-1.35)
    {Forecast $p$ (zoomed around $0.5$)};
\end{tikzpicture}\par}
\smallskip
Section~\ref{sec:pointwise-benchmark} explains why the firm would
choose this rule if it observed the forecast. 

Here we calculate the additional rent needed to implement it 
when the forecast is private. The first kink
is downward and adds no rent. The second raises the slope by
one and therefore requires a matching kink in the utility floor.
Since \(t>\mu=1/2\), the least-cost floor is
\[
R^c(p)=(p-0.5075)_+.
\]
The lowest-stake manager receives no extra rent below this
forecast; above it, her extra rent rises linearly.

With uniform forecasts, Theorem~\ref{thm:finite-rent-curvature}
therefore gives
\begin{align*}
\operatorname{Rent}(c)
&=J_G(0.5075)\\
&=\int_{0.5075}^{1}(p-0.5075)\,dp\\
&=\frac{(1-0.5075)^2}{2}
=0.121278125.
\end{align*}
Thus retaining this cutoff requires about \(0.12128\) payoff
units in additional expected forecast rent. This calculation
holds the approval rule fixed.
Section~\ref{early-capacity-expansion} returns to this
small-stakes variation and shows how the firm changes its
approval rule to reduce these costs.
\end{mdframed}

\begin{toappendix}
\subsection{Proof of the Curvature Characterization}
\label{app:full-curvature-characterization}

\begin{proof}[Proof of Theorem \ref{thm:finite-rent-curvature}]
\phantomsection\label{app:proof-curvature-rent}
We first establish a representation for convex functions on the open interval
\(\mathcal P\). For any nonnegative locally finite Radon measure \(\lambda\),
define
\[
\Psi_\lambda(p)
:=
\int_{\mathcal P}K_\mu(p,t)\,d\lambda(t).
\]
For fixed \(p\), the integrand vanishes unless \(t\) lies between \(p\) and
\(\mu\). This is a compact subinterval of \(\mathcal P\), so local finiteness
of \(\lambda\) makes \(\Psi_\lambda(p)\) finite even if \(\lambda\) has
infinite total mass near an open endpoint. Equation
\eqref{eq:centered-hinge-kernel} also gives
\(\Psi_\lambda\ge0\) and \(\Psi_\lambda(\mu)=0\).

For each fixed \(t\), the function \(p\mapsto K_\mu(p,t)\) is convex and has
second distributional derivative \(\delta_t\). Integrating the convexity
inequality therefore shows that \(\Psi_\lambda\) is convex. More formally, if
\(\varphi\in C_c^\infty(\mathcal P)\), only \(t\) in the compact hull of
\(\operatorname{supp}\varphi\cup\{\mu\}\) contributes below. Fubini's
theorem gives
\[
\int_{\mathcal P}\Psi_\lambda(p)\varphi''(p)\,dp
=
\int_{\mathcal P}
\left[
\int_{\mathcal P}K_\mu(p,t)\varphi''(p)\,dp
\right]d\lambda(t)
=
\int_{\mathcal P}\varphi(t)\,d\lambda(t).
\]
Hence \(D^2\Psi_\lambda=\lambda\).

We next compute its expected cost. If \(t\ge\mu\), then
\(K_\mu(p,t)=(p-t)_+\), whose expectation is \(J_G(t)\). If \(t<\mu\), the
identity
\[
p-t=(p-t)_+-(t-p)_+
\]
and Bayes plausibility give
\[
\int_{\mathcal P}K_\mu(p,t)\,dG(p)
=
\int_{\mathcal P}(t-p)_+\,dG(p)
=J_G(t).
\]
Since \(K_\mu\ge0\), Tonelli's theorem applies without an integrability
premise and yields, in the extended nonnegative reals,
\begin{align}
\int_{\mathcal P}\Psi_\lambda(p)\,dG(p)
&=
\int_{\mathcal P}J_G(t)\,d\lambda(t).
\label{eq:convex-curvature-cost-identity}
\end{align}

This construction gives the corresponding representation of any finite
convex function \(r\) on \(\mathcal P\). Let \(\rho:=D^2r\) and
\(s:=r'_-(\mu)\). Local finiteness and
\eqref{eq:centered-hinge-kernel} imply
\(\Psi'_{\rho,-}(\mu)=0\); an atom of \(\rho\) at \(\mu\) appears in the
right-hand hinge. The functions \(r\) and
\(r(\mu)+s(p-\mu)+\Psi_\rho(p)\) have the same second distributional
derivative, value at \(\mu\), and left derivative there. Their difference has
zero second distributional derivative and is therefore affine; these two
normalizations make it zero. Thus, pointwise on \(\mathcal P\),
\begin{align}
r(p)
&=
r(\mu)+s(p-\mu)+\Psi_\rho(p).
\label{eq:convex-curvature-representation}
\end{align}
Combining Bayes plausibility with
\eqref{eq:convex-curvature-cost-identity} and
\eqref{eq:convex-curvature-representation} gives
\begin{align}
\int_{\mathcal P}r(p)\,dG(p)
&=
r(\mu)+\int_{\mathcal P}J_G(t)\,d\rho(t),
\label{eq:convex-jensen-curvature-identity}
\end{align}
where the two sides may equal \(+\infty\). The identity is well defined: a
supporting affine function at \(\mu\) bounds \(r\) below, so its negative part
is integrable on the bounded interval \(\mathcal P\).

We now turn to the cutoff. Suppose first that a finite-valued feasible floor
\(R\) exists and let \(Q:=R-c\). Both \(R\) and \(Q\) are finite convex
functions, and
\[
c=R-Q
\]
at every report. The actual cutoff \(c\), not merely a \(G\)-almost-everywhere
version, is therefore continuous and locally DC.

Conversely, suppose \(c\) is locally DC. Its second distributional derivative
is a signed locally finite Radon measure
\(\nu=\nu_+-\nu_-\). Set
\(R^c:=\Psi_{\nu_+}\), as in
\eqref{eq:curvature-least-cost-floor}. The preceding construction makes
\(R^c\) finite-valued, nonnegative, and convex, with
\(D^2R^c=\nu_+\). Moreover,
\[
D^2(R^c-c)=\nu_+-\nu=\nu_-.
\]
To verify convexity pointwise, observe that
\((R^c-c)-\Psi_{\nu_-}\) is continuous and has zero second distributional
derivative. It therefore equals an affine function everywhere. Hence
\(R^c-c\) is convex on all of \(\mathcal P\), so \(R^c\) is a finite-valued
feasible floor.

Let \(R\) be any finite-valued feasible floor and write
\(\rho:=D^2R\). Convexity of \(R\) and \(R-c\) gives
\[
\rho\ge0,
\qquad
\rho-\nu\ge0.
\]
It follows that \(\rho\ge\nu_+\). To verify this measure inequality, fix a
relatively compact Borel set \(A\subset\mathcal P\). Because all measures are
finite on \(A\), the variational definition of the positive part gives
\[
\nu_+(A)
=
\sup_{B\subseteq A}\nu(B)
\le
\sup_{B\subseteq A}\rho(B)
\le
\rho(A),
\]
where the suprema are over Borel subsets. Exhausting \(\mathcal P\) by compact
subintervals proves the inequality on the open interval.

The curvature-dominance argument is standard in convex correction problems;
see \citet{kong2024cost}. Here the objective is the expected cost of a
nonnegative utility floor under \(G\).

Applying \eqref{eq:convex-jensen-curvature-identity} to \(R\), using
\(R(\mu)\ge0\), \(J_G\ge0\), and \(\rho\ge\nu_+\), gives
\begin{align}
\int_{\mathcal P}R(p)\,dG(p)
&=
R(\mu)+\int_{\mathcal P}J_G(t)\,d\rho(t)
\nonumber\\
&\ge
\int_{\mathcal P}J_G(t)\,d\nu_+(t).
\label{eq:curvature-rent-lower-bound}
\end{align}
For \(R^c=\Psi_{\nu_+}\), identity
\eqref{eq:convex-curvature-cost-identity} gives equality. If the last integral
is finite, \(R^c\) is \(G\)-integrable and attains the infimum. If it is
infinite, \eqref{eq:curvature-rent-lower-bound} shows that every finite-valued
feasible floor has infinite expectation. If \(c\) is not locally DC, the
pointwise argument above shows that no finite-valued feasible floor exists.
This proves the characterization and \eqref{eq:curvature-rent-formula}.

Finally, let \(R\) be any least-cost floor and choose
\(s\in\partial R(\mu)\). The supporting-line normalization
\[
\widehat R(p)
:=
R(p)-R(\mu)-s(p-\mu)
\]
is nonnegative. Subtracting an affine function preserves convexity of both
\(R\) and \(R-c\), so \(\widehat R\) is feasible. Bayes plausibility gives
\[
\int_{\mathcal P}\widehat R\,dG
=
\int_{\mathcal P}R\,dG-R(\mu).
\]
Minimality and \(R(\mu)\ge0\) therefore imply \(R(\mu)=0\).
\end{proof}

For completeness, with the cutoff virtual surplus \(B\) and net payoff
\(\mathcal J\) defined in Section~\ref{sec:pointwise-benchmark}, the
corresponding general objective is
\begin{align}
\mathcal J(c)
&=
\int_{\mathcal P}B(c(p),p)\,dG(p)
-
\int_{\mathcal P}J_G(t)\,d(D^2c)_+(t)
\label{eq:curvature-adjusted-cutoff-objective}
\end{align}
for every finite-rent cutoff. For finitely many knots, the rent term
is the sum in \eqref{eq:finite-kink-rent-main}.
\end{toappendix}

\begin{toappendix}
\begin{corollary}[Bounded transfers under finite total curvature]
\label{cor:bounded-transfers-finite-curvature}
Let \(c:\mathcal P\to[\underline\theta,\bar\theta]\) be a locally DC weak
cutoff. If
\begin{align}
|D^2c|(\mathcal P)<\infty,
\label{eq:finite-total-cutoff-curvature}
\end{align}
then \(c\) has finite implementation rent, and its canonical floor \(R^c\) in
\eqref{eq:curvature-least-cost-floor} admits a bounded measurable implementing
tilt. The state-contingent transfers recovered from
\eqref{eq:tilt-transfer-1}--\eqref{eq:tilt-transfer-0} are therefore bounded.
\end{corollary}

\begin{proof}
Write
\[
D^2c=\nu_+-\nu_-,
\qquad
\nu_+(\mathcal P)+\nu_-(\mathcal P)<\infty,
\]
and let \(L:=\bar p-\underline p\). Because
\(0\le K_\mu(p,t)\le L\), the canonical floor satisfies
\[
0\le R^c(p)
=
\int_{\mathcal P}K_\mu(p,t)\,d\nu_+(t)
\le
L\nu_+(\mathcal P)
\qquad
\forall p\in\mathcal P.
\]
Thus \(R^c\) is bounded and hence \(G\)-integrable. Theorem
\ref{thm:finite-rent-curvature} therefore implies that \(c\) has finite
implementation rent and that \(R^c\) is a least-cost floor.

For fixed \(t\), every one-sided slope of
\(p\mapsto K_\mu(p,t)\) lies in \([-1,1]\). Dominated convergence, using
\(\nu_+(\mathcal P)<\infty\), therefore permits differentiation under the
integral. Consequently, the right derivative
\[
r(p):=\partial_p^+R^c(p)
\]
is a Borel selection from \(\partial R^c(p)\) and satisfies
\[
|r(p)|\le\nu_+(\mathcal P)
\qquad
\forall p\in\mathcal P.
\]

Set \(Q^c:=R^c-c\). The construction of \(R^c\) gives
\(D^2Q^c=\nu_-\). Hence
\[
Q^c(p)=\Psi_{\nu_-}(p)+a+bp
\]
for finite constants \(a,b\). The constants are finite because
\(Q^c-\Psi_{\nu_-}\) is a finite affine function on the open interval
\(\mathcal P\). Dominated convergence under the finite measure \(\nu_-\) and
the same one-sided-derivative argument therefore give a Borel selection
\[
q(p):=\partial_p^+Q^c(p)\in\partial Q^c(p)
\]
such that
\[
|q(p)|
\le
|b|+\nu_-(\mathcal P)
\qquad
\forall p\in\mathcal P.
\]

Define
\[
\beta(\theta,p)
:=
\begin{cases}
r(p), & \theta<c(p),\\[3pt]
q(p), & \theta\ge c(p).
\end{cases}
\]
Local DC regularity makes \(c\) continuous, so this selection is measurable.
Equation \eqref{eq:threshold-beta-subgradient} shows that it implements
\(x_c\), and the preceding bounds make it bounded. Finally,
\eqref{eq:threshold-onpath-transfer} shows that \(T_{R^c,c}\) is bounded
because \(R^c\) and \(c\) are bounded. Since \(p\in[0,1]\), the recovery
formulas \eqref{eq:tilt-transfer-1}--\eqref{eq:tilt-transfer-0} then make both
state-contingent transfers bounded.
\end{proof}
\end{toappendix}

\begin{toappendix}
\subsection{State-Independent Payments after No Action}
\label{app:no-action-payments}

If payments do not depend on the state when the action is not taken, the state
need not be verified to calculate payment on that branch. We first give an
exact test for a fixed cutoff and floor. When every belief permits a no-action
report, the test reduces to concavity of the cutoff. The conclusion can differ
when the cutoff reaches the lowest stake.

\begin{prop}[State-independent no-action payments]
\label{prop:no-action-payments}
Fix a weak cutoff \(c\) and a finite, \(G\)-integrable floor \(R\). There
exists a globally BIC/IR implementation of \(x_c\) with truthful utility
\(U_{R,c}\) satisfying
\begin{align}
t(\theta,p,0)=t(\theta,p,1)
\qquad\text{whenever }\theta<c(p)
\label{eq:no-action-state-independence}
\end{align}
if and only if \(R\ge0\), both \(R\) and \(R-c\) are convex, and
\begin{align}
0\in\partial R(p)
\qquad\text{for every }p\text{ with }c(p)>\underline\theta.
\label{eq:no-action-floor-condition}
\end{align}
In particular, if \(c(p)>\underline\theta\) at every belief, the cutoff
admits such an implementation for some floor if and only if \(c\) is concave.
When it is concave, \(R\equiv0\) is a least-cost floor and no-action payments
can be zero.
\end{prop}

\begin{proof}
Proposition \ref{prop:threshold-implementability} gives the convexity and IR
conditions. Suppose first that an implementation also satisfies
\eqref{eq:no-action-state-independence}. Fix \(p\) with
\(c(p)>\underline\theta\). The report \((\underline\theta,p)\) then induces no
action. Its truthful expected transfer is \(-R(p)\) by
\eqref{eq:threshold-onpath-transfer}, and state independence makes its
realized transfer \(-R(p)\) in either state. A lowest-stake type with any true
belief \(q\) can choose this report and obtain \(R(p)\). Her truthful utility
is \(R(q)\), so global BIC gives \(R(q)\ge R(p)\) for every \(q\). Thus \(p\)
minimizes the convex function \(R\), or equivalently
\(0\in\partial R(p)\).

Conversely, suppose the stated conditions hold. Set \(\beta(\theta,p)=0\)
below the cutoff. At those beliefs, \eqref{eq:no-action-floor-condition}
gives \(0\in\partial R(p)\). At and above the cutoff, take the right
derivative of the finite convex function \(R-c\) as a measurable selection
for \(\beta(\theta,p)\). Equation \eqref{eq:threshold-beta-subgradient} and
Proposition \ref{prop:convex-utility-characterization} then give a globally
BIC/IR implementation. Since
\(\beta(\theta,p)=t(\theta,p,0)-t(\theta,p,1)\), its two payments coincide
below the cutoff.

If \(c(p)>\underline\theta\) for every \(p\), the argument above makes every
belief a minimizer of \(R\); hence \(R\) is constant. Convexity of \(R-c\)
then requires \(c\) to be concave. Conversely, if \(c\) is concave, choose
\(R\equiv0\), zero tilt below the cutoff, and a subgradient of \(-c\) at
and above it. This implements \(c\) with zero no-action payments. It also
minimizes expected floor cost because every IR floor is nonnegative.
\end{proof}

The condition that every belief permits a no-action report matters.
Proposition \ref{prop:positive-rent-soft-concavification} gives a globally
optimal cutoff with an upward kink and positive rent. It reaches
\(\underline\theta\) at high beliefs, yet its transfers are zero whenever
the action is not taken.
\end{toappendix}

\begin{toappendix}
\subsection{Special Cases of the Curvature Formula}
\label{app:cutoff-curvature-specializations}

\subsubsection{Concave cutoffs.}

\begin{corollary}[Zero rent for concave cutoffs]
\label{cor:concave-threshold-inner}
Suppose \(c\) is concave on \(\mathcal P\). Then a least-cost floor is
$R^c(p)\equiv 0$, and \(\operatorname{Rent}(c)=0\).
\end{corollary}

\begin{proof}
For a concave cutoff, \((D^2c)_+=0\). Theorem
\ref{thm:finite-rent-curvature} therefore gives \(R^c\equiv0\) and zero rent.
\end{proof}

\Needspace{5\baselineskip}
\subsubsection{Convex cutoffs.}

\begin{corollary}[Jensen-gap rent for convex cutoffs]
\label{cor:convex-threshold-inner}
Suppose \(c\) is convex on \(\mathcal P\). For any
\(s\in\partial c(\mu)\), the floor schedule
\begin{align}
R^c_s(p)
:=
c(p)-c(\mu)-s(p-\mu)
\label{eq:convex-threshold-floor}
\end{align}
attains the infimum in \eqref{eq:threshold-rent}. The associated implementation rent is
\begin{align}
\operatorname{Rent}(c)
=
\int_{\mathcal P}c(p)\,dG(p)-c(\mu).
\label{eq:convex-threshold-rent}
\end{align}
\end{corollary}

\begin{proof}
For a convex cutoff, \((D^2c)_+=D^2c\). Subtracting any supporting line at
\(\mu\) gives the floor in \eqref{eq:convex-threshold-floor}. It is
nonnegative, has curvature \(D^2c\), and differs from \(c\) by an affine
function. Theorem \ref{thm:finite-rent-curvature} makes it least cost. Bayes
plausibility then gives \eqref{eq:convex-threshold-rent}.
\end{proof}

\subsubsection{Piecewise-linear cutoffs.}

For a piecewise-linear cutoff, positive curvature consists of upward jumps in
its slope. The general rent formula therefore becomes a sum of explicit local
prices.

\begin{corollary}[Kink-by-kink implementation rent]
\label{cor:piecewise-linear-threshold-inner}
Let \(c\) be a continuous piecewise-affine weak cutoff with a locally finite
set of knots \(\{t_j\}_{j\in I}\subset\mathcal P\). Define its slope
jump at knot \(t_j\) by
\[
\Delta_j
:=
c'_+(t_j)-c'_-(t_j).
\]
Then, in the extended nonnegative reals,
\begin{align}
\operatorname{Rent}(c)
&=
\sum_{j\in I}J_G(t_j)(\Delta_j)_+.
\label{eq:piecewise-linear-threshold-rent}
\end{align}
If this sum is finite, a least-cost floor is
\begin{align}
R^c(p)
&=
\sum_{j\in I}K_\mu(p,t_j)(\Delta_j)_+.
\label{eq:piecewise-linear-threshold-floor}
\end{align}
\end{corollary}

\begin{proof}
Local finiteness of the knot set makes \(c\) locally DC, and its second
distributional derivative is
\[
D^2c
=
\sum_{j\in I}\Delta_j\delta_{t_j}.
\]
The atoms occur at distinct points, so the positive part of this signed
measure is
\[
(D^2c)_+
=
\sum_{j\in I}(\Delta_j)_+\delta_{t_j}.
\]
Substitution into \eqref{eq:curvature-rent-formula} and
\eqref{eq:curvature-least-cost-floor} gives
\eqref{eq:piecewise-linear-threshold-rent} and
\eqref{eq:piecewise-linear-threshold-floor}, respectively.
\end{proof}

An upward kink is costly because the floor must reproduce its increase in
slope; a downward kink already helps make \(R-c\) convex. An upward kink has
strictly positive cost whenever \(J_G(t_j)>0\).

\subsubsection{Smooth mixed-curvature cutoffs.}

More generally, the cutoff may be convex on some regions of belief space and
concave on others. Theorem \ref{thm:finite-rent-curvature} already covers
nonsmooth mixed curvature. In the smooth case, it reduces to the following
familiar formula.

For the next result, restrict attention to cutoffs whose values and first two
derivatives extend continuously to \(\overline{\mathcal P}\). We write
\(c\in C^2(\overline{\mathcal P})\) for this class.

\begin{corollary}[Smooth positive-curvature formula]
\label{cor:smooth-threshold-inner}
Suppose \(c\in C^2(\overline{\mathcal P})\). Let
\(\psi\in C^2(\overline{\mathcal P})\) be any function
satisfying
\begin{align}
\psi''(p)=(c''(p))_+
\qquad
\forall p\in\overline{\mathcal P}.
\label{eq:psi-positive-curvature}
\end{align}
Then the floor schedule
\begin{align}
R^c(p)
=
\psi(p)-\psi(\mu)-\psi'(\mu)(p-\mu)
\label{eq:smooth-threshold-floor}
\end{align}
attains the infimum in \eqref{eq:threshold-rent}, with associated implementation
rent
\begin{align}
\operatorname{Rent}(c)
=
\int_{\mathcal P}\psi(p)\,dG(p)-\psi(\mu).
\label{eq:smooth-threshold-rent}
\end{align}
Both expressions are invariant to the affine normalization of \(\psi\).
\end{corollary}

\begin{proof}
Here \(d(D^2c)_+(p)=(c''(p))_+\,dp=D^2\psi\). The floor in
\eqref{eq:smooth-threshold-floor} is the supporting-line normalization of
this convex curvature potential. Theorem \ref{thm:finite-rent-curvature} and
Bayes plausibility give the two claims.
\end{proof}

The floor adds exactly the positive curvature of the cutoff:
\[
(R^c)''=(c'')_+,
\qquad
(R^c-c)''=(-c'')_+.
\]
The tangent normalization makes it touch zero at \(\mu\). If \(c\) is
concave, this reduces to the zero floor; if \(c\) is convex, it reduces to the
supporting-line formula in Corollary \ref{cor:convex-threshold-inner}.
\end{toappendix}


\begin{toappendix}
\subsection{Implementable Randomized Allocation Examples}
\label{app:randomized-allocation-examples}

\begin{figure}[htbp]
\centering
\caption{Implementable randomized allocation rules}
\label{fig:randomized-allocation-examples}
\resizebox{\textwidth}{!}{%
\begingroup
\definecolor{approval}{HTML}{514B95}
\newcommand{\smoothallocationfill}{%
  \begin{scope}
    \clip (0,0) rectangle (1,1);
    \foreach \i in {0,...,179} {
      \pgfmathsetmacro{\lo}{1.5*\i/180}
      \pgfmathsetmacro{\hi}{1.5*(\i+1)/180}
      \pgfmathsetmacro{\shadelevel}{100*4*((\lo+\hi)/2)^2/9}
      \fill[approval!\shadelevel!white]
        (0,\lo) -- (1,{\lo-0.5}) -- (1,{\hi-0.5+0.006}) -- (0,{\hi+0.006}) -- cycle;
    }
  \end{scope}%
}
\newcommand{\reportaxes}{%
  \draw[black!30,thin] (0,1) -- (1,1) -- (1,0);
  \draw[axis] (0,0) -- (1.055,0);
  \draw[axis] (0,0) -- (0,1.065);
  \foreach \q/\lab in {0/0,0.5/{\frac12},1/1} {
    \draw[black!60] (\q,-0.012) -- (\q,0.012);
    \node[below,font=\fontsize{8}{10}\selectfont] at (\q,-0.026) {$\lab$};
    \draw[black!60] (-0.012,\q) -- (0.012,\q);
    \node[left,font=\fontsize{8}{10}\selectfont] at (-0.02,\q) {$\lab$};
  }
  \node at (0.5,-0.20) {Reported belief $\hat p$};
}

\begin{tikzpicture}[x=4.2cm,y=4.2cm,
  axis/.style={->,black!65},
  boundary/.style={thick,approval!70!black},
  contour/.style={thin,black!45},
  every node/.style={font=\fontsize{9}{11}\selectfont}]

  \begin{scope}
    \fill[white] (0,0) rectangle (1,1);
    \fill[approval!50!white] (0,0.5) -- (1,0) -- (1,1) -- (0,1) -- cycle;
    \fill[approval] (0,1) -- (1,0.5) -- (1,1) -- cycle;
    \draw[boundary] (0,0.5) -- (1,0);
    \draw[boundary] (0,1) -- (1,0.5);
    \node[align=center,font=\fontsize{8}{10}\selectfont] at (0.25,0.15)
      {No expansion\\$x=0$};
    \node[align=center] at (0.50,0.48) {Lottery\\$x=1/2$};
    \node[align=center,text=white,font=\fontsize{8}{10}\selectfont] at (0.69,0.88)
      {Certain expansion\\$x=1$};
    \node at (0.5,1.16) {(a) Two cutoff levels};
    \reportaxes
    \node[rotate=90] at (-0.24,0.5) {Reported stake $\hat\theta$};
  \end{scope}

  \begin{scope}[shift={(5.6cm,0)}]
    \smoothallocationfill
    \begin{scope}
      \clip (0,0) rectangle (1,1);
      \draw[contour] (0,0.75) -- (1,0.25)
        node[pos=0.32,sloped,above,font=\fontsize{8}{10}\selectfont,text=black] {$x=1/4$};
      \draw[contour] (0,{1.5*sqrt(0.5)}) -- (1,{1.5*sqrt(0.5)-0.5})
        node[pos=0.48,sloped,above,font=\fontsize{8}{10}\selectfont,text=black] {$x=1/2$};
      \draw[contour] (0,{1.5*sqrt(0.75)}) -- (1,{1.5*sqrt(0.75)-0.5})
        node[pos=0.84,sloped,above,font=\fontsize{8}{10}\selectfont,text=white] {$x=3/4$};
    \end{scope}
    \node at (0.5,1.16) {(b) Smooth lottery};
    \reportaxes
  \end{scope}

  \begin{scope}[shift={(11.2cm,0)}]
    \smoothallocationfill
    \fill[approval] (0,0.75) -- (1,0.25) -- (1,1) -- (0,1) -- cycle;
    \draw[boundary] (0,0.75) -- (1,0.25);
    \node[align=center,font=\fontsize{8}{10}\selectfont] at (0.34,0.25)
      {Smooth lottery\\$0<x<1/4$};
    \node[align=center,text=white] at (0.55,0.80)
      {Certain expansion\\$x=1$};
    \node at (0.5,1.16) {(c) Lottery or certainty};
    \reportaxes
  \end{scope}

  \begin{scope}[shift={(5.3cm,-1.85cm)},x=4.8cm,y=1cm]
    \node[above,font=\fontsize{9}{11}\selectfont] at (0.5,0.25) {Expansion probability $x(\hat\theta,\hat p)$};
    \foreach \i in {0,...,99} {
      \pgfmathsetmacro{\shadelevel}{\i+0.5}
      \fill[approval!\shadelevel!white] ({\i/100},0) rectangle ({(\i+1)/100+0.0001},0.18);
    }
    \draw[black!40,thin] (0,0) rectangle (1,0.18);
    \foreach \q/\lab in {0/0,0.25/{\frac14},0.5/{\frac12},0.75/{\frac34},1/1} {
      \draw[black!60] (\q,0) -- (\q,-0.05);
      \node[below,font=\fontsize{8}{10}\selectfont] at (\q,-0.045) {$\lab$};
    }
  \end{scope}
\end{tikzpicture}
\endgroup%
}
\par\smallskip
\begin{minipage}{\textwidth}
\small\begin{singlespace}
\textbf{Notes:} Darker shading indicates a higher expansion probability;
the scale is common to all panels. Write \(s=\hat\theta+\hat p/2\).
Panel (a) assigns probabilities \(0\), \(1/2\), and \(1\) as \(s\) crosses
\(1/2\) and \(1\). Panel (b) uses \(x=4s^2/9\); its thin lines mark equal
probabilities, not jumps. Panel (c) uses the same smooth rule below
\(s=3/4\) and jumps to certain expansion at that boundary. At jumps, the
higher probability applies. The plots show the closure of
\(\Theta\times\mathcal P=[0,1]\times(0,1)\).
Appendix \ref{app:randomized-allocation-examples} constructs globally BIC/IR
implementations with bounded state-contingent transfers for all three rules.
\end{singlespace}
\end{minipage}
\end{figure}

This subsection constructs payments for the three rules in Figure
\ref{fig:randomized-allocation-examples}. Take \(\Theta=[0,1]\) and
\(\mathcal P=(0,1)\), and write \(s=\theta+p/2\). The allocation rules are
\begin{align}
x_A(\theta,p)
&=\tfrac12\mathbbm 1\{s\ge\tfrac12\}
  +\tfrac12\mathbbm 1\{s\ge1\},\nonumber\\
x_B(\theta,p)&=\tfrac49s^2,\nonumber\\
x_C(\theta,p)&=
\begin{cases}
\tfrac49s^2,&s<3/4,\\
1,&s\ge3/4.
\end{cases}
\label{eq:randomized-allocation-examples}
\end{align}
Rule \(A\) uses the two cutoff schedules \((1-p)/2\) and \(1-p/2\).
Rule \(B\) varies smoothly in both reports. Rule \(C\) combines that
lottery with an option guaranteeing expansion.

\begin{prop}[Implementation of the randomized examples]
\label{prop:randomized-allocation-examples}
Each rule in \eqref{eq:randomized-allocation-examples} admits a globally BIC
and truthful-path interim IR implementation with bounded state-contingent
transfers. The construction does not depend on the type distribution.
\end{prop}

\begin{proof}
We construct nonnegative convex truthful utilities, then recover transfers
using Proposition \ref{prop:convex-utility-characterization}.

\emph{Step 1: Construct utilities and their supporting slopes.}
For \(s\in[0,3/2]\), define
\begin{align}
u_A(s)&:=\tfrac12(s-\tfrac12)_++\tfrac12(s-1)_+,\nonumber\\
u_B(s)&:=\tfrac4{27}s^3,\nonumber\\
u_C(s)&:=\max\{\tfrac4{27}s^3,\ s-\tfrac{11}{16}\}.
\label{eq:randomized-example-utilities}
\end{align}
All three functions are nonnegative and convex: \(u_A\) is a sum of convex
positive-part functions, \(u_B''(s)=8s/9\ge0\), and \(u_C\) is the maximum
of \(u_B\) and an affine function.

The selected allocation \(x_A\) is the right-hand slope of \(u_A\), including
at its two kinks, and \(u_B'(s)=4s^2/9=x_B\). To check rule \(C\), let
\(h(s):=s-4s^3/27\). Since
\[
h'(s)=1-4s^2/9>0\quad\text{for }0\le s<3/2,
\qquad h(3/4)=11/16,
\]
the affine branch of \(u_C\) is active exactly when \(s\ge3/4\).
At the switch, \(\partial u_C(3/4)=[1/4,1]\), so the selected allocation
\(x_C=1\) is a valid subgradient. Thus, for each \(i\in\{A,B,C\}\),
\(x_i(\theta,p)\in\partial u_i(s)\) and \(0\le x_i\le1\).
Because \(s=\theta+p/2\) is affine, setting
\[
U_i(\theta,p):=u_i(\theta+p/2),
\qquad \beta_i(\theta,p):=\tfrac12x_i(\theta,p)
\]
gives a nonnegative convex utility with
\((x_i,\beta_i)\in\partial U_i\) throughout the type domain. The displayed
allocation and tilt selections are measurable.

\emph{Step 2: Recover payments and verify joint reporting incentives.}
For a report \((\hat\theta,\hat p)\), write
\(\hat s=\hat\theta+\hat p/2\). The transfer-recovery formulas give
\begin{align}
t_i(\hat\theta,\hat p,0)
&=\hat s\,x_i(\hat\theta,\hat p)-u_i(\hat s),\nonumber\\
t_i(\hat\theta,\hat p,1)
&=(\hat s-\tfrac12)x_i(\hat\theta,\hat p)-u_i(\hat s).
\label{eq:randomized-example-transfers}
\end{align}
Indeed, these payments have tilt \(\beta_i=x_i/2\) and truthful expected
transfer \(T_i=\theta x_i-U_i\). More explicitly, a true type with
\(s=\theta+p/2\) obtains from any joint report
\begin{align*}
&\theta x_i(\hat\theta,\hat p)
 -(1-p)t_i(\hat\theta,\hat p,0)-p\,t_i(\hat\theta,\hat p,1)\\
&\qquad=u_i(\hat s)+x_i(\hat\theta,\hat p)(s-\hat s)
\le u_i(s).
\end{align*}
The inequality follows from the subgradient property in Step 1. Truthful
reporting attains \(u_i(s)\ge0\), establishing global BIC and interim IR.

\emph{Step 3: Check boundedness.}
The functions \(u_i\) are continuous on \([0,3/2]\) and hence bounded.
Since \(0\le x_i\le1\), both transfers in
\eqref{eq:randomized-example-transfers} are bounded. They therefore satisfy
both interim integrability and truthful-path absolute integrability for
every type distribution on the stated domain.
\end{proof}

The rules also have a menu interpretation. Rule \(B\) offers expansion
probability \(a\in[0,1]\) at fees \(t_0(a)=a^{3/2}\) and
\(t_1(a)=a^{3/2}-a/2\); the agent maximizes \(as-a^{3/2}\).
Rule \(C\) adds certain expansion at fees \(11/16\) and \(3/16\), giving
payoff \(s-11/16\). These are feasible implementations, not optimality or
least-cost claims.
\end{toappendix}

\section{The Public-Forecast Benchmark}
\label{sec:pointwise-benchmark}

Having calculated the cost of implementing an arbitrary cutoff, we now use that cost to derive the optimal cutoff. We start by considering a benchmark problem: Suppose the agent's $p$ is not private information, but $\theta$ still is. When is the resulting optimal mechanism globally optimal -- even when $p$ is also private information? 

Suppose the firm observes the forecast \(p\) and still has to screen
the manager's stake. It can post an expansion charge \(z\). Managers with
\(\theta\ge z\) accept, so \(z\) is also the cutoff. The firm earns
\(V(p)+z\) per expansion and a fraction \(1-F(z\mid p)\) of managers
accept. This gives the conditional contribution of a cutoff
\(z\in[\underline\theta,\bar\theta]\).\footnote{Here and
below, conditional objects are taken in jointly measurable versions, as
is implicit in the preceding iterated integrals.}
\begin{align}
B(z,p)
&:=
\int_z^{\bar\theta}W(\theta,p)f(\theta\mid p)\,d\theta
\nonumber\\
&=
[V(p)+z][1-F(z\mid p)].
\label{eq:cutoff-virtual-surplus}
\end{align}
The posted-charge expression equals the conditional virtual-surplus
integral by integration by parts, using absolute continuity of
\(F(\cdot\mid p)\). Because the firm observes the forecast, it can
choose a separate cutoff at every \(p\) and set the utility floor to zero.
The resulting rule is our \emph{public-forecast benchmark}. The
\emph{cutoff-choice problem} asks which cutoff maximizes profit once the
firm must elicit the forecast and pay for implementation.

\Needspace{8\baselineskip}
\begin{assump}[Interior pointwise benchmark]
\label{assump:interior-pointwise-cutoff}
There is a measurable function
\(c^0:\mathcal P\to(\underline\theta,\bar\theta)\) such that, for every
belief \(p\in\mathcal P\) and every candidate cutoff
\(z\in[\underline\theta,\bar\theta]\),
\begin{align}
B(c^0(p),p)\ge B(z,p).
\label{eq:interior-pointwise-benchmark}
\end{align}
\end{assump}

At each forecast, the firm chooses its public-forecast cutoff from the
project payoff \(V(p)\) and the conditional stake distribution
\(F(\cdot\mid p)\). Once the forecast is private, the marginal
distribution \(G\) determines the price of bends in that cutoff.

Assumption \ref{assump:interior-pointwise-cutoff} allows nonmonotone virtual values. The reduction to cutoffs
uses the virtual-surplus formulation of \citet{myerson1981optimal} and the
fixed-price argument of \citet[Proposition~1]{riley1983optimal}. For fixed
\(p\), every same-belief IC allocation is, up to null sets, a lottery over
upper-threshold rules. Its virtual surplus averages \(B(z,p)\), so
randomization cannot improve on a maximizing cutoff. A sufficient condition is a unique interior sign crossing of
\(W(\cdot,p)\), from negative to positive. 

For the next results, we also impose the following assumption.
\begin{assump}[Belief-support span]
\label{assump:belief-support-span}
The support of the marginal belief distribution spans the modeled belief
interval:
\begin{align}
\operatorname*{ess\,inf}_{p\sim G}p=\underline p,
\qquad
\operatorname*{ess\,sup}_{p\sim G}p=\bar p.
\label{eq:belief-support-span}
\end{align}
\end{assump}

Assumption \ref{assump:belief-support-span} is weaker than full support:
\(G\) may have gaps, but realized beliefs must approach both endpoints of 
the report interval. 

Under belief-support span, the cost formula implies that a cutoff has
zero implementation rent exactly when it is concave. The next result asks
whether the public-forecast benchmark should be kept once this cost is included.

\begin{toappendix}
\subsection{Zero Rent and Existence of an Optimal Threshold}
\label{app:threshold-attainment-zero-rent}
\begin{restatable}[Zero rent and outer threshold attainment]{prop}{thresholdattainmentzerorent}
\label{prop:threshold-attainment-zero-rent}
Suppose Assumption \ref{assump:belief-support-span} holds. Then, for every weak
cutoff \(c\),
\begin{align}
\operatorname{Rent}(c)=0
\quad\Longleftrightarrow\quad
c\text{ is concave on }\mathcal P.
\label{eq:zero-rent-iff-concavity}
\end{align}
If Assumption \ref{assump:conditional-density} also holds, the cutoff-choice
problem attains its supremum.
\end{restatable}

\begin{proof}[Proof of Proposition \ref{prop:threshold-attainment-zero-rent}]
If \(t\ge\mu\), belief-support span implies \(G\{p>t\}>0\), and hence
\(J_G(t)=\int(p-t)_+\,dG(p)>0\). If \(t<\mu\), it implies
\(G\{p<t\}>0\), and \(J_G(t)=\int(t-p)_+\,dG(p)>0\). Thus \(J_G\) is
strictly positive throughout \(\mathcal P\), whether or not \(G\) has
interior support gaps.

If \(c\) is concave, Corollary \ref{cor:concave-threshold-inner} gives zero
rent. Conversely, suppose \(\operatorname{Rent}(c)=0\). Theorem
\ref{thm:finite-rent-curvature} implies that \(c\) is locally DC and
\[
0
=
\int_{\mathcal P}J_G(t)\,d(D^2c)_+(t).
\]
Strict positivity of \(J_G\) forces \((D^2c)_+=0\), so
\(D^2c\le0\) and \(c\) is concave. This proves
\eqref{eq:zero-rent-iff-concavity}.

It remains to prove outer attainment. We first record the compactness fact used
for this purpose. Let \(r_n\) be nonnegative convex functions
on \(\mathcal P\) such that
\[
\sup_n\int_{\mathcal P} r_n\,dG<\infty.
\]
Assumption
\ref{assump:belief-support-span} implies that \(r_n\) is locally uniformly
bounded and locally equi-Lipschitz. To see this, fix a compact interval
\(K=[k_0,k_1]\subset\mathcal P\). Support span supplies Borel sets
\(A_-\subset(\underline p,k_0)\) and
\(A_+\subset(k_1,\bar p)\), each of positive \(G\)-probability and separated
from \(K\). The integral bound supplies points \(p_n^-\in A_-\) and
\(p_n^+\in A_+\) at which \(r_n\) is uniformly bounded. Convexity then bounds
\(r_n\) uniformly on \(K\). Because \(r_n\ge0\) and both points remain a
positive distance from \(K\), the secant-slope inequalities for convex
functions also give a uniform Lipschitz bound on \(K\). A diagonal
Arzel\`a--Ascoli argument therefore yields a locally uniformly convergent
subsequence.

Joint maximization over feasible pairs \((c,R)\) is equivalent to the
cutoff-choice problem because Theorem \ref{thm:finite-rent-curvature} supplies
a least-cost floor for every finite-rent cutoff. A constant cutoff with zero
floor is feasible. Choose a maximizing sequence whose payoff is bounded below
by the payoff from this pair. Because \(V\) and the cutoff range are bounded,
\eqref{eq:cutoff-virtual-surplus} is uniformly bounded, so the sequence has
uniformly bounded expected floors. For each pair, subtract from the floor a
supporting affine function at \(\mu\). This leaves both convexity constraints
unchanged, preserves nonnegativity, weakly raises the objective, and gives
\(R_n(\mu)=0\) by Bayes plausibility.

Apply the preceding compactness argument to \(R_n\) and to the nonnegative
convex functions
\[
Q_n+\bar\theta=R_n-c_n+\bar\theta.
\]
Their expected values are uniformly bounded because \(c_n\) takes values in
the bounded interval \(\Theta\). Along a common subsequence, \(R_n\) and
\(Q_n\) converge locally uniformly to convex functions \(R\) and \(Q\), and
\(c_n=R_n-Q_n\) converges locally uniformly to a weak cutoff \(c=R-Q\).
Absolute continuity of \(F(\cdot\mid p)\) makes
\(B(\cdot,p)\) continuous; dominated convergence passes virtual surplus to
the limit, while Fatou's lemma weakly lowers expected floor cost. The limit
attains the supremum.
\end{proof}
\end{toappendix}

\Needspace{20\baselineskip}
\subsection{When Is the Public-Forecast Benchmark Globally Optimal?}
\label{subsec:pointwise-cutoff-optimality}


\begin{restatable}[Optimality of the pointwise cutoff]{theorem}{pointwisecutoffoptimality}
\label{thm:pointwise-cutoff-optimality}
Suppose Assumptions \ref{assump:conditional-density},
\ref{assump:interior-pointwise-cutoff}, and
\ref{assump:belief-support-span} hold. Then the following are equivalent:
\begin{align}
x_{c^0}\text{ is the allocation rule of a globally optimal mechanism}
&\quad\Longleftrightarrow\quad
\operatorname{Rent}(c^0)=0
\nonumber\\
&\quad\Longleftrightarrow\quad
c^0\text{ is concave on }\mathcal P.
\label{eq:pointwise-optimality-equivalence}
\end{align}
If \(0<\operatorname{Rent}(c^0)<\infty\), then \(c^0\) is strictly
suboptimal even among threshold mechanisms. If
\(\operatorname{Rent}(c^0)=+\infty\), there is no globally BIC, truthful-path
interim IR implementation of \(x_{c^0}\) with \(H\)-integrable truthful-path
interim transfers; hence \(x_{c^0}\) is outside the domain of the principal's
problem.
\end{restatable}
When the interior benchmark requires positive finite rent, the firm can
lower the cutoff and the floor together, keeping the floor nonnegative
and convex. The difference between the floor and the cutoff stays fixed,
so the expected charge for expansion stays fixed. Managers who still
decline receive a smaller subsidy. Each manager who switches to
expansion adds \(V(p)+c^0(p)\) to profit relative to the old no-expansion
payoff. This gain is nonnegative because the public-forecast firm chose
the old cutoff over shelving the project. Interiority leaves room to
lower the cutoff wherever the floor is positive. A cutoff already at the
lowest stake leaves no such room. The small-stakes example in
Section~\ref{early-capacity-expansion} reaches that bound, and its
optimal adjustment is established separately.



\begin{appendixproof}[Proof of Theorem \ref{thm:pointwise-cutoff-optimality}]
\phantomsection\label{app:proof-pointwise-optimality}
Suppose first that \(\operatorname{Rent}(c^0)=0\). Theorem
\ref{thm:finite-rent-curvature} supplies a zero-cost
implementation. Fix a belief \(p\). Same-belief incentive compatibility makes
every feasible allocation \(x(\cdot,p)\) nondecreasing. Its layer-cake
representation is a lottery over upper-threshold rules, so its conditional
virtual surplus is an average of values \(B(z,p)\). Assumption
\ref{assump:interior-pointwise-cutoff} bounds this average above by
\(B(c^0(p),p)\). Thus \(x_{c^0}\) attains the unrestricted upper bound on
virtual surplus. All implementation rents are nonnegative, so
\eqref{eq:global-benefit-cost} proves global optimality.

Now suppose \(0<\operatorname{Rent}(c^0)<\infty\), and let \(R\) be the
least-cost floor supplied by Theorem
\ref{thm:finite-rent-curvature}. Write \(Q=R-c^0\), and define
\begin{align*}
h(p)&:=\min\{R(p),c^0(p)-\underline\theta\},\\
c_\varepsilon(p)&:=c^0(p)-\varepsilon h(p),
\qquad
R_\varepsilon(p):=R(p)-\varepsilon h(p),
\qquad \varepsilon\in(0,1).
\end{align*}
Then
\[
R-h=\max\{0,Q+\underline\theta\}
\]
is nonnegative and convex. Thus
\(R_\varepsilon=(1-\varepsilon)R+\varepsilon(R-h)\) is nonnegative and
convex, while \(R_\varepsilon-c_\varepsilon=Q\) remains convex. Moreover,
\[
c_\varepsilon-\underline\theta
\ge
(1-\varepsilon)(c^0-\underline\theta)>0,
\qquad
c_\varepsilon\le c^0<\bar\theta.
\]
The perturbed threshold is therefore globally implementable.

Using \eqref{eq:cutoff-virtual-surplus} and
\(c_\varepsilon=c^0-\varepsilon h\), direct algebra gives
\begin{align*}
&B(c_\varepsilon(p),p)-R_\varepsilon(p)
-\bigl[B(c^0(p),p)-R(p)\bigr]\\
&\qquad=
\varepsilon h(p)F(c_\varepsilon(p)\mid p)
+[V(p)+c^0(p)]
\bigl[F(c^0(p)\mid p)-F(c_\varepsilon(p)\mid p)\bigr].
\end{align*}
Pointwise optimality gives
\(B(c^0(p),p)\ge B(\bar\theta,p)=0\). Conditional full support and
interiority give \(1-F(c^0(p)\mid p)>0\), so
\eqref{eq:cutoff-virtual-surplus} gives
\[
V(p)+c^0(p)
=
\frac{B(c^0(p),p)}{1-F(c^0(p)\mid p)}
\ge0.
\]
Because \(R\ge0\) and
\(\int R\,dG=\operatorname{Rent}(c^0)>0\), the set \(\{R>0\}\) has positive
\(G\)-measure. On this set, \(h>0\) and
\(c_\varepsilon>\underline\theta\), so strict positivity of the conditional
density implies \(F(c_\varepsilon\mid p)>0\). The displayed gain is therefore
strictly positive on a set of positive \(G\)-measure. Since
\(R_\varepsilon\) is feasible for \(c_\varepsilon\),
\(\operatorname{Rent}(c_\varepsilon)\le\int R_\varepsilon\,dG\). Integrating
the gain therefore proves that \(c^0\) is strictly suboptimal even within the
threshold class.
If \(\operatorname{Rent}(c^0)=+\infty\), no \(G\)-integrable feasible floor
exists. Because the remaining terms in the envelope transfer formula are
bounded, there is no globally BIC, truthful-path interim IR implementation of
\(x_{c^0}\) with an \(H\)-integrable truthful-path interim transfer. Finally,
Proposition \ref{prop:threshold-attainment-zero-rent} gives the equivalence
between zero rent and concavity.
\end{appendixproof}


\subsection{Computing the Benchmark}

Suppose \(\phi(\cdot,p)\) is continuous and strictly increasing, and \(-V(p)\) lies in the interior of its range. The unique pointwise benchmark then satisfies
\[
c^0(p)=\phi_p^{-1}(-V(p)),
\qquad \phi_p(\theta):=\phi(\theta,p).
\]
If the resulting cutoff is concave in \(p\), Theorem \ref{thm:pointwise-cutoff-optimality} establishes its global optimality. If \(\theta\) and \(p\) are independent, an additional simplification is possible: the same virtual-value function \(\phi\) applies at every belief. If this \(\phi\) is also convex, then its inverse is concave. Because \(-V(p)\) is affine, \(c^0(p)=\phi^{-1}(-V(p))\) is then concave. 

For the \(h=2\) running example, cutoff virtual
surplus is
\begin{align}
B(z,p)
&=\frac12(2p-1+z)(1-z)
\nonumber\\
&=\frac12\bigl\{p^2-[z-(1-p)]^2\bigr\}.
\label{eq:running-example-benchmark-loss}
\end{align}
This is maximized at \(c^0(p)=1-p\). Because the cutoff is affine,
it requires no additional implementation rent. Theorem
\ref{thm:pointwise-cutoff-optimality} therefore establishes global
optimality of the contract introduced in
the \hyperref[ex:capacity-contract]{running example}.

The dependent capacity example in
Appendix~\ref{app:dependent-capacity-expansion} illustrates the other
case of the theorem with an interior benchmark. More optimistic
managers tend to have higher stakes. Relative to the example's
independent counterpart, dependence lowers the public-forecast cutoff
most at intermediate forecasts. The cutoff bends upward and is costly
to keep when the forecast is private. The firm optimally replaces it
with an affine rule that expands for more managers at low and high
forecasts and for fewer in the middle
(Figure~\ref{fig:running-example-adjustment}). The small sacrifice in
virtual surplus eliminates the implementation rent.

\Needspace{6\baselineskip}
\section{Rent-Adjusted Cutoff Choice}
\label{sec:optimal-threshold}

We now choose how far to move from the public-forecast benchmark. A
change in the cutoff can save implementation rent at the cost of virtual
surplus. The principal maximizes the profit from a cutoff \(c\).
\begin{align}
\operatorname{VS}(c)
&=
\int_{\mathcal P}B(c(p),p)\,dG(p),
\label{eq:cutoff-vs}\\
\mathcal J(c)
&:=
\operatorname{VS}(c)-\operatorname{Rent}(c).
\label{eq:cutoff-net-payoff}
\end{align}
The first term values the decisions net of ordinary stake-screening
rents. The second prices the upward bends the firm must support to make
belief reporting truthful.

\subsection{The Rent-Adjusted Cutoff Problem}
\label{subsec:rent-adjusted-global-problem}

We begin by measuring the virtual surplus lost when a candidate cutoff $z$ departs from a pointwise optimum. 
At each forecast \(p\), the public-forecast benchmark chooses
the cutoff that maximizes \(B(z,p)\). Define its payoff and the
loss from choosing another cutoff by
\begin{equation}
\overline B(p):=\max_{z\in[\underline\theta,\bar\theta]}B(z,p),
\qquad
\mathcal L(z,p):=\overline B(p)-B(z,p).
\label{eq:cutoff-loss}
\end{equation}
Thus virtual surplus equals the benchmark payoff minus the loss
from changing approval decisions. Substituting this identity
into the firm's profit gives
\begin{align*}
\mathcal J(c)
&=\operatorname{VS}(c)-\operatorname{Rent}(c)\\
&=\int_{\mathcal P}
  [\overline B(p)-\mathcal L(c(p),p)]\,dG(p)
  -\operatorname{Rent}(c)\\
&=\underbrace{\int_{\mathcal P}\overline B(p)\,dG(p)}
  _{\text{independent of }c}
  -
  \left[
  \int_{\mathcal P}\mathcal L(c(p),p)\,dG(p)
  +\operatorname{Rent}(c)
  \right].
\end{align*}
The benchmark payoff is fixed by the primitives. Maximizing
profit is therefore equivalent to minimizing the sum of the
loss from changing decisions and the additional forecast rent. Substituting the curvature formula
\eqref{eq:curvature-rent-formula} for the rent term yields
\begin{equation}
\min_{\substack{c:\mathcal P\to[\underline\theta,\bar\theta]\\
c\text{ locally DC pointwise}}}
\left\{
\int_{\mathcal P}\mathcal L(c(p),p)\,dG(p)
+
\int_{\mathcal P}J_G(t)\,d(D^2c)_+(t)
\right\}.
\label{eq:rent-adjusted-concavification}
\end{equation}
A divergent rent integral gives an infinite objective. As 
with before, cutoffs that are not locally DC admit no 
finite-valued implementing floor.

Choosing a concave cutoff eliminates additional forecast rent.
The firm may nevertheless prefer to keep some upward curvature
if removing it would sacrifice too much virtual surplus.

We can formulate the same problem by choosing the cutoff \(c\)
and utility floor \(R\) jointly. For a fixed implementable cutoff,
minimizing the expected floor subject to the IC/IR constraints
gives exactly \(\operatorname{Rent}(c)\). Thus
\eqref{eq:rent-adjusted-concavification} is equivalent to
\begin{align}
\min_{c,R}\quad
&
\int_{\mathcal P}\mathcal L(c(p),p)\,dG(p)
+\int_{\mathcal P}R(p)\,dG(p)
\label{eq:rent-adjusted-global-program}\\
\text{subject to}\quad
&
\underline\theta\le c\le\bar\theta,
\qquad R\ge0,
\qquad R\text{ and }R-c\text{ convex on }\mathcal P.
\nonumber
\end{align}
The constraints are the IC/IR conditions established in
Proposition~\ref{prop:threshold-implementability}.
Under Assumptions~\ref{assump:conditional-density} and
\ref{assump:belief-support-span}, the minimum is attained
(Appendix Proposition~\ref{prop:threshold-attainment-zero-rent}).
An optimal pair therefore gives a profit-maximizing cutoff
within the threshold class and a least-cost implementing floor.
Proposition~\ref{prop:threshold-implementability} supplies the
state-contingent payments that complete the mechanism.

An optimal cutoff need not eliminate all forecast rent.
Under the same assumptions, if an optimum requires positive
rent, then on a set of forecasts with positive probability
and a positive floor, either the lowest-stake manager already
expands or expansion has negative total surplus at the
approval margin
(Appendix Proposition~\ref{prop:optimal-adjustment-implications}).
In the first case, the firm cannot lower the cutoff further.
In the second, inducing more expansion can be costly.
These obstacles can prevent the firm from profitably reducing
the floor and cutoff together.

\begin{appendixproof}[Equivalence of the cutoff formulations]
Continuity of \(B(\cdot,p)\) and joint measurability make \(\overline B\)
measurable. Under Assumption \ref{assump:interior-pointwise-cutoff},
\(\overline B(p)=B(c^0(p),p)\).
For every feasible pair \((c,R)\),
\begin{align*}
&\int_{\mathcal P}\mathcal L(c(p),p)\,dG(p)
+\int_{\mathcal P}R(p)\,dG(p)\\
&\qquad=
\int_{\mathcal P}\overline B(p)\,dG(p)
-\operatorname{VS}(c)
+\int_{\mathcal P}R(p)\,dG(p).
\end{align*}
Proposition \ref{prop:threshold-implementability} makes the feasible pairs in
\eqref{eq:rent-adjusted-global-program} exactly the implementable weak
thresholds and their feasible floors. The displayed identity therefore makes
minimization of the program equivalent to maximization of the principal's
payoff within the threshold class. Proposition
\ref{prop:threshold-attainment-zero-rent} gives attainment. Proposition
\ref{prop:threshold-implementability} supplies the tilt and transfers needed
to complete any solution.
\end{appendixproof}

Given an optimal cutoff and floor, Section~\ref{subsec:threshold-inner}
shows how to recover the tilt and state-contingent payments that complete
the mechanism.

\subsection{When Do Thresholds Solve the Unrestricted Problem?}
\label{subsec:threshold-sufficiency}

Under Assumptions~\ref{assump:conditional-density} and
\ref{assump:belief-support-span}, if \(B(z,p)\) is concave
in the cutoff \(z\) for \(G\)-almost every forecast \(p\),
the firm loses no profit by restricting attention to
deterministic cutoff rules. The best cutoff mechanism
therefore maximizes profit among all BIC/IR mechanisms,
including randomized ones
(Appendix Theorem~\ref{thm:threshold-sufficiency}).
Without this concavity condition, \eqref{eq:rent-adjusted-global-program}
still identifies the best cutoff mechanism; its optimality
among all mechanisms requires separate verification.

\vspace{3mm}
\begin{mdframed}[style=runningexamplebox]
\vspace{-5mm}
\subsection{Running Example: Choosing the Cutoff}
\label{early-capacity-expansion}

Return to the \(h=0.01\) capacity example in
Section~\ref{ex:small-stakes-cutoff-cost}.
Section~\ref{ex:small-stakes-cutoff-cost}
establishes that keeping the public-forecast approval
rule can require substantial forecast rents.
In this section, we now show the firm's response:
the firm changes which expansions go ahead because it saves
more on forecast rents than it loses on those decisions.
Let one payoff unit represent \$1 million, so managerial
stakes range from \(-\$5{,}000\) to \(+\$5{,}000\).

Without a bonus cap, keeping the public-forecast rule would
require about \$121,280 in forecast rent (calculated in Section~\ref{ex:small-stakes-cutoff-cost}). As we show below, the optimal
adjustment reduces that rent to about \$3,950 --- while
sacrificing about \$2,900 in profit before forecast rents.
Changing the rule therefore raises expected profit by about
\$114,400, to \$240,663.

The adjusted cutoff rule falls less steeply with the forecast.
Under the adjusted rule, the firm approves some expansions the public-forecast rule
would reject and rejects some it would approve.
Figure~\ref{fig:private-interest-family} plots both
cutoffs. Expansion goes ahead for manager types on or above
the applicable curve. The adjusted cutoff falls through the
stake range over \(19.3\) percentage points of forecasts,
rather than one under the public-forecast benchmark.
\par\medskip
\noindent\begin{minipage}{\linewidth}
\centering
\captionsetup{hypcap=false}
\captionof{figure}{Optimal adjustment of the approval cutoff}
\label{fig:private-interest-family}
\includegraphics[width=0.76\linewidth]{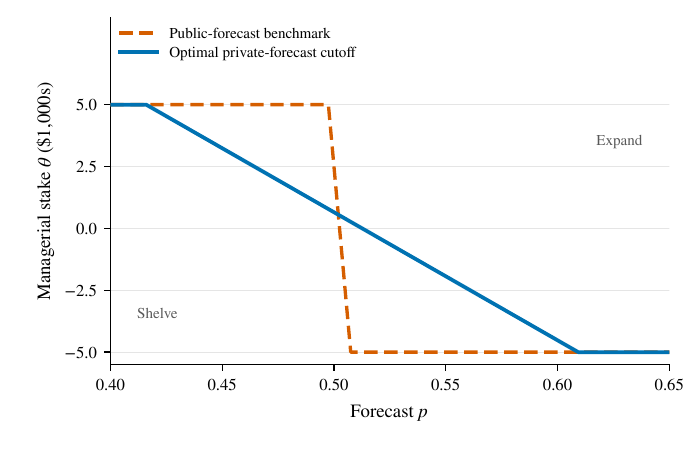}
\end{minipage}
\par\medskip


The firm achieves this flatter cutoff by narrowing 
the gap
between low- and high-demand payments. As a manager 
becomes more optimistic, the expansion price falls.
Once even the
lowest-stake manager accepts, further price reductions give
her rent without changing her decision. The firm cannot
charge away this gain: she could report a lower forecast
and obtain the same offer. The smaller payment gap 
slows price declines; the later kink
delays when forecast rents begin. 

Appendix~\ref{app:private-interest-family} derives the optimal
adjustment and proves its optimality among all BIC/IR
mechanisms, including randomized ones. For reports assigned
to its state-contingent expansion offer, normalized transfers
are
\begin{equation}
t(\hat\theta,\hat p,0)\simeq0.02652,
\qquad
t(\hat\theta,\hat p,1)\simeq-0.02522.
\label{early-expansion-price}
\end{equation}
The manager pays about \$26,500 after low demand and receives
about \$25,200 after high demand, relative to her no-expansion
bonus. She accepts when her stake covers the expected payment
\(T(\hat\theta,\hat p;p)\).
\end{mdframed}

\FloatBarrier

\section{The Cost of Private Stakes and Forecasts}
\label{sec:private-interest-family}

Suppose a manager stands to gain or lose only a little from expansion.
How much would the firm gain by learning exactly what she has at stake?
If the firm already knows her forecast, learning her stake can save
screening rents and improve expansion decisions. Together these gains
add at most \(h\) to expected profit, where \(h\) is the width of the
stake range. If her forecast is private, learning the same stake is
worth substantially more. Under the density bounds below, the gain is
of order \(h^{2/3}\), so it vanishes more slowly than the stake range.

Suppose the firm observes the stake \(\theta\) and charges
the manager exactly \(\theta\) whenever it expands. She then earns zero
under every forecast report. Truthful reporting is optimal for her,
and the firm can use her forecast to approve every expansion with
positive total surplus. Thus, starting with both dimensions private, observing the stake
brings profit to first-best surplus, at most \(h\) above what the
firm earns by observing the forecast instead
(Appendix Corollary~\ref{cor:alternative-public-benchmark}).
We state the proposition for the gain from observing the forecast
because, in the uniform example, we can trace that gain to the
forecast rents and approval distortions the firm avoids. Since
\(h=o(h^{2/3})\), the rate the proposition establishes also applies
to the gain from learning the stake.

Retain the capacity example's expected project payoff \(V(p)=2p-1\).
Let \(Q_h\) be the product of the uniform laws on \([-h/2,h/2]\)
and \((0,1)\), and let \(H_h\) be the joint law of stakes and
forecasts. For the same \(H_h\), write \(\Pi^*_{H_h}\) for the
profit supremum when both dimensions are private and \(\Pi^0_{H_h}\)
for the profit supremum when the firm observes the forecast. Both
problems allow randomized approvals and retain the paper's exact
incentive, interim participation, and transfer requirements.

\begin{restatable}[Distribution-robust cost of private forecasts]{prop}{distributionrobustsmallstakes}
\label{prop:distribution-robust-loss}
Suppose \(0<m\le dH_h/dQ_h\le M<\infty\), \(Q_h\)-almost
everywhere, with \(m,M\) independent of \(h\). Then, uniformly
over these laws for fixed \(m,M\),
\begin{equation}
\Pi^0_{H_h}-\Pi^*_{H_h}=\Theta(h^{2/3})
\qquad\text{as }h\downarrow0.
\label{eq:robust-private-belief-rate}
\end{equation}
More precisely, let \(d(h)\) be the width of the optimal cutoff's
transition in the uniform economy \(Q_h\)
(Appendix~\ref{app:private-interest-family}). For \(0<h<2/3\),
\begin{equation}
m\frac{d(h)^2}{4}-h
\le \Pi^0_{H_h}-\Pi^*_{H_h}
\le M\frac{d(h)^2}{4}.
\label{eq:robust-private-belief-sandwich}
\end{equation}
\end{restatable}

\paragraph{Why the loss is of order \(h^{2/3}\).}
Figure~\ref{fig:private-interest-family} illustrates the tradeoff
in the uniform case. The cutoff falls by \(h\) over a forecast
interval of width \(d\). At the upper end \(t\), even the lowest-stake
manager qualifies, and the cutoff's slope rises from \(-h/d\) to zero.
By \eqref{eq:curvature-rent-formula}, this upward kink costs
\(J_G(t)\,h/d\) in forecast rent. As stakes shrink, the kink approaches
the break-even forecast \(1/2\), where \(J_G(1/2)=1/8\). Widening
the transition therefore saves rent by making the kink smaller.

Relative to the public-forecast rule, the firm approves some expansions
below break-even and forgoes some above it.
At a given stake, the horizontal distance between the two cutoffs in
Figure~\ref{fig:private-interest-family} is the change in the forecast
needed for approval. For a typical stake, this distance is of order
\(d\). With uniform forecasts, the distance is also the probability
that the decision changes. The virtual surplus of expansion is zero
at that stake's public-forecast threshold and varies linearly with
the forecast, so the average loss on changed decisions is also of
order \(d\).
The firm thus balances a
decision loss of order \(d^2\) against forecast rent of order \(h/d\),
as in \eqref{eq:rent-adjusted-concavification}. The optimal width is
of order \(h^{1/3}\), at which both costs are of order \(h^{2/3}\)
(Appendix Proposition~\ref{prop:private-interest-comparative-statics}).

The rate extends to nonuniform and dependent types because each
contract's shortfall from first-best surplus is the sum of two
nonnegative terms at each type: forgone surplus and the manager's
utility.
Changing the joint distribution reweights these losses. For each
contract, its expected shortfall under \(H_h\) lies between \(m\)
and \(M\) times its expected shortfall under the uniform law \(Q_h\).
Since both laws admit the same contracts, minimizing the expected
shortfall preserves these bounds and hence the \(h^{2/3}\) rate
when both dimensions are private. Public-forecast profit is within \(h\) of first-best
surplus, so subtracting that smaller shortfall leaves a profit gain
of order \(h^{2/3}\) from observing the forecast
(Appendix~\ref{app:distribution-robustness}).

\FloatBarrier
\begin{toappendix}
\subsection{Threshold Sufficiency and Unrestricted Optimality}
\label{app:threshold-sufficiency}

This subsection gives a sufficient condition for the best threshold to solve
the unrestricted problem. The result is stronger: every feasible randomized
mechanism can be replaced by a threshold with the same utility floor and
weakly higher payoff. We state the condition and theorem, prove this
replacement, and then apply the strict-concavity conclusion to the running
examples. The following subsection gives conditions on primitives.

Two curvature conditions play distinct
roles. Curvature of the schedule \(p\mapsto c(p)\) across beliefs determines
its implementation rent; curvature of \(z\mapsto B(z,p)\) in the candidate
cutoff, holding \(p\) fixed, determines whether randomizing over cutoffs can
improve virtual surplus. The condition below imposes concavity in this second
dimension.

\begin{assump}[Concavity of cutoff virtual surplus]
\label{assump:cutoff-vs-concavity}
For \(G\)-almost every \(p\), the function
\(z\mapsto B(z,p)\) is concave on
\([\underline\theta,\bar\theta]\).
\end{assump}

At each belief, Assumption \ref{assump:cutoff-vs-concavity} makes the cutoff at
a lottery's mean weakly dominate the lottery in virtual surplus. Uniform \(F\)
satisfies it strictly; see Appendix
\ref{app:cutoff-vs-concavity-primitives} for primitives.
If this assumption fails, Theorem \ref{thm:finite-rent-curvature}, Theorem
\ref{thm:pointwise-cutoff-optimality}, and program \eqref{eq:rent-adjusted-global-program} continue to characterize fixed cutoffs,
the pointwise benchmark, and the best threshold, respectively. What is lost is
the guarantee that the best threshold is optimal among all BIC/IR mechanisms.

\begin{theorem}[Threshold sufficiency]
\label{thm:threshold-sufficiency}
Suppose Assumptions \ref{assump:conditional-density} and
\ref{assump:cutoff-vs-concavity} hold. For every globally BIC/IR mechanism,
there is a deterministic threshold mechanism that is globally BIC/IR,
has the same floor, and gives the principal a weakly higher payoff. Hence the
cutoff-choice and unrestricted problems have the same value.
If Assumption \ref{assump:belief-support-span} also holds, this value is
attained by a deterministic threshold.

If \(B(\cdot,p)\) is strictly concave for \(G\)-almost every \(p\), every
optimal allocation, whenever one exists, coincides \(H\)-almost everywhere
with a threshold. Under Assumption
\ref{assump:belief-support-span}, the optimal cutoff is unique
\(G\)-almost everywhere, although its implementing floor, tilt, and transfers
need not be unique.
\end{theorem}

\begin{proof}
Let \(U\) be the convex truthful utility induced by a globally BIC/IR
mechanism, and write \(R(p)=U(\underline\theta,p)\). Define its
endpoint-matched cutoff by
\begin{align}
c_x(p)
:=
\bar\theta-
\int_{\underline\theta}^{\bar\theta}x(s,p)\,ds.
\label{eq:endpoint-matched-cutoff}
\end{align}
The same-belief envelope gives
\[
U(\bar\theta,p)
=
R(p)+\int_{\underline\theta}^{\bar\theta}x(s,p)\,ds.
\]
Consequently,
\[
R(p)-c_x(p)=U(\bar\theta,p)-\bar\theta.
\]
Both \(R\) and \(R-c_x\) are convex as restrictions of the jointly convex
function \(U\). Proposition \ref{prop:threshold-implementability} therefore
implements \(x_{c_x}\) with
the same floor \(R\).

Fix \(p\). Since \(x(\cdot,p)\) is nondecreasing, its layer-cake
representation is a lottery over upper-threshold rules: there are generalized
cutoffs
\[
c_z(p)
:=
\inf\{s\in\Theta:x(s,p)\ge z\},
\qquad z\in(0,1),
\]
where the infimum of the empty set is \(\bar\theta\). These generalized
inverses are measurable in \((z,p)\), and, up to a Lebesgue-null set,
\[
x(\theta,p)
=
\int_0^1\mathbbm 1\{\theta\ge c_z(p)\}\,dz.
\]
Fubini's theorem gives
\[
c_x(p)=\int_0^1c_z(p)\,dz,
\qquad
\int_\Theta W(\theta,p)f(\theta\mid p)x(\theta,p)\,d\theta
=
\int_0^1B(c_z(p),p)\,dz.
\]
Concavity of \(B(\cdot,p)\) and Jensen's inequality imply that the latter
expression is at most \(B(c_x(p),p)\). Integration over beliefs proves that
the threshold has weakly higher virtual surplus and the same floor cost. This
establishes dominance and equality of the two values. Proposition
\ref{prop:threshold-attainment-zero-rent} supplies an attaining threshold
under belief-support span.

Under strict concavity, equality in Jensen's inequality requires
\(c_z(p)=c_x(p)\) for almost every \(z\), for \(G\)-almost every \(p\).
Conditional absolute continuity then implies that every unrestricted optimum
is a threshold \(H\)-almost everywhere. Finally, \(\operatorname{Rent}\) is
convex: averaging \(\varepsilon\)-optimal floors for two cutoffs gives a
feasible floor for their average, and the claim follows as
\(\varepsilon\downarrow0\). Hence integrated cutoff virtual surplus minus
rent is strictly concave, up to \(G\)-almost-everywhere equality. Any two
optimal cutoffs therefore agree \(G\)-almost everywhere; under belief-support
span, Proposition \ref{prop:threshold-attainment-zero-rent} supplies their
existence.
\end{proof}

The independent uniform family satisfies strict concavity for every
\(h>0\), since \(B_{zz}=-2/h\). Its optimal cutoff is therefore unique
\(G\)-almost everywhere, and every optimal allocation coincides
\(H\)-almost everywhere with that threshold. The early \(h=2\) baseline
has cutoff \(1-p\); Appendix \ref{app:private-interest-family} solves
other stake ranges. The supplementary dependent example in Appendix
\ref{app:dependent-capacity-expansion} also satisfies strict concavity.
\end{toappendix}

\begin{toappendix}
\subsection{Primitive Conditions for Concavity of Cutoff Virtual Surplus}
\label{app:cutoff-vs-concavity-primitives}

Under Assumption \ref{assump:conditional-density}, \(B(\cdot,p)\) is
absolutely continuous, with
\[
B_z(z,p)=-W(z,p)f(z\mid p)
\quad\text{for almost every }z.
\]
Hence Assumption \ref{assump:cutoff-vs-concavity} is equivalent to
\(\theta\mapsto W(\theta,p)f(\theta\mid p)\) having a nondecreasing
almost-everywhere representative. Thus neither monotonicity of
\(W(\cdot,p)\) nor preference--belief independence suffices on its own. If
\(f(\cdot\mid p)\) is absolutely continuous in \(\theta\), then
\[
B_{zz}(\theta,p)
=-2f(\theta\mid p)
-[V(p)+\theta]\partial_\theta f(\theta\mid p)
\]
almost everywhere. Assumption \ref{assump:cutoff-vs-concavity} is therefore
equivalent to
\begin{align}
2f(\theta\mid p)
+[V(p)+\theta]\partial_\theta f(\theta\mid p)
\ge 0
\label{eq:primitive-cutoff-vs-concavity}
\end{align}
for almost every \(\theta\) and \(G\)-almost every \(p\). Because the density
is positive, the same condition can be written as
\[
2+[V(p)+\theta]\partial_\theta\log f(\theta\mid p)\ge0.
\]
These conditions permit nonmonotone densities but rule out slopes that are too
steep relative to total surplus \(V(p)+\theta\). A simple, stronger sufficient
condition is
\[
\left|[V(p)+\theta]\partial_\theta\log f(\theta\mid p)\right|\le2.
\]
\end{toappendix}

\begin{toappendix}
\subsection{Bounds on Gains from Randomization}
\label{app:randomization-bounds}
\label{subsec:concave-envelope-bound}

Program \eqref{eq:rent-adjusted-global-program} identifies the best threshold.
When Assumption \ref{assump:cutoff-vs-concavity} fails, an unrestricted
mechanism may do better. To bound this possible gain, let
\(\widehat B(\cdot,p)\) be the least concave majorant of \(B(\cdot,p)\) on
\([\underline\theta,\bar\theta]\). Equivalently,
\begin{align}
\widehat B(z,p)
:=
\max_{\pi}
\left\{
\int B(y,p)\,d\pi(y):
\int y\,d\pi(y)=z
\right\},
\label{eq:concave-envelope-lottery}
\end{align}
where the maximum is over probability distributions on the cutoff interval.
Thus \(\widehat B(z,p)\) is the largest conditional virtual surplus from a
cutoff lottery with mean \(z\), before imposing cross-belief IC.

Let \(\mathcal V_{\mathrm{thr}}\) and \(\mathcal V_{\mathrm{glob}}\) denote the
values of the cutoff-choice and unrestricted problems, respectively. Define
\begin{align}
\widehat{\mathcal V}
:=
\sup_c
\left\{
\int_{\mathcal P}\widehat B(c(p),p)\,dG(p)
-
\operatorname{Rent}(c)
\right\},
\label{eq:concave-envelope-upper-program}
\end{align}
where the supremum is over measurable weak cutoffs and an infinite
implementation rent gives value \(-\infty\).

\begin{prop}[Concave-envelope bound and contact]
\label{prop:concave-envelope-bound}
Suppose Assumption \ref{assump:conditional-density} holds. Then
\begin{align}
\mathcal V_{\mathrm{thr}}
\le
\mathcal V_{\mathrm{glob}}
\le
\widehat{\mathcal V}.
\label{eq:concave-envelope-value-bound}
\end{align}
If Assumption \ref{assump:belief-support-span} also holds, the supremum in
\eqref{eq:concave-envelope-upper-program} is attained.

Let \(\widehat c\) be any optimizer of
\eqref{eq:concave-envelope-upper-program}, whenever one exists. Every such
optimizer has finite implementation rent. Define its envelope gap by
\begin{align}
\Delta_{\mathrm{env}}(\widehat c)
:=
\int_{\mathcal P}
\bigl[
\widehat B(\widehat c(p),p)-B(\widehat c(p),p)
\bigr]\,dG(p).
\label{eq:concave-envelope-contact-gap}
\end{align}
Then
\begin{align}
0
\le
\mathcal V_{\mathrm{glob}}-\mathcal V_{\mathrm{thr}}
\le
\mathcal V_{\mathrm{glob}}-\mathcal J(\widehat c)
\le
\Delta_{\mathrm{env}}(\widehat c).
\label{eq:concave-envelope-suboptimality-bound}
\end{align}
In particular, if
\begin{align}
\widehat B(\widehat c(p),p)
=
B(\widehat c(p),p)
\qquad\text{for \(G\)-almost every \(p\),}
\label{eq:concave-envelope-contact}
\end{align}
then the deterministic threshold \(x_{\widehat c}\) is globally optimal among
all BIC/IR mechanisms.
\end{prop}

\begin{proof}
We first record the properties of the concave envelope used below. For a
fixed \(p\), the right side of
\eqref{eq:concave-envelope-lottery} is concave in \(z\), because mixtures of
feasible lotteries remain feasible, and it dominates \(B(z,p)\) by admitting
the degenerate lottery at \(z\). Conversely, if \(b\) is any concave majorant
of \(B(\cdot,p)\), Jensen's inequality gives, for every feasible \(\pi\),
\[
b(z)
\ge
\int b(y)\,d\pi(y)
\ge
\int B(y,p)\,d\pi(y).
\]
Taking the maximum over \(\pi\) proves that
\eqref{eq:concave-envelope-lottery} is the least concave majorant.

Because the cutoff interval is compact and \(B(\cdot,p)\) is continuous, the
maximum is attained. In one dimension, it is enough to consider lotteries
supported on at most two cutoffs: a linear objective over the distributions
with a fixed mean attains its maximum at an extreme distribution, whose
support has at most two points. The measurable maximum theorem applied to
this finite-dimensional representation makes \(\widehat B\) jointly
measurable, while the maximum theorem makes \(\widehat B(\cdot,p)\)
continuous. Moreover, \(\widehat B\) is uniformly bounded because \(B\) is
uniformly bounded.

The first inequality in \eqref{eq:concave-envelope-value-bound} follows
because deterministic thresholds are a subset of the unrestricted mechanism
class. To prove the second, fix any globally BIC/IR mechanism with allocation
\(x\), floor \(R\), and truthful utility \(U\). Define its endpoint-matched
cutoff by
\[
c_x(p)
:=
\bar\theta-
\int_{\underline\theta}^{\bar\theta}x(s,p)\,ds.
\]
The same-belief envelope gives
\[
R(p)-c_x(p)=U(\bar\theta,p)-\bar\theta.
\]
Since \(U\) is jointly convex, both \(R(p)=U(\underline\theta,p)\) and
\(R-c_x\) are convex in \(p\). Proposition
\ref{prop:threshold-implementability} therefore implies that \(c_x\) is
implementable with floor \(R\), and hence
\begin{align}
\operatorname{Rent}(c_x)
\le
\int_{\mathcal P}R(p)\,dG(p).
\label{eq:mean-cutoff-rent-bound}
\end{align}

For each fixed \(p\), monotonicity of \(x(\cdot,p)\) gives the layer-cake
representation as a lottery over upper-threshold rules. Specifically, define
the generalized cutoffs
\[
c_a(p):=\inf\{s\in\Theta:x(s,p)\ge a\},
\qquad a\in(0,1),
\]
where the infimum of the empty set is \(\bar\theta\). These generalized
inverses are measurable, and conditional absolute continuity makes cutoff
ties irrelevant. Fubini's theorem then gives
\[
c_x(p)=\int_0^1c_a(p)\,da
\]
and
\[
\int_\Theta W(\theta,p)f(\theta\mid p)x(\theta,p)\,d\theta
=
\int_0^1 B(c_a(p),p)\,da.
\]
The definition of \(\widehat B\) therefore gives
\begin{align}
\int_\Theta W(\theta,p)f(\theta\mid p)x(\theta,p)\,d\theta
\le
\widehat B(c_x(p),p).
\label{eq:concave-envelope-conditional-bound}
\end{align}
Using Proposition \ref{prop:prefvirtualsurplus},
\eqref{eq:mean-cutoff-rent-bound}, and
\eqref{eq:concave-envelope-conditional-bound}, the mechanism's payoff
satisfies
\begin{align*}
\Pi(x,R)
&\le
\int_{\mathcal P}\widehat B(c_x(p),p)\,dG(p)
-
\int_{\mathcal P}R(p)\,dG(p)\\
&\le
\int_{\mathcal P}\widehat B(c_x(p),p)\,dG(p)
-
\operatorname{Rent}(c_x)\\
&\le
\widehat{\mathcal V}.
\end{align*}
Taking the supremum over all globally BIC/IR mechanisms proves the upper
bound in \eqref{eq:concave-envelope-value-bound}.

Suppose now that belief-support span holds. The upper program is equivalent
to joint maximization of
\[
\int_{\mathcal P}\widehat B(c(p),p)\,dG(p)
-
\int_{\mathcal P}R(p)\,dG(p)
\]
over the cutoff bounds and the convexity constraints
\(R\ge0\), \(R\) convex, and \(R-c\) convex. The equivalence follows from
Theorem \ref{thm:finite-rent-curvature}, which supplies a least-cost floor for
each finite-rent cutoff. Since \(\widehat B\) is uniformly bounded,
measurable, and continuous in its first argument, the compactness and
dominated-convergence argument in the proof of Proposition
\ref{prop:threshold-attainment-zero-rent} applies with \(B\) replaced by
\(\widehat B\). It yields an optimizer \(\widehat c\).

Finally, let \(\widehat c\) be any optimizer of the upper program. Since
constant cutoffs have finite value whereas an infinite-rent cutoff has value
\(-\infty\), \(\widehat c\) has finite implementation rent. Hence
\(x_{\widehat c}\) is a feasible threshold mechanism,
\[
\mathcal J(\widehat c)
\le
\mathcal V_{\mathrm{thr}}
\le
\mathcal V_{\mathrm{glob}}
\le
\widehat{\mathcal V}.
\]
Optimality of \(\widehat c\) and
\eqref{eq:concave-envelope-contact-gap} give
\[
\widehat{\mathcal V}
-
\mathcal J(\widehat c)
=
\Delta_{\mathrm{env}}(\widehat c).
\]
Combining the last two displays proves
\eqref{eq:concave-envelope-suboptimality-bound}. Under
\eqref{eq:concave-envelope-contact}, the envelope gap is zero, so all four
payoffs coincide and \(x_{\widehat c}\) is globally optimal.
\end{proof}

The envelope gap in \eqref{eq:concave-envelope-contact-gap} bounds both the
gain from randomization and the unrestricted payoff loss from implementing
\(\widehat c\). When \(\widehat B=B\), the result recovers Theorem
\ref{thm:threshold-sufficiency}; more generally, contact only requires \(B\)
to meet \(\widehat B\) at the optimized cutoff. The bound does not assert that
an envelope-attaining lottery is globally implementable. The next example
shows why.

\begin{example}[The mean cutoff is insufficient]
Let \(\Theta=[0,1]\), \(\mathcal P=(0,1)\), and
\(d(p):=p(1-p)\). At each belief \(p\), consider the equal-probability lottery
over cutoffs \(1/2-d(p)\) and \(1/2+d(p)\). Its mean cutoff is the constant
\(1/2\), which is implementable with the zero floor. The lottery itself is
not implementable with that floor. Indeed, the same-belief envelope implies
that truthful utility at \(\theta=1/2\) would be
\[
U(1/2,p)=\frac{1}{2}d(p)=\frac{1}{2}p(1-p).
\]
This function is strictly concave in \(p\), whereas global BIC requires every
fixed-\(\theta\) slice of truthful utility to be convex by Proposition
\ref{prop:convex-utility-characterization}. Thus implementability of the mean
cutoff does not imply implementability of the specified lottery with the same
floor.
\end{example}

Appendix \ref{app:partial-conjugate-extensions} characterizes the additional
reporting constraints on randomized allocations. Appendix
\ref{app:regional-threshold-optimality} gives an example in which the
concave-envelope bound is strictly above the unrestricted optimum.
\end{toappendix}

\begin{toappendix}
\subsection{A Global-Optimality Certificate}
\label{app:global-optimality-certificate}

A complementary certificate can establish unrestricted optimality directly
for a candidate finite-rent threshold. Let
\(\lambda\in L^1(G)\) satisfy
\begin{align}
\int_{\mathcal P}\lambda(p)\,dG(p)=0,
\qquad
\int_{\mathcal P}p\lambda(p)\,dG(p)=0,
\label{eq:certificate-affine-moments}
\end{align}
and define
\begin{align}
h_\lambda(t)
&:=
\int_{\mathcal P}(p-t)_+\lambda(p)\,dG(p),
\label{eq:certificate-curvature-price}\\
D(\lambda)
&:=
\int_{\mathcal P}
\max_{z\in[\underline\theta,\bar\theta]}
\{B(z,p)-\lambda(p)z\}\,dG(p).
\label{eq:certificate-dual-bound}
\end{align}

The function \(\lambda(p)\) is the shadow value of changing the cutoff at
belief \(p\). The two moment restrictions leave every affine cutoff unpriced,
consistent with implementation rent depending on curvature rather than the
cutoff's level or slope. The potential \(h_\lambda(t)\) aggregates these
belief-specific shadow values into a price on curvature at \(t\). The contact
conditions below are complementary-slackness conditions: the candidate cutoff
must maximize locally adjusted virtual surplus, and its curvature can occur
only where the corresponding global incentive price binds.

\begin{restatable}[Global-optimality certificate]{theorem}{globaloptimalitycertificate}
\label{thm:global-optimality-certificate}
Suppose Assumption \ref{assump:conditional-density} holds. If \(\lambda\)
satisfies \eqref{eq:certificate-affine-moments} and
\begin{align}
0\le h_\lambda(t)\le J_G(t)
\qquad\text{for every }t\in\mathcal P,
\label{eq:certificate-dual-feasibility}
\end{align}
then
\begin{align}
\mathcal V_{\mathrm{glob}}\le D(\lambda).
\label{eq:certificate-global-bound}
\end{align}

Let \(c\) be a finite-rent weak cutoff and write
\(D^2c=\nu_+-\nu_-\) for its curvature decomposition. The threshold
\(x_c\) attains the bound in \eqref{eq:certificate-global-bound}, and is
therefore globally optimal, if
\begin{align}
c(p)
&\in
\argmax_{z\in[\underline\theta,\bar\theta]}
\{B(z,p)-\lambda(p)z\},
\label{eq:certificate-pointwise-contact}
\end{align}
holds for \(G\)-almost every \(p\) and
\begin{align}
\nu_+\bigl(\{t:h_\lambda(t)<J_G(t)\}\bigr)&=0,
&
\nu_-\bigl(\{t:h_\lambda(t)>0\}\bigr)&=0.
\label{eq:certificate-curvature-contact}
\end{align}
In that case,
\(\mathcal J(c)=\mathcal V_{\mathrm{glob}}=D(\lambda)\).
\end{restatable}

\begin{proof}
The proof has two steps. We first establish the integration-by-parts identity
that converts the linear cutoff price into a curvature price. We then combine
that identity with Proposition \ref{prop:concave-envelope-bound} and derive
the equality conditions.

\medskip
\noindent
\emph{Curvature integration-by-parts lemma.}
Let \(\sigma\) be a finite signed measure on the bounded open interval
\(\mathcal P=(\underline p,\bar p)\) such that
\[
\sigma(\mathcal P)=0,
\qquad
\int_{\mathcal P}p\,d\sigma(p)=0,
\]
and set
\[
h(t):=\int_{\mathcal P}(p-t)_+\,d\sigma(p).
\]
If \(c\) is bounded and locally DC, \(\nu:=D^2c\),
\(\int_{\mathcal P}J_G\,d|\nu|<\infty\), and
\(|h|\le J_G\), then
\begin{align}
\int_{\mathcal P}c(p)\,d\sigma(p)
=
\int_{\mathcal P}h(t)\,d\nu(t).
\label{eq:certificate-integration-by-parts}
\end{align}

To prove the lemma, first note that \(D^2h=\sigma\) in the distributional
sense. The two moment conditions also give
\[
h(t)
=
\begin{cases}
\displaystyle\int_{\mathcal P}(t-p)_+\,d\sigma(p),&t<\mu,\\[5pt]
\displaystyle\int_{\mathcal P}(p-t)_+\,d\sigma(p),&t\ge\mu.
\end{cases}
\]
Thus both \(h\) and its one-sided derivatives converge to zero at each open
endpoint. The same endpoint behavior holds for \(J_G\); moreover, \(J_G\) is
nondecreasing to the left of \(\mu\) and nonincreasing to its right.

We next verify that the remaining boundary terms vanish. Fix a
differentiability point \(t_0\in(\underline p,\mu)\). At differentiability points
\(t<t_0\), local bounded variation of \(c'\) gives
\[
|c'(t)|
\le
|c'(t_0)|+|\nu|((t,t_0]).
\]
For any \(\delta\in(\underline p,t_0)\) and \(t<\delta\), monotonicity of
\(J_G\) yields
\begin{align*}
J_G(t)|c'(t)|
&\le
J_G(t)\bigl[|c'(t_0)|+|\nu|([\delta,t_0])\bigr]
+
\int_{(t,\delta)}J_G(s)\,d|\nu|(s).
\end{align*}
Letting \(t\downarrow\underline p\) and then
\(\delta\downarrow\underline p\) shows that
\(J_G(t)|c'(t)|\to0\). The symmetric argument at \(\bar p\) gives the
same conclusion there. Since \(|h|\le J_G\), it follows that
\(h(t)c'(t)\to0\) at both endpoints. Boundedness of \(c\) and the endpoint
limits of \(h'\) also give \(c(t)h'(t)\to0\).

Choose differentiability points
\(a_n\downarrow\underline p\) and
\(b_n\uparrow\bar p\) that are not atoms of \(\sigma\) or \(\nu\).
Stieltjes integration by parts on \([a_n,b_n]\) gives
\[
\int_{(a_n,b_n)}c\,d\sigma
-
\int_{(a_n,b_n)}h\,d\nu
=
\bigl[ch'-hc'\bigr]_{a_n}^{b_n}.
\]
The boundary terms converge to zero by the preceding paragraph. The first
integral converges because \(c\) is bounded and \(\sigma\) is finite; the
second converges absolutely because \(|h|\le J_G\) and
\(\int J_G\,d|\nu|<\infty\). This proves
\eqref{eq:certificate-integration-by-parts}.

We now prove the theorem. Define the signed measure
\(d\sigma=\lambda\,dG\). The moment restrictions in
\eqref{eq:certificate-affine-moments} give the two moment conditions in the
lemma, and its potential is \(h_\lambda\).

Let \(c\) be any finite-rent cutoff and write
\(D^2c=\nu_+-\nu_-\). Theorem \ref{thm:finite-rent-curvature} gives
\[
\int_{\mathcal P}J_G\,d\nu_+<\infty.
\]
Its least-cost floor \(R^c\) is integrable, and
\(Q^c:=R^c-c\) is an integrable convex function with
\(D^2Q^c=\nu_-\). Equation
\eqref{eq:convex-jensen-curvature-identity} therefore gives
\[
\int_{\mathcal P}J_G\,d\nu_-
=
\int_{\mathcal P}Q^c\,dG-Q^c(\mu)
<\infty.
\]
The integration-by-parts lemma and
\eqref{eq:certificate-dual-feasibility} now imply
\begin{align}
\int_{\mathcal P}\lambda(p)c(p)\,dG(p)
&=
\int_{\mathcal P}h_\lambda\,d(\nu_+-\nu_-)
\nonumber\\
&\le
\int_{\mathcal P}J_G\,d\nu_+
=
\operatorname{Rent}(c).
\label{eq:certificate-rent-bound}
\end{align}

Let
\[
M_\lambda(p)
:=
\max_{z\in[\underline\theta,\bar\theta]}
\{B(z,p)-\lambda(p)z\}.
\]
This function is measurable and integrable because \(B\) is jointly
measurable, continuous in \(z\), and uniformly bounded, while
\(\lambda\in L^1(G)\) and the cutoff interval is bounded. For any lottery
\(\pi\) over cutoffs with mean \(z\),
\[
\int B(y,p)\,d\pi(y)
\le
M_\lambda(p)+\lambda(p)z.
\]
Maximizing the left side over such lotteries gives
\begin{align}
\widehat B(z,p)
\le
M_\lambda(p)+\lambda(p)z.
\label{eq:certificate-envelope-bound}
\end{align}
For every finite-rent cutoff, equations
\eqref{eq:certificate-rent-bound} and
\eqref{eq:certificate-envelope-bound} imply
\[
\int_{\mathcal P}\widehat B(c(p),p)\,dG(p)
-\operatorname{Rent}(c)
\le
\int_{\mathcal P}M_\lambda(p)\,dG(p)
=D(\lambda).
\]
The same inequality is immediate for infinite-rent cutoffs. Proposition
\ref{prop:concave-envelope-bound} now proves
\eqref{eq:certificate-global-bound}.

It remains to establish the equality conditions. For any finite-rent cutoff
\(c\), direct substitution of
\eqref{eq:certificate-integration-by-parts} and
\eqref{eq:curvature-rent-formula} gives the exact decomposition
\begin{align}
D(\lambda)-\mathcal J(c)
={}&
\int_{\mathcal P}
\bigl[
M_\lambda(p)-B(c(p),p)+\lambda(p)c(p)
\bigr]\,dG(p)
\nonumber\\
&+
\int_{\mathcal P}[J_G(t)-h_\lambda(t)]\,d\nu_+(t)
+
\int_{\mathcal P}h_\lambda(t)\,d\nu_-(t).
\label{eq:certificate-gap-decomposition}
\end{align}
All three terms are nonnegative. Condition
\eqref{eq:certificate-pointwise-contact} makes the first zero, and the two
support conditions in \eqref{eq:certificate-curvature-contact} make the
remaining terms zero. Hence \(\mathcal J(c)=D(\lambda)\). Since
\(x_c\) is globally feasible and
\(\mathcal V_{\mathrm{glob}}\le D(\lambda)\), it is globally optimal and all
three values coincide.
\end{proof}

The certificate accommodates boundary cutoffs, kinks, and singular curvature,
but is sufficient rather than necessary: the theorem neither asserts that
every optimum admits such a \(\lambda\) nor that minimizing \(D(\lambda)\)
yields strong duality. Appendix \ref{app:positive-rent-soft-concavification}
uses the certificate to construct
an optimum that retains positive curvature; the quadratic application later
uses it to verify the supplementary dependent example's exact optimum.
\end{toappendix}

\begin{toappendix}
\subsection{Exact Uniform Solution}
\label{app:private-interest-family}

The independent uniform family gives an exact optimal contract for every
stake range. Its small-stake solution also anchors the comparison across
type laws in Appendix~\ref{app:distribution-robustness}. Each payoff unit
corresponds to \$1 million in the capacity example. Here,
\(p\sim U(0,1)\), \(\theta\sim U[-h/2,h/2]\), and the two coordinates
are independent. Conditional on them, \(\omega\) is Bernoulli with mean
\(p\). We maintain the paper's commitment, zero interim outside option,
full-domain exact BIC/IR, and interim-integrable transfer requirements.

Write \(\Pi^*(h)\) for the optimal profit with private forecasts and
private stakes, \(\Pi^0(h)\) for the optimal profit with public forecasts
and private stakes, and \(\Pi^B(h)\) for the best profit when the allocation
ignores forecasts. The cost of private forecasts is the gap between
\(\Pi^0(h)\) and \(\Pi^*(h)\). The value of using forecasts is the gain
\(\Delta(h)\) from \(\Pi^B(h)\) to \(\Pi^*(h)\). The action probability
below averages the deterministic decision over private stakes at each
forecast.

\begin{prop}[Independent uniform refinement]
\label{prop:private-interest-comparative-statics}
In this independent uniform family, \(\Delta(h)>0\) and is strictly
decreasing for all \(h>0\). For \(0<h<2/3\), the optimal allocation is a
deterministic threshold whose action probability, averaged over managerial
stakes, rises linearly from zero to one over a belief interval of width
\(d(h)\). As \(h\downarrow0\),
\begin{align}
d(h)&\sim (3h/4)^{1/3},\label{eq:private-interest-transition-rate}\\
\Pi^0(h)-\Pi^*(h)&\sim
\frac14\left(\frac34\right)^{2/3}h^{2/3}.
\label{eq:private-interest-loss-rate}
\end{align}
The optimal implementation rent accounts for two-thirds of this leading
additional cost; the remaining one-third is virtual surplus lost by
changing decisions. These conclusions allow all BIC/IR competitors,
including randomized mechanisms.
\end{prop}

For small stakes, the proof first finds the best linear transition by
choosing its location and width. It then shows that this candidate attains
an upper bound on profit for all feasible mechanisms. The certificate
assigns zero curvature price to the downward kink and matches the rent
weight at the upward kink. The candidate also maximizes the certificate's
adjusted virtual surplus at almost every forecast. The remaining steps
compute the benchmark profits, solve the larger-stake cases, and derive
the small-stake rates.

\Needspace{8\baselineskip}
\begin{proof}[Proof of Proposition
\ref{prop:private-interest-comparative-statics}]
\emph{1. Screening reduction and the exact small-stake candidate.}
For a cutoff \(c\in[-h/2,h/2]\), conditional virtual surplus is
\begin{align}
B_h(c,p)
&=(2p-1+c)\frac{h/2-c}{h},
& \partial_c^2B_h(c,p)&=-\frac2h<0.
\label{eq:private-interest-surplus}
\end{align}
The conditional density is \(1/h>0\), and beliefs have full support.
Theorem \ref{thm:threshold-sufficiency} therefore permits a deterministic
threshold with the same floor and weakly higher payoff than any randomized
BIC/IR mechanism. Its strict-concavity conclusion gives allocation
uniqueness almost everywhere once an optimum is constructed.

Suppose first that \(0<h<2/3\). Define
\[
C(d,h):=4d^3+3hd^2-3h-\frac{h^3}{4}.
\]
The polynomial arises by optimizing over a linear transition.
Consider a cutoff that falls from \(h/2\) to \(-h/2\) over a belief
interval \([t-d,t]\), with \(0<t-d<t<1\) and \(t>1/2\).
Write \(\delta=1-t\) for the length of the right tail. Averaging over
stakes, the expansion probability is \(a(p)=(p-t+d)/d\) on the
transition, zero to the left, and one to the right. Substituting
\(c=h/2-ha\) in \eqref{eq:private-interest-surplus}, the conditional
virtual surplus is \((2p-1)a+h(a/2-a^2)\). Integrating gives
\begin{align*}
\mathbb E[a]&=\delta+d/2,
&
\mathbb E[a^2]&=\delta+d/3,\\
\mathbb E[(2p-1)a]&=d/2-d\delta-d^2/3+\delta-\delta^2.
\end{align*}
The upward kink at \(t\) has slope change \(h/d\). By
\eqref{eq:finite-kink-rent-main}, its rent is \(h\delta^2/(2d)\) and the
downward kink is free. Profit in this family is therefore
\[
\frac d2-\frac{d^2}3-\frac{hd}{12}
+\delta(1-h/2-d)-\delta^2\left(1+\frac h{2d}\right).
\]
At fixed \(d>0\), profit is strictly concave in \(\delta\) and has a
stationary point at \(\delta=d(1-h/2-d)/(2d+h)\). Substituting this value,
the derivative with respect to \(d\) is \(-C(d,h)/(6(2d+h)^2)\).
The root of \(C\) thus gives the candidate within the linear-transition
family. The certificate in step 2 verifies it against all mechanisms.

For \(d>0\), the derivative of \(C\) in \(d\) is
\(12d^2+6hd>0\). Moreover,
\[
C(h,h)=3h\left(\frac{9h^2}{4}-1\right)<0,
\qquad C(1-h/2,h)=4-6h>0.
\]
Consequently there is exactly one root \(d=d(h)\) in
\((h,1-h/2)\). Set
\begin{align}
b&:=\frac hd,
&\delta&:=\frac{d(1-h/2-d)}{h+2d},
&t&:=1-\delta,
&s&:=t-d.
\label{eq:private-interest-parameters}
\end{align}
In particular \(0<b<1\), \(\delta>0\), and
\[
(1-d)-2\delta=\frac{h}{h+2d}>0.
\]
Thus \(s>(1-d)/2>0\) and \(t>1/2\). Our candidate is
\begin{align}
c_h(p)&=
\min\{h,\max\{0,b(t-p)\}\}-\frac h2,
\label{eq:private-interest-cutoff}\\
R_h(p)&=b(p-t)_+.
\label{eq:private-interest-floor}
\end{align}
It has curvature
\(D^2c_h=-b\delta_s+b\delta_t\), where \(\delta_s\) and
\(\delta_t\) denote unit point masses. The notation \(\delta\) without
a subscript continues to mean the tail length in
\eqref{eq:private-interest-parameters}.

\emph{2. A global-optimality certificate.}
The pointwise benchmark can reach the stake bounds in this family, so the
interior-benchmark assumption of Theorem
\ref{thm:pointwise-cutoff-optimality} is not available. We instead verify
Theorem \ref{thm:global-optimality-certificate} directly.

Shift the cutoff by writing \(z=c+h/2\in[0,h]\), and define
\[
\mathcal B_h(z,p)
:=\left(2p-1-\frac h2+z\right)\left(1-\frac zh\right).
\]
This equals \(B_h(c,p)\); an affine shift leaves cutoff curvature and
implementation rent unchanged. For uniform beliefs,
\begin{align}
J_G(p)&=
\begin{cases}
p^2/2,&0<p\le1/2,\\
(1-p)^2/2,&1/2\le p<1.
\end{cases}
\label{eq:private-interest-kernel}
\end{align}
Put
\[
\ell:=\frac{1-h/2-2s}{h},
\qquad \gamma:=\frac{2(1-b)}h>0,
\]
and consider the potential
\begin{align}
\mathcal H(p)&:=
\begin{cases}
0,&p\le s,\\
\ell(p-s)^2/2-\gamma(p-s)^3/6,&s<p<t,\\
(1-p)^2/2,&p\ge t.
\end{cases}
\label{eq:private-interest-potential}
\end{align}
For the middle expression, direct substitution before imposing \(C=0\)
gives the identities
\begin{align}
\mathcal H(t-)-\frac{\delta^2}{2}
&=\frac{d^2C(d,h)}{6h(2d+h)^2},
&\mathcal H'(t-)+\delta&=0.
\label{eq:private-interest-contacts}
\end{align}
Since \(C(d,h)=0\), these contacts make \(\mathcal H\) continuously
differentiable. With \(r:=\delta/d>0\), the two contact equations imply
\begin{align}
\ell&=r(2+3r)>0,
&\gamma d&=6r(1+r),
&\mathcal H''(t-)&=-r(4+3r)<0.
\label{eq:private-interest-contact-coefficients}
\end{align}
The positivity of \(\ell\) also gives
\(s<(1-h/2)/2<1/2\).

We verify \(0\le\mathcal H\le J_G\) on the whole interval. On
\((s,t]\), the ratio \(\mathcal H(p)/(p-s)^2\) decreases to
\(r^2/2>0\), proving nonnegativity there. For \(p\in[1/2,t]\), put
\(u=t-p\in[0,d]\). Expanding the cubic at \(t\) and using
\eqref{eq:private-interest-contacts} gives
\begin{align*}
J_G(p)-\mathcal H(p)
&=u^2\left[\frac{1-\mathcal H''(t-)}2-\frac{\gamma u}{6}\right]\\
&\ge\frac{(1+r)^2u^2}{2}\ge0.
\end{align*}
For \(p\in[s,1/2]\), let \(a:=1-2s>0\). The definition of
\(\delta\) gives
\[
d(1+2r)=1-h(r+1/2)<1.
\]
Hence \(a=2d(1+r)-1<1/(1+2r)\). Since \(s>0\), the function
\((p-s)/p\) is increasing in \(p\), so
\begin{align*}
\frac{\mathcal H(p)}{p^2/2}
&\le\ell\left(\frac{p-s}{p}\right)^2
\le\ell a^2
<\frac{r(2+3r)}{(1+2r)^2}<1.
\end{align*}
The last denominator exceeds its numerator by \((1+r)^2>0\). Outside
\([s,t]\), \eqref{eq:private-interest-potential} gives the required bound.

Let \(\lambda(p)=\mathcal H''(p)\) almost everywhere, taking arbitrary
finite values at \(s\) and \(t\). This is bounded and thus integrable.
Because \(\mathcal H\) and \(\mathcal H'\) vanish at both endpoints,
integration by parts gives
\[
\int_0^1\lambda(p)\,dp=0,
\qquad \int_0^1p\lambda(p)\,dp=0,
\qquad
\int_0^1(p-u)_+\lambda(p)\,dp=\mathcal H(u).
\]
Thus the certificate's affine moment restrictions and curvature-price
bounds hold. On \((s,t)\), differentiation gives
\[
\lambda(p)=\partial_z\mathcal B_h\bigl(b(t-p),p\bigr).
\]
For \(p<s\), \(\lambda(p)=0\) and
\[
\partial_z\mathcal B_h(h,p)
=\frac{1-h/2-2p}{h}\ge\ell>0.
\]
For \(p>t\), \(\lambda(p)=1\), whereas
\[
\partial_z\mathcal B_h(0,p)
\le\partial_z\mathcal B_h(0,t)=\mathcal H''(t-)<0.
\]
Since \(\mathcal B_h\) is strictly concave in \(z\), these derivative
signs prove that the shifted candidate maximizes
\(\mathcal B_h(z,p)-\lambda(p)z\) at every belief except the two kink
points. Maximizers are unchanged by shifting back to \(c\).
Finally, its negative curvature is supported at \(s\), where
\(\mathcal H(s)=0\), and its positive curvature is supported at \(t\),
where \(\mathcal H(t)=J_G(t)\). All contact conditions in Theorem
\ref{thm:global-optimality-certificate} hold. That theorem proves
unrestricted optimality once feasibility and finite rent are verified
below; \eqref{eq:private-interest-surplus} and Theorem
\ref{thm:threshold-sufficiency} give allocation uniqueness almost
everywhere.

\emph{3. Transfers, action probabilities, and payoff.}
The floor and its difference from the cutoff are
\[
R_h(p)=b(p-t)_+,
\qquad R_h(p)-c_h(p)=\max\{-h/2,h/2+b(p-t)\}.
\]
Both are convex, and the floor is nonnegative and bounded. Proposition
\ref{prop:threshold-implementability} therefore gives full-domain BIC/IR.
Its positive curvature is exactly the positive curvature of \(c_h\),
and \(R_h(1/2)=0\). Theorem \ref{thm:finite-rent-curvature} then gives
the minimum expected floor
\begin{align}
\operatorname{Rent}(c_h)
&=\int_0^1R_h(p)\,dp
=\frac{b\delta^2}{2}.
\label{eq:private-interest-rent}
\end{align}

A direct implementation offers no action at price zero and action at the
state-contingent price
\begin{align}
\tau_h(\omega)&=bt-\frac h2-b\omega.
\label{eq:private-interest-action-price}
\end{align}
Also offer action at the constant price \(h/2\). For a true type,
the maximum utility from these options is
\[
\max\{0,\theta-h/2,\theta+h/2+b(p-t)\}
=\max\{0,\theta+h/2+b(p-t)\}.
\]
This equals \(R_h(p)+(\theta-c_h(p))_+\). Assign every truthful report
a maximizing option, so a joint misreport cannot raise utility. To respect
the weak-cutoff convention at the upper stake bound, a report
\((h/2,p)\) with \(p<s\) selects the constant-price action option and
obtains zero utility. Every other action tie can select action at
\eqref{eq:private-interest-action-price}. All no-action reports pay zero.
The payments are bounded for each \(h\), so both the maintained interim
integrability and absolute integrability of realized transfers hold.

Averaging the deterministic threshold over the uniform stake gives
\begin{align}
a_h(p):=\Pr(a=1\mid p)
&=\begin{cases}
0,&p\le s,\\
(p-s)/d,&s<p<t,\\
1,&p\ge t.
\end{cases}
\label{eq:private-interest-action-probability}
\end{align}
The exceptional upper-stake reports have zero conditional probability
and do not affect this calculation. On every positive-probability action
region the principal's expected action payoff, including payment, is
\((2-b)p+bt-1-h/2\). Integrating it against
\eqref{eq:private-interest-action-probability} gives
\[
\Pi_{\mathrm{direct}}
=\frac{d(12+h^2-4d^2)}{24(2d+h)},
\qquad
\Pi_{\mathrm{direct}}-\frac{d^3}{3h}
=-\frac{dC(d,h)}{6h(2d+h)}.
\]
Consequently
\begin{align}
\Pi^*(h)&=\frac{d^3}{3h}
=\frac14+\frac{h^2}{48}-\frac{d^2}{4},
\qquad 0<h<2/3.
\label{eq:private-interest-exact-profit}
\end{align}
The second equality is \(C(d,h)=0\) divided by \(12h>0\).

\emph{4. Public beliefs and the remaining stake ranges.}
Maximization of \eqref{eq:private-interest-surplus} gives the pointwise
cutoff
\begin{align}
c_h^0(p)&=
\min\left\{\frac h2,
\max\left\{-\frac h2,\frac12+\frac h4-p\right\}\right\}.
\label{eq:private-interest-public-cutoff}
\end{align}
For \(0<h<2/3\), it first leaves the upper bound at
\(p_0=1/2-h/4\) and reaches the lower bound at
\(p_1=1/2+3h/4\). Thus the transition width is \(p_1-p_0=h\).
On \((p_0,p_1)\), its conditional virtual surplus is
\((p-p_0)^2/h\), and above \(p_1\) it is \(2p-1-h/2\). Therefore
\begin{align}
\Pi^0(h)
&=\int_{p_0}^{p_1}\frac{(p-p_0)^2}{h}\,dp
+\int_{p_1}^1(2p-1-h/2)\,dp\nonumber\\
&=\frac14-\frac h4+\frac{7h^2}{48},
\qquad 0<h<2/3.
\label{eq:private-interest-public-payoff}
\end{align}
Combining this with \eqref{eq:private-interest-exact-profit} gives the
additional cost of private beliefs,
\begin{align}
L_h:=\Pi^0(h)-\Pi^*(h)
&=\frac{d^2}{4}-\frac h4+\frac{h^2}{8}.
\label{eq:private-interest-additional-cost}
\end{align}
Strict concavity of \(B_h\) makes \(L_h>0\) in this range: otherwise
an optimal mechanism would attain the pointwise upper bound, forcing
\(c_h=c_h^0\) almost everywhere, whereas their transition widths satisfy
\(d>h\).

For \(h\ge2/3\), the pointwise cutoff
\eqref{eq:private-interest-public-cutoff} never reaches the lower stake
bound at an interior belief. It is the minimum of a constant and an
affine function, hence concave. Its zero-floor implementation attains the
pointwise upper bound, so it is globally optimal, including against
randomized mechanisms. For \(2/3\le h\le2\), the only clipping is at
the upper bound below \(p_0=1/2-h/4\), and
\[
\Pi^*(h)=\int_{p_0}^1\frac{(p-p_0)^2}{h}\,dp
=\frac{(h+2)^3}{192h}.
\]
For \(h\ge2\) there is no clipping at an interior belief, giving
\begin{align}
\Pi^*(h)
&=\int_0^1\frac{(p-1/2+h/4)^2}{h}\,dp
=\frac h{16}+\frac1{12h}.
\label{eq:private-interest-large-range-profit}
\end{align}
These formulas agree at \(h=2/3\) and \(h=2\), and with
\eqref{eq:private-interest-exact-profit} at \(d=h=2/3\).
The action contract is then \(1/2+h/4-\omega\), with no-action price
zero; add the constant-price option \(h/2\) when upper-bound action
ties occur. This proves existence for every \(h>0\).

\emph{5. The value of using forecasts.}
If the allocation ignores beliefs, independence and
\(\int_0^1(2p-1)\,dp=0\) make the expected project payoff zero.
Same-belief screening bounds expected revenue by
\(\max_{c\in[-h/2,h/2]}c(h/2-c)/h=h/16\).
The constant cutoff and price \(c=h/4\) attain this bound without
state-contingent payments. Hence
\begin{align}
\Pi^B(h)&=\frac h{16}.
\label{eq:private-interest-blind-payoff}
\end{align}

\phantomsection\label{app:state-independent-benchmark}
The same bound applies if payments are state-independent, even when the
contract can solicit both reports and randomize. Let \(x(m)\) and \(t(m)\)
be the action probability and expected payment for report \(m\), averaging
over any internal lotteries. At a given stake \(\theta\), the utility
\(\theta x(m)-t(m)\) of each report is the same at every belief.
The supremum of these affine functions defines a common, finite convex
indirect utility \(U(\theta)\), attained at each truthful report by
incentive compatibility. At every interior stake where \(U\) is
differentiable, any maximizing report has slope \(x(m)=U'(\theta)\)
because its utility line supports \(U\) there. Its payment is then
\(\theta U'(\theta)-U(\theta)\). The uniform stake law assigns probability
zero to the remaining stakes, so the action probability and expected
payment ignore beliefs almost surely. The bound
\eqref{eq:private-interest-blind-payoff} applies, and the posted price
\(h/4\) attains it with state-independent payments.

For any cutoff, put \(A(p)=(h/2-c(p))/h\in[0,1]\). Its virtual surplus
is
\[
\int_0^1(2p-1)A(p)\,dp
+h\int_0^1\left[\frac{A(p)}2-A(p)^2\right]dp.
\]
For a bounded locally difference-of-convex \(A\), define
\(\mathcal R(A):=\int_0^1J_G\,d(D^2(-A))_+\), allowing infinity.
Theorem \ref{thm:finite-rent-curvature}, affine invariance, and positive
homogeneity of curvature give
\(\operatorname{Rent}(c)=h\mathcal R(A)\). Thus the normalized
finite-rent class \(\mathcal A\) consists of functions
\(A:(0,1)\to[0,1]\) that are locally difference-of-convex with
\(\mathcal R(A)<\infty\); it is independent of \(h\). Therefore
\begin{align}
\Delta(h)
&:=\Pi^*(h)-\Pi^B(h)\nonumber\\
&=\sup_{A\in\mathcal A}
\left\{\int_0^1(2p-1)A(p)\,dp
-h\left[\int_0^1(A(p)-1/4)^2\,dp+\mathcal R(A)\right]\right\}.
\label{eq:private-interest-premium-program}
\end{align}
The affine function \(A_\varepsilon(p)=1/4+\varepsilon(p-1/2)\), for
\(0<\varepsilon<\min\{1/2,2/h\}\), is feasible with zero rent and
delivers
\(\varepsilon/6-h\varepsilon^2/12>0\). Hence \(\Delta(h)>0\).
The constructed optimum attains the supremum. At every optimizer the
coefficient of \(h\) in square brackets is strictly positive: if it
were zero, \(A=1/4\) almost everywhere and the premium would be zero.
For \(0<h_1<h_2\), use an optimizer \(A_{h_2}\) as a feasible
competitor at \(h_1\) to obtain
\begin{align*}
\Delta(h_1)
&\ge\Delta(h_2)+(h_2-h_1)
\left[\int_0^1(A_{h_2}(p)-1/4)^2\,dp+\mathcal R(A_{h_2})\right]\\
&>\Delta(h_2).
\end{align*}
This proves the all-range comparative static. In particular,
\(\Delta(h)=1/(12h)\) for \(h\ge2\).

\emph{6. Signs and small-stake asymptotics.}
Implicit differentiation of \(C(d,h)=0\) gives
\begin{align}
d'(h)&=\frac{3+3h^2/4-3d^2}{12d^2+6hd}>0.
\label{eq:private-interest-width-derivative}
\end{align}
The numerator is positive because \(d<1-h/2\). Moreover,
\(12dd'>h\) is equivalent to
\(6+h^2/2-6d^2-2hd>0\). Replacing \(d\) by its strict upper bound
\(1-h/2\) makes the last left side exceed \(4h>0\).
Consequently \(\Pi^{*\prime}(h)=h/24-dd'/2<0\) for \(0<h<2/3\).
This local profit monotonicity is distinct from the all-range
monotonicity of the forecast premium.

Let \(\alpha=(3/4)^{1/3}\), set \(u=h^{2/3}\), and write
\(d=h^{1/3}f(u)\). The cubic becomes
\[
4f^3+3uf^2=3+\frac{u^3}{4}.
\]
At \((u,f)=(0,\alpha)\), the derivative with respect to \(f\) is
\(12\alpha^2>0\). The implicit-function theorem gives
\(f(u)=\alpha-u/4+O(u^2)\), and hence
\begin{align}
d(h)&=\alpha h^{1/3}-\frac h4+O(h^{5/3}).
\label{eq:private-interest-width-expansion}
\end{align}
Since \(\delta\to1/2\) and \(b\sim h^{2/3}/\alpha\),
\eqref{eq:private-interest-rent} and
\eqref{eq:private-interest-additional-cost} imply
\begin{align}
L_h&\sim\frac{\alpha^2}{4}h^{2/3},
&\operatorname{Rent}(c_h)
&\sim\frac{1}{8\alpha}h^{2/3}
=\frac{\alpha^2}{6}h^{2/3}.
\label{eq:private-interest-loss-expansion}
\end{align}
Because \(\Pi^*(h)=\operatorname{VS}(c_h)-\operatorname{Rent}(c_h)\),
subtraction yields
\begin{align}
\Pi^0(h)-\operatorname{VS}(c_h)
&\sim\frac{\alpha^2}{12}h^{2/3}.
\label{eq:private-interest-loss-decomposition}
\end{align}
These limits give the two-thirds rent share and one-third
virtual-surplus-loss share of the conditional-screening profit gap.
The comparison concerns profit, not total welfare. Together with
\eqref{eq:private-interest-action-probability}, they complete the proof.\qedhere
\end{proof}

\begin{figure}[!htbp]
\centering
\includegraphics[width=0.68\textwidth]{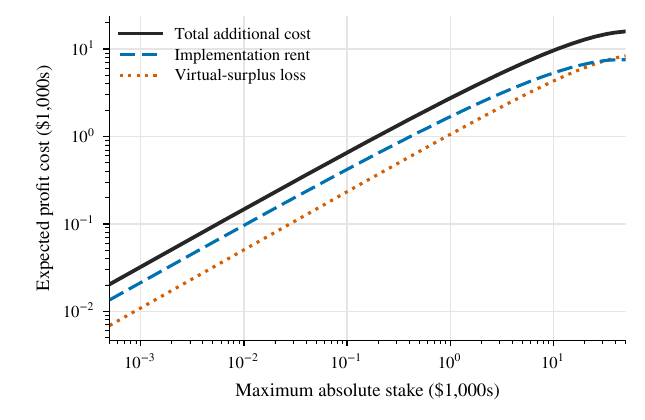}
\caption{The cost of private forecasts as the stake range shrinks}
\label{fig:private-interest-loglog}
\begin{minipage}{0.85\textwidth}
\small\begin{singlespace}
\textbf{Notes.} Both axes use logarithmic scales. The maximum absolute
stake in dollars is \(10^6h/2\), using the project payoffs of \(\pm\$1\)
million in Section~\ref{early-capacity-expansion}. Total expected profit
cost and both components scale as \(h^{2/3}\) as \(h\) tends to zero.
\end{singlespace}\end{minipage}
\end{figure}

\paragraph{A known-stake comparison.}
For \(0<h<2/3\), if the stake is publicly known, truthful-path IR bounds
the principal's payoff by
\(\mathbb E[(2p-1+\theta)_+]\). Taking the action when
\(2p-1+\theta\ge0\) and charging \(\theta\) after action attains the
bound under weak IC, since the expert then obtains zero from every report.
For each \(\theta\in[-h/2,h/2]\), integration over \(p\) gives
\(\mathbb E[(2p-1+\theta)_+\mid\theta]=(1+\theta)^2/4\).
Consequently
\[
\Pi^{\mathrm{known}\,\theta}(h)=\frac14+\frac{h^2}{48},
\qquad
\Pi^{\mathrm{known}\,\theta}(h)-\Pi^*(h)=\frac{d(h)^2}{4}.
\]
This has the same leading loss as \(L_h\); it is a different benchmark.

\paragraph{Exact and strict incentives in the uniform family.}
The amplification concerns exact reporting incentives. To see the
distinction from approximate IC, take action according to the reported
condition \(\hat p\ge1/2\), ignore the stake report, and pay a flat
participation subsidy \(h/2\). Truthful utility is
\(\theta\mathbf1\{p\ge1/2\}+h/2\ge0\). A joint deviation gains at most
\(|\theta|\le h/2\). At the intended truthful reports, the principal
earns \(1/4-h/2\), so the intended-truth cost relative to public beliefs
is only \(O(h)\). This is an approximate-IC performance calculation, not
an equilibrium payoff under exact best responses. With exact best
responses, every positive-stake type induces action and every
negative-stake type prevents it, regardless of belief. Under \(Q_h\),
the expected project payoff is then zero and the principal pays the
subsidy, earning \(-h/2\). The approximate-IC example therefore shows
how the exact-IC requirement affects intended-truth performance; it
does not implement that performance in an exact truthful equilibrium.

Strict incentives for the belief coordinate can instead be added at
smaller-order cost. Add to the agent's payoff \(\varepsilon>0\) times
the bounded state score
\[
S(q,\omega):=(q-1/2)^2+2(q-1/2)(\omega-q).
\]
At true belief \(p\), its expected value is
\((p-1/2)^2-(p-q)^2\). Starting from any exact BIC/IR mechanism, truth
therefore gains \(\varepsilon(p-1/2)^2\ge0\), preserving participation,
and its utility advantage over every joint report with \(q\ne p\)
increases by \(\varepsilon(p-q)^2>0\). The original allocation remains
implementable, and the additional expected payment to the expert is
\(\varepsilon\int_0^1(p-1/2)^2dp=\varepsilon/12\).
Choosing \(\varepsilon=h\) preserves the leading loss rate, since
\(h=o(h^{2/3})\). This added score uses state verification after either
action; the original optimal menu in
\eqref{eq:private-interest-action-price} requires it only after action.
Stake reports within an allocation region may remain indifferent.


\subsection{Distribution-Robust Loss}
\label{app:distribution-robustness}

This subsection proves Proposition~\ref{prop:distribution-robust-loss}.
First-best total surplus bounds optimal profit under either information
regime. The cost of private forecasts is the drop in optimal profit when
the firm must elicit the forecast as well as the stake.

The density bounds let us compare the same contracts across type laws.
Each contract's shortfall from first-best surplus is the sum of its
decision loss and the manager's utility, both nonnegative. The uniform
solution gives a first-best shortfall of \(d(h)^2/4\) for the optimal
private-forecast contract. Public-forecast profit is within \(h\) of
first-best surplus under either law. Subtracting that smaller shortfall
gives the bounds on the cost of private forecasts.

\paragraph{Type laws and benchmarks.}
For \(0<h<2/3\), fix the modeled type domain
\[
\mathcal T_h=[-h/2,h/2]\times(0,1),
\qquad V(p)=2p-1,
\]
and let \(Q_h\) be the product of the uniform laws on its two coordinates.
Let \(H_h\) be any probability law on \(\mathcal T_h\) with a density
\(r_h=dH_h/dQ_h\) satisfying
\begin{equation}
0<m\le r_h(\theta,p)\le M<\infty
\qquad Q_h\text{-almost everywhere},
\label{eq:robust-density-bounds}
\end{equation}
where \(m,M\) do not depend on \(h\). Under every type law, the conditional
state law is \(\omega\mid(\theta,p)\sim\operatorname{Bernoulli}(p)\),
specified at every modeled type. Thus belief--stake dependence may change,
but it does not change how any type evaluates a state-contingent contract.

\Needspace{3\baselineskip}
All mechanisms below are measurable, have \(x:\mathcal T_h\to[0,1]\),
and have finite state payments \(t(\theta,p,\omega)\) at every report.
Write their truthful expected payment and utility as
\[
T(\theta,p)=p\,t(\theta,p,1)+(1-p)t(\theta,p,0),
\qquad U(\theta,p)=\theta x(\theta,p)-T(\theta,p).
\]
Maintain the main paper's full-domain exact joint BIC and interim IR, with
zero outside option, commitment, and the same state-verification and
instrument class. In particular, IC and IR are imposed at every type and
report in \(\mathcal T_h\), including null types. Under a law \(P\),
admissibility additionally requires
\(\int_{\mathcal T_h}|T|\,dP<\infty\). Define
\[
\Pi^*_P:=\sup_{(x,t)\text{ admissible under }P}
\int_{\mathcal T_h}[V(p)x(\theta,p)+T(\theta,p)]\,dP.
\]
The value is a supremum over all allocations, including randomized ones;
its attainment under a general \(H_h\) is not assumed.

Define \(\Pi^0_P\) in the same way when \(p\) is publicly observed.
These mechanisms must be exact IC against every stake misreport
\(\hat\theta\in[-h/2,h/2]\) at each observed \(p\), and IR at every
\((\theta,p)\). There is no belief-reporting constraint because the
posterior is observed. The finite-payment, measurability and truthful
expected-transfer integrability requirements are unchanged. Every joint
BIC/IR mechanism is feasible in this public-posterior problem, so
\(\Pi^*_P\le\Pi^0_P\).

For \(K=\tfrac14(\tfrac34)^{2/3}\), the rate bound also gives
\begin{equation}
mK\ \le\
\liminf_{h\downarrow0}
\frac{\Pi^0_{H_h}-\Pi^*_{H_h}}{h^{2/3}}
\ \le\
\limsup_{h\downarrow0}
\frac{\Pi^0_{H_h}-\Pi^*_{H_h}}{h^{2/3}}
\ \le\ MK.
\label{eq:robust-private-belief-limits}
\end{equation}
Here \(d(h)\in(h,1-h/2)\) is the unique root of
\(4d(h)^3+3h d(h)^2=3h+h^3/4\), as proved above.

\begin{proof}[Proof of Proposition~\ref{prop:distribution-robust-loss}]
\emph{1. The admissible contract class is common to the two laws.}
For any nonnegative measurable function \(f\), the density bounds give
\begin{equation}
m\int f\,dQ_h\le\int f\,dH_h\le M\int f\,dQ_h,
\label{eq:robust-nonnegative-comparison}
\end{equation}
with the inequalities understood also for extended nonnegative integrals.
In particular, applying them to \(f=|T|\) shows
\[
\int|T|\,dH_h<\infty
\quad\Longleftrightarrow\quad
\int|T|\,dQ_h<\infty.
\]
The pointwise IC constraints, pointwise IR, allocation bounds and finite
state-payment requirement do not involve the type distribution. They use
only \((\theta,p)\), the reported type, and the common conditional state
law. Thus the admissible joint-reporting mechanism classes under \(H_h\)
and \(Q_h\) are identical; denote this common class by \(\mathcal M_h\).
The same reasoning equates the two public-posterior classes, whose IC
constraints restrict only stake deviations.

The almost-everywhere density condition does not remove any modeled type
or report. To match the paper's conditional-support convention, choose a
measurable density version bounded between \(m\) and \(M\) throughout
\(\mathcal T_h\), changing it only on a \(Q_h\)-null set. Such a
version is obtained by choosing a jointly Borel Radon--Nikodym version
on the whole domain and clipping its values to \([m,M]\). Its belief marginal density
\[
g_h(p)=\frac1h\int_{-h/2}^{h/2}r_h(\theta,p)\,d\theta
\]
lies in \([m,M]\) at every \(p\). The conditional stake density
\[
f_h(\theta\mid p)=\frac{r_h(\theta,p)}{h\,g_h(p)}
\]
then lies between \(m/(Mh)\) and \(M/(mh)\), integrates to one for
every \(p\), and has the common full support. These are versions of the
distributions induced by \(H_h\). IC and IR remain imposed pointwise on
the full domain, regardless of the chosen versions.

\emph{2. Compare nonnegative losses rather than signed profits.}
For \(P\in\{H_h,Q_h\}\), let
\[
W_P:=\int_{\mathcal T_h}(V(p)+\theta)_+\,dP.
\]
This is finite because \(|V|\le1\) and \(|\theta|\le h/2\). For a
mechanism \(\mu=(x,t)\in\mathcal M_h\), use \(T=\theta x-U\) to write
\begin{align*}
W_P-\Pi_P(\mu)
&=\int_{\mathcal T_h}
\bigl[(V(p)+\theta)_+-(V(p)+\theta)x(\theta,p)+U(\theta,p)\bigr]\,dP\\
&=: \int_{\mathcal T_h}\ell_\mu(\theta,p)\,dP.
\end{align*}
The integrand is nonnegative at every modeled type. Indeed,
\(U\ge0\) by IR; if \(V+\theta\ge0\), its other terms equal
\((V+\theta)(1-x)\ge0\), while if \(V+\theta<0\), they equal
\(-(V+\theta)x\ge0\). It is integrable under both laws because the
surplus terms are bounded and
\(|U|\le h/2+|T|\). The do-nothing contract belongs to
\(\mathcal M_h\), and every admissible profit is at most \(W_P\).
Thus \(0\le\Pi^*_P\le W_P\) and
\[
L_P:=W_P-\Pi^*_P
=\inf_{\mu\in\mathcal M_h}\int\ell_\mu\,dP.
\]
Apply \eqref{eq:robust-nonnegative-comparison} to each
\(\ell_\mu\) and take infima over the same nonempty class:
\begin{equation}
mL_{Q_h}\le L_{H_h}\le ML_{Q_h}.
\label{eq:robust-surplus-loss-comparison}
\end{equation}
For clarity, the upper inequality does not require attainment: an
\(\varepsilon\)-minimizer under \(Q_h\) is admissible under \(H_h\)
and gives \(L_{H_h}\le M(L_{Q_h}+\varepsilon)\), then
\(\varepsilon\downarrow0\). The lower inequality holds for every
mechanism before taking its infimum.

The solved uniform economy gives the exact reference loss. Conditional
on \(\theta\in[-h/2,h/2]\subset(-1,1)\),
\[
\int_0^1(2p-1+\theta)_+\,dp=\frac{(1+\theta)^2}{4}.
\]
Averaging over the centered uniform stake yields
\(W_{Q_h}=1/4+h^2/48\). Equation
\eqref{eq:private-interest-exact-profit} therefore implies
\begin{equation}
L_{Q_h}=W_{Q_h}-\Pi^*_{Q_h}=\frac{d(h)^2}{4}.
\label{eq:robust-uniform-reference-loss}
\end{equation}
The uniform value used here was established against every BIC/IR
mechanism, including randomized allocations. Therefore the comparison
does not introduce a threshold restriction under \(H_h\).

\emph{3. Relate the surplus bound to public-posterior screening.}
IR also gives \(\Pi^0_P\le W_P\), since the same pointwise profit
identity applies to every admissible public-posterior mechanism.
For a lower bound on \(\Pi^0_P\), let the publicly observed \(p\)
determine
\[
a^0(p)=\mathbf1\{V(p)\ge0\},
\qquad
t^0(\hat\theta,p,\omega)=-\frac h2a^0(p).
\]
Both allocation and payment ignore the stake report. Hence stake IC
holds with equality at every \((\theta,p)\). Truthful utility is
\((\theta+h/2)a^0(p)\ge0\), so IR holds pointwise. Payments are
bounded and finite, and the mechanism belongs to the maintained
public-posterior class. Its expected principal payoff is
\[
\int V(p)_+\,dP-\frac h2 P\{V(p)\ge0\}.
\]
Because \(z\mapsto z_+\) is one-Lipschitz and
\(|\theta|\le h/2\),
\[
W_P-\int V(p)_+\,dP\le\frac h2.
\]
Subtracting the payoff of the displayed feasible contract gives
\begin{equation}
0\le W_P-\Pi^0_P
\le\frac h2+\frac h2P\{V(p)\ge0\}\le h.
\label{eq:robust-public-benchmark-gap}
\end{equation}
The posterior in this construction is observed, not reported; the
mechanism is exact IC in its only private report. It is different from
an approximately truthful mechanism that solicits a belief report.

Finally,
\[
\Pi^0_{H_h}-\Pi^*_{H_h}
=L_{H_h}-(W_{H_h}-\Pi^0_{H_h}).
\]
Combining \eqref{eq:robust-surplus-loss-comparison},
\eqref{eq:robust-uniform-reference-loss}, and
\eqref{eq:robust-public-benchmark-gap} proves
\eqref{eq:robust-private-belief-sandwich}. This also shows directly that
the gap is finite and nonnegative, since the public-posterior problem
relaxes the joint-reporting constraints.

\emph{4. Obtain a uniform rate, without claiming a common constant.}
The uniform solution satisfies
\(d(h)\sim(3h/4)^{1/3}\) by
\eqref{eq:private-interest-width-expansion}. Hence
\[
\frac{d(h)^2}{4h^{2/3}}\longrightarrow
K=\frac14\left(\frac34\right)^{2/3}>0,
\qquad \frac h{h^{2/3}}=h^{1/3}\longrightarrow0.
\]
Divide the sandwich by \(h^{2/3}>0\) and take lower and upper limits
to obtain \eqref{eq:robust-private-belief-limits}. In particular, for
all sufficiently small \(h\), uniformly over every allowed \(H_h\),
the gap lies between \((mK/2)h^{2/3}\) and \(2MK h^{2/3}\).
This proves \eqref{eq:robust-private-belief-rate}. The proof compares
values only: it does not identify the optimizer, its curvature, or how
the additional cost divides between rent and decision distortion under
a general \(H_h\).
\end{proof}

\paragraph{Scope and interpretation.}
The bounds require constants \(m,M\) independent of \(h\) and the
maintained domain, full-domain exact IC/IR, zero outside option, state
consistency and instrument class. They allow irregular densities,
nonuniform marginals and dependence, but not concentration or vanishing
density ratios. The result bounds a profit gap relative to public-posterior
screening. Exact coefficients, ramps and rent shares remain uniform-family
conclusions; no generic optimizer or attainment is asserted.

\begin{corollary}[Public stakes as an alternative benchmark]
\label{cor:alternative-public-benchmark}
Maintain the assumptions of
Proposition~\ref{prop:distribution-robust-loss}.
Let \(\Pi^{\theta\mathrm{pub}}_{H_h}\) denote the supremum
expected profit when the firm observes the stake but must
elicit the forecast, retaining the same full-domain weak IC,
interim IR, and payment-integrability requirements.
For every \(0<h<2/3\),
\[
0\le
\Pi^{\theta\mathrm{pub}}_{H_h}-\Pi^0_{H_h}
\le h.
\]
Consequently, uniformly over the stated class of type laws
with fixed density bounds \(m,M\),
\[
\Pi^{\theta\mathrm{pub}}_{H_h}-\Pi^*_{H_h}
=\Theta(h^{2/3})
\qquad\text{as }h\downarrow0.
\]
\end{corollary}

\begin{proof}
Write
\[
W_{H_h}:=\mathbb E_{H_h}[(2p-1+\theta)_+]
\]
for first-best total surplus. For any admissible contract,
the principal's conditional expected payoff is
\[
(2p-1+\theta)x(\theta,p)-U(\theta,p)
\le (2p-1+\theta)_+,
\]
because \(U\ge0\) and \(0\le x\le1\).
Hence \(\Pi^{\theta\mathrm{pub}}_{H_h}\le W_{H_h}\).

When \(\theta\) is observed, offer the allocation and payment
\[
x_\theta(\hat p)
=\mathbf1\{2\hat p-1+\theta\ge0\},
\qquad
t_\theta(\hat p,\omega)=\theta x_\theta(\hat p).
\]
The manager receives zero utility under every forecast report.
The contract therefore satisfies weak IC and interim IR at
every type, and its payments are bounded by \(h/2\) in absolute
value. Under truthful reporting it yields first-best total
surplus to the principal. Thus
\(\Pi^{\theta\mathrm{pub}}_{H_h}=W_{H_h}\).

Equation~\eqref{eq:robust-public-benchmark-gap} now gives
\[
0\le
\Pi^{\theta\mathrm{pub}}_{H_h}-\Pi^0_{H_h}
=W_{H_h}-\Pi^0_{H_h}
\le h.
\]
Finally,
\[
\Pi^{\theta\mathrm{pub}}_{H_h}-\Pi^*_{H_h}
=
\bigl(\Pi^0_{H_h}-\Pi^*_{H_h}\bigr)
+
\bigl(W_{H_h}-\Pi^0_{H_h}\bigr).
\]
The first term is uniformly \(\Theta(h^{2/3})\) by
Proposition~\ref{prop:distribution-robust-loss};
the second lies between zero and \(h=o(h^{2/3})\).
This proves the stated rate.
\end{proof}

\end{toappendix}

\begin{toappendix}
\subsection{Approximate Choice in the Uniform Menu}
\label{app:approximate-response}

Consider the independent uniform economy of Appendix~\ref{app:private-interest-family}
with \(0<h<2/3\). Use \(b=h/d\) and the belief endpoints \(s,t\) from
\eqref{eq:private-interest-parameters}, and write \(C=bt-h/2\) for
the low-demand charge in \eqref{eq:private-interest-action-price}.
The options are complete offers of a decision and a payment. A positive
charge is paid to the firm. At the true type \((\theta,p)\) the expected payoffs are
\[
\begin{array}{@{}lccc@{}}
 &\text{Decision and charge}&\text{Manager utility}&\text{Firm profit}\\ \hline
 N&\text{No expansion, zero}&0&0\\
 S&\text{Expansion, }C-b\omega
  &z=\theta-C+bp&g(p)=(2-b)p+C-1\\
 B&\text{Expansion, }h/2&\theta-h/2&j(p)=2p-1+h/2
\end{array}
\]
Let \(r=h/2-\theta\) be the distance from the upper stake bound, and set
\(x=b(p-s)\). The contingent offer gives utility \(z=x-r\), while the
boundary offer gives \(-r\). The firm earns \(j(p)=g(p)+x\) from the
boundary offer. Since \(-r\le0\), the manager's maximum utility in either
menu is \(U=\max\{0,z\}\).

\Needspace{3\baselineskip}
The two-option menu deletes \(B\). Assign each report its preferred
remaining offer. A truthful report then attains the maximum utility at
every type, and no expansion guarantees zero, so the implementation is
exactly incentive compatible and individually rational throughout the
type space. Both offers can be selected by reports, since \(s>0\) and
\(t<1\). The original weak-cutoff assignment selects \(B\) only at
\(\theta=h/2,\ p<s\), where the two-option menu uniquely selects \(N\).
At all other types the original selection is still available. The
removed boundary and the ties between \(N\) and \(S\) have probability
zero under the uniform law. Thus both menus earn the exact expected profit
\(\Pi^*(h)\) in \eqref{eq:private-interest-exact-profit} under any exact
best responses, including the least favorable tie selection.

At every true type, an admissible response is a measurable choice of
probabilities \(\lambda_i\) over the available offers with
\begin{equation}
 \sum_i\lambda_i u_i\ \ge\ U-\eta,\qquad
 \lambda_i\ge0,\qquad \sum_i\lambda_i=1 .
 \label{eq:approximate-choice-constraint}
\end{equation}
Here \(u_i\) is the manager's expected utility from offer \(i\), and
\(\eta\ge0\). A deterministic response selects one offer. A randomized
response may use any lottery satisfying
\eqref{eq:approximate-choice-constraint}, with the decision and payment
drawn together. The no-expansion option remains available at zero utility.
The response restriction can admit chosen expected utility as low as
\(-\eta\), while exact participation is guaranteed by \(U\ge0\).

Let \(\underline\Pi_{m,R}(h,\eta)\) be the lowest expected firm profit
over these responses to the fixed \(m\)-option menu, for \(m=2,3\) and
\(R=\mathrm{det},\mathrm{mix}\). Write
\(D_{m,R}(h,\eta)=\Pi^*(h)-\underline\Pi_{m,R}(h,\eta)\) for the
nonnegative loss relative to exact best responses. By
\eqref{eq:private-interest-additional-cost} and
\eqref{eq:private-interest-loss-expansion}, the benchmark gap is
\(L_h+D_{m,R}(h,\eta)\), where
\(L_h=\Pi^0(h)-\Pi^*(h)\sim d^2/4\sim Kh^{2/3}\) and
\(K=\tfrac14(3/4)^{2/3}\). Public-belief screening still uses exact
incentives and best responses.

\begin{prop}[Sharp tolerances for preserving the uniform profit gap]
\label{prop:approx-response-robustness}
For each menu and response class, the condition on \(\eta(h)\ge0\)
below is necessary and sufficient for
\(\Pi^0(h)-\underline\Pi_{m,R}(h,\eta(h))\sim Kh^{2/3}\) as
\(h\downarrow0\).
\[
\begin{array}{@{}lcc@{}}
 &\text{Deterministic choices}&\text{Arbitrary lotteries}\\ \hline
 \text{Two options }(N,S)&\eta=o(h)&\eta=o(h^{4/3})\\
 \text{With boundary offer }(N,S,B)&\eta=o(h^{5/3})
   &\eta=o\bigl(h^{5/3}/\log(1/h)\bigr)
\end{array}
\]
These conditions retain the coefficient \(K\) of the exact profit gap.
\end{prop}

\begin{proof}
\emph{1. Utility gaps and adverse choices with two offers.}
If \(z\ne0\), switching away from the exact best choice costs the manager
\(|z|\). A pure switch is feasible when \(|z|\le\eta\). A lottery can
give the worse offer probability up to \(\min\{1,\eta/|z|\}\).
For \(v>0\) write these maximum switch probabilities as
\[
 \kappa_{\mathrm{det}}(\eta,v)=\mathbf 1_{\{v\le\eta\}},\qquad
 \kappa_{\mathrm{mix}}(\eta,v)=\min\{1,\eta/v\}.
\]
A switch hurts the firm precisely when \(g(p)z>0\). If \(z>0\) and
\(g(p)>0\), the switch forgoes profitable expansion. If \(z<0\) and
\(g(p)<0\), it induces unprofitable expansion. Use the largest feasible
switch probability in these cases and the exact choice otherwise. This
gives the exact loss
\begin{equation}
 D_{2,R}(h,\eta)=
 \frac1h\int_0^1\int_0^h |g(p)|\,\mathbf 1_{\{g(p)z>0\}}\,
     \kappa_R(\eta,|z|)\,dr\,dp .
 \label{eq:approximate-choice-two-option-loss}
\end{equation}
The selection is measurable. At ties, choose an exact maximizer with the
lowest firm payoff. Those types have probability zero, so kernel values
at zero do not affect the integral. The selected responses attain the
worst-case profit.

\emph{2. Locating the manager and firm indifference points.}
Set \(\mu:=(s+t)/2\), \(\rho:=b/(4-b^2)\), and \(q:=\mu-\rho\).
We first establish
\[
\mu=\frac{1+b}{2+b},\qquad C=b\mu,\qquad
g(p)=(2-b)(p-q),\qquad
d^2=\frac{12b}{(4-b)(b+2)^2}.
\]
Substituting \(h=bd\) into the definition of \(\delta\) gives
\[
\delta=\frac1{b+2}-\frac d2,
\qquad
\mu=1-\delta-d/2=\frac{1+b}{2+b}.
\]
The original formula for \(\delta\) has a positive numerator since
\(d<1-h/2\), so \(t<1\). Also
\(s=\mu-d/2>(1-d)/2>0\), using \(b>0\) and \(d<1\).
The identity \(C=b(t-d/2)=b\mu\) follows from \(h=bd\). The zero of
\(g\) is
\[
\frac{1-b\mu}{2-b}
=\mu+\frac{1-2\mu}{2-b}
=\mu-\frac b{(b+2)(2-b)}=q.
\]
Substituting \(h=bd\) into the defining cubic for \(d\) in
Appendix~\ref{app:private-interest-family} and dividing by \(d>0\) gives
\(d^2(4+3b-b^3/4)=3b\). Factoring
\(16+12b-b^3=(4-b)(b+2)^2\) proves the displayed formula for \(d^2\).
Because \(0<b<1\),
\[
\frac{4\rho^2}{d^2}
=\frac{b(4-b)}{3(2-b)^2}<1,
\]
where the denominator minus the numerator is
\(4(1-b)(3-b)>0\). Hence
\[
s=\mu-d/2<q<\mu+d/2=t.
\]
Moreover, \(C=bs+h/2>0\) and \(C=bt-h/2<b\). Since \(g\) is
increasing, its endpoint values \(C-1>-1\) and \(1-b+C<1\) imply
\(|g(p)|\le1\) on \([0,1]\). At \(p=s\), \(x=0\), so
\(j(s)=g(s)<0\). Since \(j\) is increasing, \(j(p)<0\) on \(p<s\), and
\[
A_h:=-\int_0^s j(p)\,dp=s(1-h/2-s)>0.
\]
For a given stake, the manager is indifferent between \(N\) and \(S\) at
\[
p_\theta:=s+r/b=t-(\theta+h/2)/b.
\]
As the stake varies, \(p_\theta\) is uniform on \([s,t]\), independent of \(p\), and
\(z=b(p-p_\theta)\). We can integrate over \(p_\theta\) using
\(dr/h=dp_\theta/d\).

\emph{3. Bounding the loss with two offers.}
Put \(w=\eta/b\), the largest belief distance from \(p_\theta\) at which
a pure switch is feasible, and set \(k=2-b\). Write
\(\mathcal D_h(w)=D_{2,\mathrm{det}}(h,bw)\).
When \(d\le1/4\), Step 2 gives \(1/2<\mu<2/3\), \(s\ge3/8\), and
\(t\le19/24\). Therefore every interval
\([p_\theta-w,p_\theta+w]\) lies in \([0,1]\) when \(w\le1/8\). For a fixed
\(p_\theta\ge q\), the firm loses profit from forgone expansion on
\([p_\theta,p_\theta+w]\). Before multiplying by \(k\), the loss is
\[
\int_{p_\theta}^{p_\theta+w}(p-q)\,dp
=w|p_\theta-q|+w^2/2.
\]
Unprofitable expansion can also occur if \(p_\theta-w<q\). Integrating
\(q-p\) on that interval gives \((w-|p_\theta-q|)_+^2/2\). For \(p_\theta<q\),
integration on \([p_\theta-w,p_\theta]\) and on the possible side
\([q,p_\theta+w]\) gives the same expression. Consequently,
\[
\mathcal D_h(w)=\frac k d\int_s^t
\left[w|p_\theta-q|+\frac{w^2}{2}
+\frac{(w-|p_\theta-q|)_+^2}{2}\right]dp_\theta.
\]
The mean absolute distance is
\[
\frac1d\int_s^t|p_\theta-q|\,dp_\theta
=\frac{(q-s)^2+(t-q)^2}{2d}
=\frac d4+\frac{\rho^2}{d}\in[d/4,d/2].
\]
The final summand in the preceding integral lies between zero and
\(w^2/2\). Combining these bounds gives
\begin{equation}
 k\left(\frac{wd}{4}+\frac{w^2}{2}\right)
 \le\mathcal D_h(w)\le k\left(\frac{wd}{2}+w^2\right)
 \qquad(d\le1/4,\ 0\le w\le1/8).
 \label{eq:approximate-choice-pure-bounds}
\end{equation}

A lottery can induce losses even far from indifference. To get a lower
bound for every \(w\ge0\) when \(d\le1/4\), restrict attention to
\(0<p<1/8\). Here
\(p<p_\theta,q\), \(q-p\ge1/4\), and \(0<p_\theta-p\le1\). The feasible
mixed adverse probability is at least \(\min\{w,1\}\), and
\(|g|\ge k/4\). The belief interval has length \(1/8\), and averaging
over \(p_\theta\) preserves the bound. Thus
\[
D_{2,\mathrm{mix}}(h,bw)\ge\frac k{32}\min\{w,1\}.
\]
On this region, choose \(S\) with probability \(\min\{1,w/(p_\theta-p)\}\).
Choose an exact response elsewhere. This lottery delivers at least the
loss in the lower bound.

For the upper bound, put \(v=|p-p_\theta|\). Since \(q,p_\theta\in[s,t]\),
\[
|g|\le k(v+d),\qquad
v\min\{1,w/v\}\le w.
\]
For \(0<w\le1\), integrate separately on either side of \(p_\theta\),
enlarging each interval to length one, to obtain
\[
\begin{aligned}
\int_0^1\min\{1,w/|p-p_\theta|\}\,dp
&\le2\int_0^1\min\{1,w/v\}\,dv\\
&=2w[1+\log(1/w)].
\end{aligned}
\]
Bounding the loss by counting all switches and averaging over
\(p_\theta\) gives
\[
D_{2,\mathrm{mix}}(h,bw)
\le kw\{1+2d[1+\log(1/w)]\}\qquad(0<w\le1).
\]

\emph{4. The sharp two-option tolerances.}
Since \(L_h\sim d^2/4\), preserving its leading coefficient is
equivalent to \(D_{m,R}(h,\eta)=o(d^2)\). If \(w=o(d)\), the
deterministic upper bound in \eqref{eq:approximate-choice-pure-bounds}
gives \(\mathcal D_h(w)=o(d^2)\). Conversely, monotonicity in \(w\) and
the deterministic lower bound, applied at
\(\bar w:=\min\{w,1/8\}\), give
\[
\mathcal D_h(w)\ge k\bar w d/4.
\]
Thus \(\mathcal D_h(w)=o(d^2)\) forces \(\bar w=o(d)\). Eventually
\(\bar w<1/8\), so \(w=\bar w=o(d)\). As \(bw=\eta\) and
\(bd=h\), this is exactly \(\eta=o(h)\).

For mixtures, the lower bound forces \(w=o(d^2)\) whenever
\(D_{2,\mathrm{mix}}(h,\eta)=o(d^2)\), since otherwise \(\min\{w,1\}/d^2\) cannot vanish.
For sufficiency uniformly over arbitrarily small tolerances, on points
where \(w>0\) write \(w=\varepsilon d^2\), with
\(\varepsilon\to0\). The upper bound, divided by \(d^2\), becomes
\[
k\left\{\varepsilon+2\varepsilon d
\left[1+2\log(1/d)+\log(1/\varepsilon)\right]\right\}
\longrightarrow0.
\]
Indeed \(d\log(1/d)\to0\), while
\(\varepsilon\log(1/\varepsilon)\le1/e\) once
\(0<\varepsilon\le1\). Points with \(w=0\) have zero loss. Hence the
mixed criterion is
\[
\eta=o(bd^2)=o(hd)=o(h^{4/3}).
\]

\emph{5. Losses from adding the boundary offer.}
When \(p<s\), one has \(x<0\). Option \(B\) has utility
\(-r>x-r=u_S\) and principal payoff \(j=g+x<g\). Replacing any
\(S\)-weight by \(B\)-weight increases utility and decreases principal
payoff. The exact choice is \(N\) almost everywhere on this region,
and \(j<0\). The adverse deterministic response is therefore \(B\) when
\(r\le\eta\) and \(N\) otherwise. The adverse mixed response places
weight \(\min\{1,\eta/r\}\) on \(B\), with the remainder on \(N\).
At \(r=0\), select \(B\); this null set does not affect the integral.
The average switch probability is \(F_R(\eta/h)\), where for \(u\ge0\)
\[
 F_{\mathrm{det}}(u)=\min\{u,1\},\qquad
 F_{\mathrm{mix}}(u)=
 \begin{cases}
 0,&u=0,\\
 u[1+\log(1/u)],&0<u<1,\\
 1,&u\ge1.
 \end{cases}
\]
The logarithm comes from lotteries that assign some weight to the
boundary offer even for \(r>\eta\). Splitting the integral at \(r=\eta\)
gives
\[
\frac1h\int_0^h\min\{1,\eta/r\}\,dr
=\frac\eta h[1+\log(h/\eta)]
\qquad(0<\eta<h).
\]
Multiplication by
\(-\int_0^s j(p)\,dp=A_h\) yields the first term of the decomposition.

When \(p\ge s\), \(x\ge0\), and \(S\) instead has weakly higher
utility and weakly lower principal payoff than \(B\). Replacing
\(B\)-weight by \(S\)-weight preserves response feasibility and weakly
decreases profit. Thus the adverse problem reduces exactly to \(N,S\)
on this region, giving the part \(T_R\) of the integral in
\eqref{eq:approximate-choice-two-option-loss} restricted to \(p\ge s\).
Combining the two measurable minimizers gives the exact decomposition
\begin{equation}
 D_{3,R}(h,\eta)=A_hF_R(\eta/h)+T_R.
 \label{eq:approximate-choice-three-option-loss}
\end{equation}

\emph{6. The sharp tolerances with the boundary offer.}
At each fixed \(r\), the set \(\{|z|\le\eta\}\) has belief length
at most \(2\eta/b\), and \(|g|\le1\). For mixtures, the separate-side
integral used in Step 3 bounds the belief integral of the kernel by
\(2F_{\mathrm{mix}}(\eta/b)\). If \(\eta/b\ge1\), this bound follows
also from the kernel being at most one. Also \(T_R\le1\) since \(|g|\) and the switch probability are bounded by one.
These observations give
\[
0\le T_R\le\min\{1,2F_R(\eta/b)\}.
\]
As \(h\downarrow0\), the small-stake expansion and Step 2 imply
\(s\to1/2\), so \(A_h\to1/4\). For \(0<\eta\le h\), put
\(u:=\eta/h\in(0,1]\). Since \(\eta/b=du\),
\[
\frac{T_{\mathrm{det}}}{u}\le2d,
\]
and
\[
\begin{aligned}
\frac{T_{\mathrm{mix}}}{F_{\mathrm{mix}}(u)}
&\le2d\left[1+
\frac{\log(1/d)}{1+\log(1/u)}\right]\\
&\le2d[1+\log(1/d)]\longrightarrow0.
\end{aligned}
\]
The bounds do not depend on \(u\), proving uniformly over
\(0<\eta\le h\) that
\[
D_{3,R}(h,\eta)=\left(\frac14+o(1)\right)F_R(\eta/h).
\]
If \(\eta\ge h\) along any sequence \(h\downarrow0\), the exact
decomposition gives \(D_{3,R}\ge A_h\to1/4\), so preserving the
leading additional cost is impossible. On \(\eta<h\), the
deterministic criterion is therefore
\[
\eta/h=o(h^{2/3}),\qquad\text{equivalently }\eta=o(h^{5/3}).
\]

For mixtures, preservation is equivalent to
\(F_{\mathrm{mix}}(u)=o(v)\), where \(v:=h^{2/3}\). We show this is
equivalent to \(u=o(v/\log(1/v))\). For necessity,
\(F_{\mathrm{mix}}(u)=o(v)\) first forces \(u<v\) eventually, since
\(F_{\mathrm{mix}}(u)\ge u\) for \(0\le u\le1\). Hence
\[
\frac{u\log(1/v)}v
\le\frac{F_{\mathrm{mix}}(u)}v\longrightarrow0.
\]
For sufficiency, on points where \(u>0\), write
\(u=vy/\log(1/v)\) with \(y\to0\). Direct expansion gives
\[
\frac{F_{\mathrm{mix}}(u)}v
=\frac y{\log(1/v)}
\left[1+\log(1/v)+\log\log(1/v)+\log(1/y)\right]
\longrightarrow0,
\]
using \(y\log(1/y)\le1/e\) for all sufficiently small \(y>0\).
Since \(\log(1/v)=(2/3)\log(1/h)\), the equivalent criterion is
\[
\eta=o\!\left(h^{5/3}/\log(1/h)\right).
\]
When \(\eta=0\), both losses are zero by
\eqref{eq:approximate-choice-two-option-loss} and
\eqref{eq:approximate-choice-three-option-loss}. This covers
zero-tolerance subsequences as well.
\end{proof}

\paragraph{The role of opting out.}
For \(p<s\), the adverse boundary lottery gives the manager utility
\(-\min\{r,\eta\}\). If a chosen response had to give at least the
zero payoff from opting out, both expansion offers would be infeasible
there for \(r>0\), even as components of a lottery with no expansion.
The low-forecast boundary loss in
\eqref{eq:approximate-choice-three-option-loss} relies on the expected
utility tolerance in \eqref{eq:approximate-choice-constraint}.

Appendix~\ref{app:private-interest-family} also studies intended-truth
payoffs under approximate IC and a scoring addition that makes belief
reporting strictly optimal. Strict belief incentives establish that
misreporting the belief lowers utility. The profit calculation under
approximate choice also needs the size of the utility gaps and the firm's
payoff from switching offers.
\end{toappendix}

\begin{toappendix}
\Needspace{8\baselineskip}
\subsection{Constructing an Optimal Cutoff}
\label{app:constructive-cutoff-selection}

The rent formula turns cutoff choice on a finite grid into a dynamic
program: its last two cutoff values determine the next slope change and its
cost. Refining the grid controls payoff loss.

\Needspace{3\baselineskip}
To select one optimum when several tie, the construction also narrows
allowable cutoff values at a fixed dense sequence of belief reports. It
prefers the lower interval when the payoff loss is small. Each stage starts
afresh, allowing earlier choices to be revised. Every fixed finite sequence
of choices eventually stabilizes and together these choices identify one
optimal cutoff.

Maintain Assumptions \ref{assump:conditional-density} and
\ref{assump:belief-support-span}. Write
\(L=\underline\theta\), \(H=\bar\theta\),
\(I=\mathcal P=(\ell,r)\), and \(W=H-L>0\).
The law \(G\) may be atomic, singular, atomless, or mixed, with support gaps.
Use the full contribution
\(B(z,p)=(V(p)+z)(1-F(z\mid p))\), and fix a known positive bound
\(M\ge\sup_{z\in[L,H],p\in I}|B(z,p)|\); for example, take
\(M=1+\max_{p\in[\ell,r]}|V(p)|+\max\{|L|,|H|\}\).
Let \(\mu=\int_Ip\,dG(p)\), and write
\[
v^*=\sup_c\mathcal J(c),
\]
where the supremum is over finite-rent weak cutoffs and \(\mathcal J\) is
defined in \eqref{eq:cutoff-net-payoff}. By Proposition
\ref{prop:threshold-implementability}, feasibility means that a finite
floor \(R\ge0\) and \(Q=R-c\) are convex on all of \(I\), with
\(L\le c\le H\) and \(\int_I R\,dG<\infty\).

The construction uses exact real arithmetic and comparisons, integer
searches, exact \(G\)-integrals of the displayed bounded measurable
functions, and the following compact stake maxima:
\begin{equation}
\Omega(h)=\int_I\max_{\substack{z,z'\in[L,H]\\|z-z'|\le h}}
 |B(z,p)-B(z',p)|\,dG(p),\qquad h\ge0.
\label{eq:constructive-modulus}
\end{equation}
Stake continuity makes this a measurable, nondecreasing function with
\(\Omega(h)\to0\) as \(h\downarrow0\), by dominated convergence.
Exact evaluation of these operations is assumed. Numerical integration and
maximization require separate error bounds.

\begin{theorem}[Constructive cutoff selection]
\label{thm:constructive-cutoff-selection}
The construction below terminates at every integer stage \(k\ge1\), using
at most \(k+1\) finite conditional dynamic programs. It returns a cutoff
\(c_k\), its canonical implementation, and finite numbers
\(P_k=\mathcal J(c_k)\) and \(U_k\) satisfying
\begin{equation}
 P_k\le v^*\le U_k,
 \qquad 0\le U_k-P_k\le (k+1/8)2^{-k}.
\label{eq:constructive-bracket}
\end{equation}
The entire sequence \(c_k\) converges uniformly on every compact subinterval
of \(I\) to one cutoff \(c^*\) with \(\mathcal J(c^*)=v^*\). Its canonical floor and payments
are recovered from \(c^*\) by the formulas below. Every positive payoff
tolerance is certified at a finite stage.
\end{theorem}

Exact optimality is obtained through a limit; no practical runtime bound is
claimed.

\begin{proof}
\emph{Closed bands and their schedule.}
Enumerate the normalized interior dyadics
\(\ell+(r-\ell)j/2^n\), with odd \(j\) satisfying \(1\le j<2^n\),
by increasing \(n\) and then increasing \(j\), as
\(\rho_1,\rho_2,\ldots\).
In sweep \(n\), bisect the current bands at
\(\rho_1,\ldots,\rho_n\), in that order. Introduce each report with band
\([L,H]\) inside its first split. At depth
\(d=n(n-1)/2+i\), \(1\le i\le n\), the split report is \(\rho_i\).
A band \([\alpha,\beta]\) has ordered children
\([\alpha,\gamma]\) and \([\gamma,\beta]\), where
\(\gamma=(\alpha+\beta)/2\). Every coordinate is bisected infinitely
often. A finite node \(\sigma\) specifies the current closed bands;
shared boundaries belong to both children. Set
\(S_k=\{\rho_1,\ldots,\rho_k\}\), which contains all reports named in
the first \(k\) splits.

\emph{One common grid at stage \(k\).}
Set \(\lambda_k=2^{-k}\) and \(e_k=2^{-k-3}\).
Search \(h=4,8,16,\ldots\), putting
\[
 a=\ell+(\mu-\ell)/h,\qquad b=r-(r-\mu)/h,
 \qquad T=(\ell,a)\cup(b,r),
\]
until \(S_k\subset(a,b)\) and
\((2M+W)G(T)\le e_k/3\). With this core fixed, put
\begin{gather*}
 a_-=(\ell+a)/2,\qquad b_+=(b+r)/2,\\
 m_* =\min\{G((\ell,a_-]),G([b_+,r))\},\qquad
 d_* =\min\{a-a_-,b_+-b\},\\
 S_R=\frac{M}{m_*d_*},\qquad
 S_Q=\frac{M+W}{m_*d_*},\qquad S_c=S_R+S_Q.
\end{gather*}
Support span makes both denominators positive. Start at the integer level
\(\nu_0=\max\{1,1+\log_2h,\text{largest dyadic level in }S_k\}\).
For successive \(d\ge\nu_0\), sort and exactly deduplicate
\(\{a,b,\mu\}\cup S_k\) and all global lattice points
\(\ell+(r-\ell)j/2^d\) strictly between \(a\) and \(b\).
Write the resulting mesh as \(a=p_0<\cdots<p_N=b\), and stop when,
for \(\Delta=\max_i(p_{i+1}-p_i)\),
\[
 S_R\Delta/2+\Omega(S_c\Delta/2)\le e_k/3.
\]
Let \(d_i=p_{i+1}-p_i\), and compute
\[
 \Lambda=\sum_{i=0}^{N-1}
 \frac{J_G(p_i)+J_G(p_{i+1})}{d_i}.
\]
Starting at \(m=1\) and doubling, stop when \(\eta=W/m\le2^{-k}\) and
\(\Omega(\eta)+\eta\Lambda\le e_k/3\). Use stake levels
\(z_j=L+j\eta\), \(0\le j\le m\), and record
\begin{equation}
 E_k=(2M+W)G(T)+S_R\Delta/2+\Omega(S_c\Delta/2)
                 +\Omega(\eta)+\eta\Lambda\le e_k.
\label{eq:constructive-grid-error}
\end{equation}
Every search terminates: the open tails decrease to the empty set, a
finite set of interior reports eventually lies in the core, the largest
mesh gap tends to zero, and then \(\Lambda<\infty\) is fixed while
\(\eta\to0\). Thus \(\eta_k\le2^{-k}\). The same mesh, stake grid, and
\(E_k\) serve every conditional problem at this stage.

\emph{Finite conditional optimization.}
Candidates are the polygons through grid stakes at the mesh reports,
extended constantly on both open tails. At a node \(\sigma\) through depth \(k\), relax every
band \([\alpha,\beta]\) by the same \(\eta=\eta_k\), allowing
\[
 [\alpha-\eta,\beta+\eta]\cap[L,H]
\]
at its named report. These reports are mesh nodes. Let \(D_i\) be the
allowed stake indices there; at other nodes all indices are allowed.
Each \(D_i\) is an explicit integer interval, possibly empty.
For nodal values \(x_i\), set \(s_i=(x_{i+1}-x_i)/d_i\). Theorem
\ref{thm:finite-rent-curvature} gives their exact rent as
\begin{equation}
 J_G(a)(s_0)_+
 +\sum_{i=1}^{N-1}J_G(p_i)(s_i-s_{i-1})_+
 +J_G(b)(-s_{N-1})_+.
\label{eq:constructive-polygon-rent}
\end{equation}
Both boundary terms are curvature impulses at interior reports, including
when \(G\{a\}=G\{b\}=0\). All these polygons have finite canonical floors.

Precompute the full rewards
\begin{align*}
 L_a(x)&=\int_{(\ell,a)} B(x,p)\,dG(p),&
 R_b(x)&=\int_{[b,r)} B(x,p)\,dG(p),\\
 T_i(x,y)&=\int_{[p_i,p_{i+1})}
 B\left(x+\frac{y-x}{d_i}(p-p_i),p\right)dG(p).
\end{align*}
These sets partition \(I\); each atom enters exactly once, including at
every mesh node. Initialize the last-two-values table by
\[
 \mathsf D^1(j,j')=L_a(z_j)+T_0(z_j,z_{j'})
       -J_G(p_0)\left(\frac{z_{j'}-z_j}{d_0}\right)_+
 \quad(j\in D_0,\ j'\in D_1).
\]
All disallowed entries are \(-\infty\). For \(1\le i<N\), set
\begin{align}
 \mathsf D^{i+1}(j',j'')={}&T_i(z_{j'},z_{j''})
 +\max_{j\in D_{i-1}}\left\{\mathsf D^i(j,j')\right.
 \nonumber\\[-2pt]
 &\left.\hspace{8mm}-J_G(p_i)
 \left(\frac{z_{j''}-z_{j'}}{d_i}
             -\frac{z_{j'}-z_j}{d_{i-1}}\right)_+\right\},
 \quad j'\in D_i,\ j''\in D_{i+1}.
\label{eq:constructive-dp-recursion}
\end{align}
Terminate with
\[
 V_k(\sigma)=\max_{\substack{j\in D_{N-1}\\j'\in D_N}}
 \left\{\mathsf D^N(j,j')+R_b(z_{j'})
 -J_G(p_N)\left(-\frac{z_{j'}-z_j}{d_{N-1}}\right)_+\right\}.
\]
For \(N=1\), initialization and termination apply directly. Empty problems
return \(-\infty\) and no witness. Choose the smallest maximizing
predecessor index at every tie, and the first terminal pair in index order;
backtracking specifies a deterministic maximizing polygon.

Indeed, \(\mathsf D^i(j,j')\) is the largest reward on the left tail and
cells \(0,\ldots,i-1\), less the curvature charges at
\(p_0,\ldots,p_{i-1}\), among allowed prefixes with these last two values.
Initialization establishes this assertion. Appending one value adds
exactly its cell reward and its new slope-jump charge, which depend on the
past only through the last two values. This proves the assertion by
induction; termination adds the right tail and the final impulse in
\eqref{eq:constructive-polygon-rent}. Thus \(V_k(\sigma)\) is the exact
maximum of \(\mathcal J\) over its finite candidate set, denoted
\(\mathcal C_k(\sigma)\).

\emph{The temporary path.}
At the unrestricted root compute \(V_k^0\) and its maximizing witness
\(w\), and set \(f=V_k^0\). The constant \(H\) is present and has zero
payoff, so \(f\ge0\). At each of the first \(k\) scheduled splits, solve
only the lower child's problem, obtaining \(f_0,w_0\). If
\(f_0\ge f-\lambda_k\), including equality, take that child and replace
\((f,w)\) by \((f_0,w_0)\). Otherwise take the upper child and retain
\((f,w)\). An empty lower child cannot pass this finite threshold.
The rule is valid because, at the split report,
\[
 [\alpha-\eta,\beta+\eta]
 =[\alpha-\eta,\gamma+\eta]\cup[\gamma-\eta,\beta+\eta].
\]
Intersecting with \([L,H]\), the grid, and the unchanged other constraints
gives \(\mathcal C_k(\sigma)=\mathcal C_k(\sigma0)
\cup\mathcal C_k(\sigma1)\). On rejection, \(f_0<f\), so the parent
maximizer cannot be in the lower child. It remains an upper-child
maximizer. Thus the stored value is always the current exact maximum,
without an upper-child solve. After \(k\) splits output
\[
 c_k=w,\qquad P_k=f,\qquad U_k=V_k^0+E_k.
\]
Each stage restarts at the unrestricted root; earlier choices are not
locked in. The stage uses one root solve and at most \(k\) lower solves.

\emph{The common approximation bound.}
We verify that every finite-rent cutoff \(c\) with \(\mathcal J(c)\ge0\) has
a stage-grid approximation losing at most \(E_k\), compatible with the
relaxed bands of every node containing \(c\). Write \(R=R^c\) and
\(Q=R-c\). Canonicalization and the reward bound imply
\[
 \int R\,dG\le M,\qquad Y=Q+H\ge0,
 \qquad\int Y\,dG\le M+W.
\]
For a nonnegative convex function with integral at most \(C\), each guard
set \((\ell,a_-]\), \([b_+,r)\) contains a point of value at most
\(C/m_*\). Their chord bounds its core values by \(C/m_*\).
Supporting-line inequalities and nonnegativity then bound every core
supporting slope in absolute value by \(C/(m_*d_*)\).
Consequently \(R,Q,c\) have core Lipschitz constants \(S_R,S_Q,S_c\).

At \(e=a,b\), choose supporting slopes \(u_e\) of \(R\) and \(v_e\)
of \(Q\). On the left use the common slope
\(t_a=\min(u_a,v_a)\); on the right use
\(t_b=\max(u_b,v_b)\). Extend both branches affinely from each endpoint
with that common slope. Since \(R(\mu)=0\), \(u_a\le0\) and
\(u_b\ge0\), so the extended floor is nonnegative. The slope ordering
gives convex joins, and the cutoff is constant on each tail. For a tail
report \(p\), the new floor equals the larger of the two endpoint
supporting lines for \(R\) and \(Q+c(e)\). It is therefore at most
\[
 \max\{R(p),Q(p)+c(e)\}\le R(p)+W.
\]
The expected floor increases by at most \(WG(T)\), while gross reward
falls by at most \(2MG(T)\). This comparison retains the old tail floor;
it requires no uniform bound tending to zero on old tail rents.

Interpolate \(R\) and \(Q\) separately by chords on the common core
mesh, retaining those affine tails. Convex chord slopes preserve both
joins. On a cell of length \(d\), a convex \(S\)-Lipschitz function's
chord excess at fraction \(t\) is at most
\(2St(1-t)d\le Sd/2\). Thus the floor rises by at most
\(S_R\Delta/2\), and the cutoff changes by at most
\(S_c\Delta/2\) on the core. Canonicalizing gives a polygon through the
actual \(c(p_i)\), with payoff loss at most
\((2M+W)G(T)+S_R\Delta/2+\Omega(S_c\Delta/2)\).

Round each nodal value downward to the largest \(z_j\le c(p_i)\).
The errors \(e_i\in[-\eta,0]\) change the polygon everywhere by at most
\(\eta\), including both tails and exact values \(L,H\). Each boundary
slope error is at most \(\eta/d_i\). An interior jump error is
\[
 \frac{e_{i-1}}{d_{i-1}}+\frac{e_{i+1}}{d_i}
 -e_i\left(\frac1{d_{i-1}}+\frac1{d_i}\right),
\]
whose absolute value is at most
\(\eta(1/d_{i-1}+1/d_i)\), since all errors occupy one interval of
width \(\eta\). The positive-part map is 1-Lipschitz. Applying these
bounds in \eqref{eq:constructive-polygon-rent} gives rent change at most
\(\eta\Lambda\), including both boundary charges. The reward loss is
at most \(\Omega(\eta)\). This proves the bound \(E_k\).
At a named report the rounded value differs from the actual \(c\) by at
most \(\eta\); hence it satisfies its relaxed band, even at a null report
or a shared band boundary. In particular,
\begin{equation}
 V_k(\sigma)\ge\mathcal J(c)-E_k
 \quad\text{if }\mathcal J(c)\ge0\text{ and }c\text{ satisfies }\sigma.
\label{eq:constructive-node-density}
\end{equation}

\emph{The payoff certificate.}
Every finite candidate is feasible, so \(V_k^0\le v^*\). Applying
\eqref{eq:constructive-node-density} at the root to all nonnegative-payoff
cutoffs gives \(v^*\le V_k^0+E_k\); negative payoffs cannot exceed the
root's zero-payoff candidate. Each lower choice loses at most
\(\lambda_k\), and an upper choice loses nothing. Therefore
\(P_k\ge V_k^0-k\lambda_k\), proving
\eqref{eq:constructive-bracket}. In particular \(P_k\ge-1/2\) and
\(P_k\to v^*\). A requested \(\epsilon>0\) is certified by the actual
test \(U_k-P_k\le\epsilon\), or by any integer \(k\) with
\((k+1/8)2^{-k}\le\epsilon\); such an integer is found by finite search.

\emph{Compactness and upper semicontinuity.}
For any sequence with \(\mathcal J(c_n)\ge-1\), its canonical floors satisfy
\(\int R_n\,dG\le M+1\), and the nonnegative convex functions
\(Y_n=R_n-c_n+H\) satisfy \(\int Y_n\,dG\le M+1+W\).
The guard argument just used, on each compact subinterval, bounds their
values and slopes uniformly. A diagonal Arzel\`a--Ascoli subsequence
converges locally uniformly to finite convex functions \(A,Y\).
Pointwise limits give \(c=A-Y+H\in[L,H]\), \(A\ge0\), and
\(A-c=Y-H\) convex. Fatou gives \(\int A\,dG<\infty\), so \(c\) is
feasible at every report. If a competitive sequence already converges
locally uniformly to \(c\), apply this argument along a subsequence
realizing the lower limit of its rents. Minimality of the canonical floor
and Fatou give
\[
 \operatorname{Rent}(c)\le\int A\,dG
 \le\liminf_n\operatorname{Rent}(c_n).
\]
Bounded stake continuity gives convergence of gross rewards by dominated
convergence, for this arbitrary \(G\). Thus
\(\mathcal J(c)\ge\limsup_n\mathcal J(c_n)\). This also proves attainment by applying
compactness to a maximizing sequence, as in Proposition
\ref{prop:threshold-attainment-zero-rent}.

\emph{Every fixed prefix stabilizes.}
Fix an exact finite-band node \(\sigma\) containing no optimizer.
Then \(\limsup_k V_k(\sigma)<v^*\). Otherwise, maximizing finite
witnesses along a subsequence would have payoffs tending to \(v^*\).
Competitive compactness and upper semicontinuity would give an optimal
limit. Their relaxed band constraints pass to the exact closed bands
because \(\eta_k\to0\), a contradiction. Hence some fixed
\(\delta_\sigma>0\) gives \(V_k(\sigma)\le v^*-\delta_\sigma\)
eventually. These gaps are used only in the proof.

Now fix a node containing an optimizer. If its lower child contains one,
\eqref{eq:constructive-node-density} gives
\(f_0\ge v^*-E_k\ge f-\lambda_k\), since \(f\le v^*\) and
\(E_k\le\lambda_k\); that child is accepted. If the lower child contains
none, the upper child contains one. Eventually its lower-child gap
exceeds \(E_k+\lambda_k\), giving
\[
 f_0\le v^*-\delta<v^*-E_k-\lambda_k\le f-\lambda_k,
\]
where the last inequality uses the parent's
\(f=V_k(\sigma)\ge v^*-E_k\). Thus the upper child is eventually
chosen. Starting at the root, induction on any fixed depth proves
stabilization to the first optimizer-containing child at every split.
Root restarts allow all premature finite-stage decisions to be revised.

\emph{The whole cutoff sequence and its implementation.}
Competitive compactness supplies cluster points of \(c_k\), all optimal
since \(P_k\to v^*\). Every cluster point satisfies every stabilized
exact prefix, by closedness and vanishing relaxation. At each \(\rho_i\),
the stabilized bands are repeatedly bisected and their widths tend to
zero. Thus all cluster points agree at every \(\rho_i\), and continuity
and density make them agree on \(I\). If the sequence did not converge
locally uniformly to this unique cluster point \(c^*\), a subsequence
would remain a fixed positive distance away on some compact interval.
Its compactness subsequence would converge to \(c^*\), a contradiction.

After taking this cutoff limit, set
\[
 \kappa=(D^2c^*)_+,\qquad
 R^*(p)=\int_I K_\mu(p,t)\,d\kappa(t),\qquad Q^*=R^*-c^*,
\]
using \(K_\mu\) from \eqref{eq:centered-hinge-kernel}, including its
right-hinge convention at \(t=\mu\). Compactness supplied an integrable
feasible floor; Theorem \ref{thm:finite-rent-curvature} therefore makes
\(R^*,Q^*\) finite convex functions and \(R^*\) an integrable least-cost
floor. For a finite output the same formula is the finite sum
\(R(p)=\sum_{i=0}^N(s_i-s_{i-1})_+K_\mu(p,p_i)\), with
\(s_{-1}=s_N=0\). No value or supporting slope is required at \(\ell,r\).

For either recovered cutoff and floor, take right derivatives
\(u(p)=R'_+(p)\), \(q(p)=Q'_+(p)\), and use action 1 at a cutoff tie.
The principal's state-contingent receipt is
\[
 t(\theta,p,\omega)=
 \begin{cases}
 -R(p)-u(p)(\omega-p),&\theta<c(p),\\
 c(p)-R(p)-q(p)(\omega-p),&\theta\ge c(p).
 \end{cases}
\]
These measurable supporting slopes implement the cutoff at all reports
by \eqref{eq:threshold-beta-subgradient} and
\eqref{eq:tilt-transfer-1}--\eqref{eq:tilt-transfer-0}. More explicitly,
each report offers supporting lines bounded by \(R\) and \(\theta+Q\)
at the true belief; truth attains their maximum. Nonnegativity of \(R\)
gives interim IR, and truthful expected receipts, \(-R\) and \(c-R\),
have absolute values bounded by \(R+\max\{|L|,|H|\}\). This proves the
maintained interim integrability as well as IC at null reports.
The exact output is the specified cutoff limit followed by this canonical
recovery. No general finite exact stopping guarantee, cutoff-profile
convergence rate, or convergence of stagewise floors, derivatives, or
payments is asserted.
\end{proof}
\end{toappendix}

\begin{toappendix}
\subsection{Positive Rent at an Optimal Threshold}
\label{app:optimal-adjustment-implications}

The following result identifies two obstacles to reducing rent through
further expansion of the action: a binding lower cutoff bound and
negative total surplus at the margin.

\begin{prop}[Necessary conditions for positive rent]
\label{prop:optimal-adjustment-implications}
Suppose Assumptions \ref{assump:conditional-density} and
\ref{assump:belief-support-span} hold. Let
\((c^*,R^*)\) solve \eqref{eq:rent-adjusted-global-program}.
If \(\operatorname{Rent}(c^*)>0\), then either
\begin{align}
&G\{p:R^*(p)>0,\ c^*(p)=\underline\theta\}>0
\nonumber\\[-2pt]
&\qquad\text{or}\qquad
G\{p:R^*(p)>0,\ c^*(p)>\underline\theta,\
V(p)+c^*(p)<0\}>0.
\label{eq:optimal-residual-rent}
\end{align}
\end{prop}

\begin{proof}
Suppose \(\operatorname{Rent}(c^*)>0\) but both sets in
\eqref{eq:optimal-residual-rent} have zero \(G\)-measure.
We lower the cutoff and floor together and show that this raises payoff.
Define
\[
h(p):=\min\{R^*(p),c^*(p)-\underline\theta\},
\qquad
c_\varepsilon:=c^*-\varepsilon h,
\qquad
R_\varepsilon:=R^*-\varepsilon h,
\qquad \varepsilon\in(0,1).
\]
Write \(Q^*=R^*-c^*\). Since
\[
R^*-h=\max\{0,Q^*+\underline\theta\},
\]
the function \(R^*-h\) is nonnegative and convex. Thus
\(R_\varepsilon=(1-\varepsilon)R^*+\varepsilon(R^*-h)\)
is nonnegative and convex, while
\(R_\varepsilon-c_\varepsilon=Q^*\) is convex. Moreover,
\[
c_\varepsilon-\underline\theta
\ge
(1-\varepsilon)(c^*-\underline\theta)\ge0,
\qquad
c_\varepsilon\le c^*\le\bar\theta.
\]
The perturbed pair is therefore feasible.

Using \eqref{eq:cutoff-virtual-surplus}, its pointwise payoff gain is
\begin{align*}
&B(c_\varepsilon(p),p)-R_\varepsilon(p)
-\bigl[B(c^*(p),p)-R^*(p)\bigr]\\
&\qquad=
\varepsilon h(p)F(c_\varepsilon(p)\mid p)
+[V(p)+c^*(p)]
\bigl[F(c^*(p)\mid p)-F(c_\varepsilon(p)\mid p)\bigr].
\end{align*}
Feasibility of \(R^*\) gives
\[
\int_{\mathcal P}R^*(p)\,dG(p)
\ge \operatorname{Rent}(c^*)>0,
\]
so \(\{R^*>0\}\) has positive \(G\)-measure. Because both alternatives in
\eqref{eq:optimal-residual-rent} are assumed to fail, almost every point in
this set satisfies \(c^*(p)>\underline\theta\) and
\(V(p)+c^*(p)\ge0\). At each such point, \(h(p)>0\),
\(c_\varepsilon(p)>\underline\theta\), and strict positivity of the
conditional density gives \(F(c_\varepsilon(p)\mid p)>0\).
Outside this set, \(h=0\), so the payoff gain is zero.
The gain is therefore nonnegative almost everywhere and strictly positive
on a set of positive \(G\)-measure. Integrating contradicts optimality of
\((c^*,R^*)\), proving \eqref{eq:optimal-residual-rent}.
\end{proof}
\end{toappendix}

\begin{toappendix}
\subsection{Partial Conjugates and Allocation Extensions}
\label{app:partial-conjugate-extensions}

The endpoint-matched cutoff records the total increment in envelope utility.
For a randomized allocation, global IC can also constrain every intermediate
increment. The following result records all of them through a partial
conjugate.

\begin{theorem}[Partial-conjugate characterization]
\label{thm:partial-conjugate-characterization}
Let \(x:\Theta\times\mathcal P\to[0,1]\) be measurable, with
\(x(\cdot,p)\) nondecreasing for every \(p\). Define
\begin{align}
A_x(\theta,p)
&:=
\int_{\underline\theta}^{\theta}x(s,p)\,ds,
\label{eq:partial-conjugate-envelope}\\
c_x(a,p)
&:=
\inf\{\theta\in\Theta:x(\theta,p)\ge a\},
\qquad a\in(0,1),
\label{eq:partial-conjugate-generalized-cutoff}
\end{align}
where the infimum of the empty set is \(\bar\theta\). For every
\(a\in[0,1]\), the partial conjugate of \(A_x\) satisfies
\begin{align}
\Gamma_x(a,p)
&:=
\sup_{\theta\in\Theta}
\{a\theta-A_x(\theta,p)\}
=
\int_0^a c_x(z,p)\,dz,
\label{eq:partial-conjugate-identity}
\end{align}
and, for every finite measurable floor \(R:\mathcal P\to\mathbb R\),
\begin{align}
U_R(\theta,p)
&:=R(p)+A_x(\theta,p)\\
&=
\sup_{a\in[0,1]}
\{a\theta+R(p)-\Gamma_x(a,p)\}.
\label{eq:partial-conjugate-utility-representation}
\end{align}

There exists a globally BIC and truthful-path interim IR mechanism with
allocation \(x\) and truthful utility \(U_R\) if and only if \(R\ge0\) and
\begin{align}
p\longmapsto R(p)-\Gamma_x(a,p)
\quad\text{is convex on \(\mathcal P\) for every \(a\in[0,1]\).}
\label{eq:partial-conjugate-convexity-test}
\end{align}
Consequently,
\begin{align}
\operatorname{Rent}(x)
=
\inf_R
\left\{
\int_{\mathcal P}R(p)\,dG(p):
R\ge0,
\quad R-\Gamma_x(a,\cdot)\text{ convex for every }a\in[0,1]
\right\},
\label{eq:partial-conjugate-rent-program}
\end{align}
where the infimum is over finite, \(G\)-integrable floors and equals
\(+\infty\) if none exists.
\end{theorem}

\begin{proof}
Fix \(p\) and suppress it from the notation. Fubini's theorem and the
layer-cake representation give, for every \(\theta\in\Theta\),
\begin{align}
A_x(\theta)
&=
\int_{\underline\theta}^{\theta}
\int_0^1\mathbbm 1\{a\le x(s)\}\,da\,ds\\
&=
\int_0^1(\theta-c_x(a))_+\,da.
\label{eq:partial-conjugate-layer-cake}
\end{align}
Indeed, monotonicity makes \(\{s:x(s)\ge a\}\) an upper interval; whether it
contains its lower endpoint does not affect its Lebesgue measure.

Set \(H(a):=\int_0^a c_x(z)\,dz\). The generalized cutoff is nondecreasing in
\(a\), so \(H\) is convex. Moreover,
\[
\sup_{a\in[0,1]}\{a\theta-H(a)\}
=
\sup_{a\in[0,1]}\int_0^a[\theta-c_x(z)]\,dz
=
\int_0^1(\theta-c_x(z))_+\,dz
=A_x(\theta).
\]
The middle equality holds because \(z\mapsto\theta-c_x(z)\) is
nonincreasing, so its positive set is an initial interval. Conversely, for
each \(a\in[0,1]\), \(H\) has at least one supporting slope in \(\Theta\):
\[
\partial H(a)\cap\Theta\ne\varnothing.
\]
At an interior rank, the one-sided slopes of \(H\) are limits of generalized
cutoffs and therefore lie in \(\Theta\); every slope between them supports
\(H\). At \(a=0\) and \(a=1\), \(\underline\theta\) and \(\bar\theta\),
respectively, are supporting slopes. Convex conjugacy therefore gives
\[
H(a)=\sup_{\theta\in\Theta}\{a\theta-A_x(\theta)\},
\]
which proves \eqref{eq:partial-conjugate-identity}. Taking the conjugate once
more and adding \(R(p)\) proves
\eqref{eq:partial-conjugate-utility-representation}.

We next prove the convexity test. Define
\begin{align}
g_a(p)
&:=R(p)-\Gamma_x(a,p)\\
&=
\inf_{\theta\in\Theta}
\{U_R(\theta,p)-a\theta\}.
\label{eq:partial-conjugate-value-function}
\end{align}
If \(U_R\) is jointly convex, partial minimization over the common convex set
\(\Theta\) makes \(g_a\) convex for every \(a\). Conversely, if every
\(g_a\) is convex, then
\[
U_R(\theta,p)
=
\sup_{a\in[0,1]}\{a\theta+g_a(p)\}
\]
is jointly convex as a supremum of jointly convex functions.

It remains to connect joint convexity to implementation of the specified
allocation, including at its discontinuities. Monotonicity gives, for every
\(\theta,s\in\Theta\),
\[
A_x(s,p)
\ge
A_x(\theta,p)+x(\theta,p)(s-\theta),
\]
so \(x(\theta,p)\) is a subgradient of \(A_x(\cdot,p)\) at \(\theta\).
Let \(a=x(\theta,p)\). Then \(\theta\) minimizes the expression in
\eqref{eq:partial-conjugate-value-function}. To select a belief subgradient
measurably, write \(g(a,p):=g_a(p)\). First note that \(c_x\) is jointly
measurable on \((0,1)\times\mathcal P\). Indeed, for every \(r\in\Theta\),
monotonicity gives
\[
\{(a,p):c_x(a,p)<r\}
=
\bigcup_{s\in\mathbb Q\cap(\underline\theta,r)}
\{(a,p):a\le x(s,p)\}.
\]
Extending \(c_x\) arbitrarily to the endpoint ranks does not affect its rank
integral. Because \(c_x\) is bounded, \(\Gamma_x\), and hence \(g\), is jointly
measurable on \([0,1]\times\mathcal P\). Moreover,
\(g(a,\cdot)\) is finite and convex for every \(a\). For \(n\ge1\), let
\[
\delta_n(p):=\frac{\bar p-p}{n+1},
\]
and define the right derivative
\begin{align}
\widehat b(a,p)
&:=\partial_p^+g(a,p)
=
\lim_{n\to\infty}
\frac{g(a,p+\delta_n(p))-g(a,p)}{\delta_n(p)}.
\label{eq:partial-conjugate-measurable-slope}
\end{align}
The limit is finite and belongs to \(\partial g_a(p)\) because
\(g(a,\cdot)\) is finite and convex on the open interval \(\mathcal P\).
It is jointly measurable as a limit of jointly measurable difference
quotients. Hence
\[
b(\theta,p):=\widehat b(x(\theta,p),p)
\]
is a measurable selection from \(\partial g_{x(\theta,p)}(p)\). For every
\((s,q)\),
\begin{align*}
U_R(s,q)-as
&\ge g_a(q)\\
&\ge g_a(p)+b(\theta,p)(q-p)\\
&=U_R(\theta,p)-a\theta+b(\theta,p)(q-p).
\end{align*}
Hence \((a,b(\theta,p))\in\partial U_R(\theta,p)\). Setting
\(\beta=b\) and applying Proposition
\ref{prop:convex-utility-characterization} gives global BIC. Since
\(A_x\ge0\) and \(A_x(\underline\theta,p)=0\), truthful-path interim IR is
equivalent to \(R\ge0\). This proves the equivalence. Taking the infimum over
integrable floors gives \eqref{eq:partial-conjugate-rent-program}.
\end{proof}

The characterization becomes finite-dimensional in the allocation rank for a
finite multithreshold rule.

\begin{corollary}[Finite multithreshold allocations]
\label{cor:finite-multithreshold-implementation}
Fix constants \(\alpha_1,\ldots,\alpha_m>0\) with
\(\sum_{j=1}^m\alpha_j\le1\), and measurable cutoff schedules
\[
c_j:\mathcal P\to[\underline\theta,\bar\theta],
\qquad j=1,\ldots,m,
\]
satisfying \(c_1(p)\le\cdots\le c_m(p)\). Consider
\begin{align}
x(\theta,p)
&=
\sum_{j=1}^m\alpha_j
\mathbbm 1\{\theta\ge c_j(p)\}.
\label{eq:finite-multithreshold-allocation}
\end{align}
There exists a globally BIC and truthful-path interim IR implementation with
floor \(R\) if and only if \(R\ge0\) and each of the \(m+1\) functions
\begin{align}
R(p)-\sum_{k=1}^j\alpha_kc_k(p),
\qquad j=0,1,\ldots,m,
\label{eq:finite-multithreshold-convexity}
\end{align}
is convex, where the sum for \(j=0\) is zero. A deterministic threshold is
the case \(m=1\) and \(\alpha_1=1\), which recovers convexity of \(R\) and
\(R-c_1\).
\end{corollary}

\begin{proof}
Let \(q_j:=\sum_{k=1}^j\alpha_k\), with \(q_0=0\). For
\(a\in[q_{j-1},q_j]\), equation
\eqref{eq:partial-conjugate-identity} gives
\[
\Gamma_x(a,p)
=
\sum_{k=1}^{j-1}\alpha_kc_k(p)
+(a-q_{j-1})c_j(p).
\]
Writing \(\tau=(a-q_{j-1})/\alpha_j\) therefore yields
\begin{align*}
R-\Gamma_x(a,\cdot)
&=
(1-\tau)
\left(R-\sum_{k=1}^{j-1}\alpha_kc_k\right)
+\tau
\left(R-\sum_{k=1}^{j}\alpha_kc_k\right).
\end{align*}
Thus convexity at the finitely many ranks \(q_j\) implies convexity at every
intermediate rank. If \(q_m<1\), then for \(a\in[q_m,1]\),
\[
\Gamma_x(a,p)
=
\sum_{k=1}^m\alpha_kc_k(p)+(a-q_m)\bar\theta,
\]
so the remaining condition is equivalent to the condition at \(q_m\).
Necessity follows by evaluating
\eqref{eq:partial-conjugate-convexity-test} at each \(q_j\), and Theorem
\ref{thm:partial-conjugate-characterization} completes the proof.
\end{proof}

For a smooth monotone allocation, the same family of constraints becomes a
curvature inequality for its inverse allocation rule.

\begin{corollary}[Smooth monotone allocations]
\label{cor:smooth-monotone-implementation}
Let \(x:\Theta\times\mathcal P\to[0,1]\) be measurable. Suppose
\(x(\cdot,p)\) is continuous and strictly increasing from zero to one on
\(\Theta\) for every \(p\). Let
\[
c(a,p):=x(\cdot,p)^{-1}(a),
\qquad (a,p)\in[0,1]\times\mathcal P.
\]
Then \(x\) has a globally BIC and truthful-path interim IR implementation with
floor \(R\) if and only if \(R\ge0\) and
\begin{align}
p\longmapsto
R(p)-\int_0^a c(z,p)\,dz
\quad\text{is convex for every \(a\in[0,1]\).}
\label{eq:smooth-monotone-integrated-cutoff-test}
\end{align}
If \(R\in C^2(\mathcal P)\) and \(c_{pp}\) exists and is continuous on
\([0,1]\times\mathcal P\), this is equivalent to
\begin{align}
R''(p)
\ge
\int_0^a c_{pp}(z,p)\,dz
\qquad
\forall(a,p)\in[0,1]\times\mathcal P.
\label{eq:smooth-monotone-curvature-test}
\end{align}
\end{corollary}

\begin{proof}
Strict monotonicity makes the generalized cutoff in
\eqref{eq:partial-conjugate-generalized-cutoff} equal to \(c(a,p)\).
Equation \eqref{eq:smooth-monotone-integrated-cutoff-test} therefore follows
from Theorem \ref{thm:partial-conjugate-characterization}. Under the stated
smoothness, differentiation under the integral gives
\[
\partial_{pp}
\left[R(p)-\int_0^a c(z,p)\,dz\right]
=
R''(p)-\int_0^a c_{pp}(z,p)\,dz.
\]
A twice continuously differentiable function on an interval is convex if and
only if its second derivative is nonnegative, which proves
\eqref{eq:smooth-monotone-curvature-test}.
\end{proof}

Theorem \ref{thm:partial-conjugate-characterization} replaces the single
cutoff by the cumulative cutoff \(\Gamma_x(a,p)\) at every allocation rank
\(a\). For deterministic thresholds, \(\Gamma_x(a,p)=ac(p)\), and the
continuum of constraints reduces to convexity of \(R\) and \(R-c\). For a
genuine lottery, intermediate ranks can impose additional constraints even
when the mean cutoff \(\Gamma_x(1,p)\) is implementable. This identity unifies
the three allocation classes, but the continuum of rank-specific conditions
also explains why it does not solve the unrestricted outer problem.
\end{toappendix}

\begin{toappendix}
\subsection{A Regional Test for Threshold Optimality}
\label{app:regional-threshold-optimality}

When is a given threshold globally optimal, including against randomized
mechanisms? For continuous cutoffs with finitely many affine pieces, this
subsection gives an exact test under two additional conditions. We first
express profit using a signed measure of types, then test that measure on
each region of the candidate's menu. Finally, an example shows why these
tests can succeed when the certificate in Appendix
\ref{app:global-optimality-certificate} cannot: that certificate bounds a
relaxation, whereas the regional argument retains all reporting incentives.
The method adapts the finite-menu dominance approach of
\citet*[Theorem~3]{daskalakis2017strong}; the proof below accommodates
unrestricted belief slopes, the principal's action value, and the maintained
interim-integrable contract class.

\paragraph{The revenue measure.}
Let \(Z=\Theta\times\mathcal P\), let
\(\overline Z=\Theta\times\overline{\mathcal P}\), and set
\(s(\theta,p)=V(p)+\theta\). Extend \(G\) by zero to the two belief
endpoints. In addition to Assumption \ref{assump:conditional-density}, require
that \(f(\cdot\mid p)\) have an absolutely continuous version for
\(G\)-almost every \(p\), with jointly measurable derivative and boundary
traces, and
\begin{equation*}
\int_{\mathcal P}\left[
\sum_{a\in\{\underline\theta,\bar\theta\}}|s(a,p)f(a\mid p)|
+\int_\Theta |2f+s\,\partial_\theta f|\,d\theta
\right]dG(p)<\infty.
\tag{A}\label{eq:regional-regularity}
\end{equation*}
The arguments \((\theta,p)\) are suppressed in the inner integral. Define a
finite signed revenue measure \(\nu\) by its integral against any bounded
Borel function \(\varphi\):
\begin{align}
\int_{\overline Z}\varphi\,d\nu
:=\int_{\mathcal P}\bigg[
&s(\bar\theta,p)f(\bar\theta\mid p)\varphi(\bar\theta,p)
-s(\underline\theta,p)f(\underline\theta\mid p)\varphi(\underline\theta,p)
\nonumber\\
&-\int_\Theta(2f+s\,\partial_\theta f)\varphi\,d\theta
\bigg]dG(p).
\label{eq:regional-revenue-measure}
\end{align}
For a finite maximum of affine utility functions with stake slopes in
\([0,1]\), integration by parts in \(\theta\) gives
\begin{equation}
\Pi(U)=\int_Z(sU_\theta-U)\,dH
=\int_{\overline Z}U\,d\nu.
\label{eq:regional-payoff-identity}
\end{equation}
Indeed, \(\partial_\theta(sf)=f+s\,\partial_\theta f\), and the utility
cost supplies the second \(f\). The stake derivative equals the implemented
allocation \(H\)-almost everywhere. Condition
\eqref{eq:regional-regularity} justifies integration over beliefs. The
interior density of \(\nu\) is \(B_{zz}=-(2f+s\,\partial_\theta f)\);
the stake-boundary terms in \eqref{eq:regional-revenue-measure} must also
be retained.

\paragraph{The candidate's regions.}
Fix a continuous, finite-piece-affine cutoff \(c^*\) and its canonical floor
\(R^*\) from \eqref{eq:curvature-least-cost-floor}. Both \(R^*\) and
\(R^*-c^*\) are convex and finite-piece affine. Hence
\begin{align}
U^*(\theta,p)&=\max\{R^*(p),\theta+R^*(p)-c^*(p)\}
=\max_{0\le j\le m}\ell_j(\theta,p),
\nonumber\\
\ell_j(\theta,p)&=a_j\theta+b_jp+d_j,\qquad a_j\in\{0,1\}.
\label{eq:regional-candidate-menu}
\end{align}
Retain the zero plane \(\ell_0=0\), and remove duplicate planes and
redundant nonzero planes. All candidate pieces extend to \(\overline Z\).
Let
\[
C_j=\{z\in\overline Z:\ell_j(z)>\ell_i(z)\ \text{for every }i\ne j\},
\qquad
\mathcal T=\overline Z\setminus\bigcup_j C_j.
\]
These are regions with a unique preferred menu option, not merely the
action and no-action sets. We require
\begin{equation*}
|\nu|(\mathcal T)=0.
\tag{T0}\label{eq:regional-no-tie-mass}
\end{equation*}
This restriction concerns revenue measure, including its stake-boundary
mass. Atomless preferences alone do not imply it. Belief endpoints carry
no mass; no endpoint report or endpoint extension is imposed on competitors.

Let \(\mathcal F_+\) and \(\mathcal F_-\) be the finite maxima of affine
functions on \(\mathbb R^2\) whose stake slopes are, respectively,
nonnegative and nonpositive. Belief slopes and intercepts are unrestricted.
These test functions are convex and bounded on \(\overline Z\). For the
outside-option region use the nonnegative class
\(\mathcal F_+^0=\{\max(0,\psi):\psi\in\mathcal F_+\}\).

\begin{theorem}[Regional threshold optimality]
\label{thm:regional-threshold-optimality}
Under Assumption \ref{assump:conditional-density},
\eqref{eq:regional-regularity}, and \eqref{eq:regional-no-tie-mass}, the
candidate \eqref{eq:regional-candidate-menu} is globally optimal among all
maintained BIC/IR mechanisms if and only if
\begin{align}
\int_{C_0}\psi\,d\nu&\le0
&&\text{for every }\psi\in\mathcal F_+^0,
\nonumber\\
\int_{C_j}\psi\,d\nu&\le0
&&\text{for every }\psi\in\mathcal F_+,\quad j\ne0,\ a_j=0,
\label{eq:regional-dominance}\\
\int_{C_j}\psi\,d\nu&\le0
&&\text{for every }\psi\in\mathcal F_-,\quad a_j=1.
\nonumber
\end{align}
A violating test yields a profitable finite-menu BIC/IR mechanism by a
sufficiently small perturbation of the corresponding menu option.
\end{theorem}

The test compares the positive and negative parts of the revenue measure
within each region. A no-action option can be perturbed toward positive
action probabilities; an action option can be perturbed toward lower
probabilities. The outside option additionally requires nonnegative utility.
There are finitely many regions but potentially infinitely many dominance
inequalities. The theorem is not a finite scalar checklist or a general
algorithm for finding a violating test.

\begin{proof}
We first justify using finite menus without restricting the competitor
class. We then prove sufficiency region by region and necessity by
perturbing one option.

\emph{Finite-menu approximation.}
For any maintained mechanism, interim integrability and
\(T=\theta x-U\) imply \(U\in L^1(H)\); IR gives \(U\ge0\).
Global IC gives convexity and, for \(\theta_1\le\theta_2\),
\[
0\le U(\theta_2,p)-U(\theta_1,p)\le\theta_2-\theta_1.
\]
Thus the finite convex floor \(R(p)=U(\underline\theta,p)\) bounds
utility on every compact belief strip, including both stake boundaries.
On a smaller strip, the supporting-slope inequalities evaluated at two
beliefs a fixed distance to its left and right bound the belief slopes.
Stake slopes already lie in \([0,1]\). These bounds give continuity and
bounded supporting slopes locally, also at the stake boundaries.

Select a supporting payoff plane \(q_k\) at each of countably many dense
interior types, and let \(U_n=\max(0,q_1,\ldots,q_n)\). Approaching any
type by those interior types, local continuity and the slope bounds show
that their supporting planes approach its utility. Consequently
\(0\le U_n\uparrow U\) pointwise on \(Z\). Each \(U_n\) is a finite
menu with stake slopes in \([0,1]\); choosing the first active plane
measurably supplies an allocation \(x_n\) and tilt. The recovery formulas
\eqref{eq:tilt-transfer-1}--\eqref{eq:tilt-transfer-0} give bounded transfers
and full-domain BIC/IR.

For each fixed belief, secant bounds for convex functions imply
\(x_n(\theta,p)\to U_\theta(\theta,p)=x(\theta,p)\) wherever the latter
derivative exists in the interior of \(\Theta\). Conditional absolute
continuity and Fubini make the exceptional set \(H\)-null. Since \(s\)
is bounded, dominated convergence for \(s x_n\) and monotone convergence
for \(U_n\) yield
\[
\Pi(U_n,x_n)=\int_Z(sx_n-U_n)\,dH
\longrightarrow\int_Z(sx-U)\,dH=\Pi(U,x).
\]
There is no common transfer bound on this sequence. In particular, this
argument does not require an arbitrary competitor to be integrable against
\(\nu\) or to extend to the belief endpoints.

\emph{Sufficiency.}
For a finite-menu utility \(U\), the difference \(U-\ell_j\) belongs to
\(\mathcal F_+\) when \(a_j=0\), and to \(\mathcal F_-\) when \(a_j=1\).
On \(C_0\), it is \(U\in\mathcal F_+^0\). The payoff identity
\eqref{eq:regional-payoff-identity}, the tie condition
\eqref{eq:regional-no-tie-mass}, and the regional inequalities give
\[
\Pi(U)-\Pi(U^*)=\sum_j\int_{C_j}(U-\ell_j)\,d\nu\le0.
\]
The approximation just proved passes this bound to every maintained
mechanism, including arbitrary randomized ones.

\emph{Necessity and an improving mechanism.}
Suppose a test \(\psi=\max_k(\alpha_k\theta+\zeta_kp+q_k)\) violates
the inequality for \(j\ne0\). For sufficiently small \(\varepsilon>0\)
with \(\varepsilon\max_k|\alpha_k|\le1\), replace option \(j\) by
\[
U_\varepsilon
=\max\{\ell_i:i\ne j;\ \ell_j+\varepsilon\psi\}.
\]
Every new plane has stake slope in \([0,1]\): the test slopes are
nonnegative if \(a_j=0\) and nonpositive if \(a_j=1\). The zero
option remains. This is therefore a finite-menu BIC/IR mechanism.
At every unique-menu type its directional utility change is
\(\mathbf1_{C_j}\psi\), because the winning plane is separated from
the others by a strictly positive gap. Also
\[
|U_\varepsilon-U^*|\le\varepsilon\|\psi\|_{\infty,\overline Z}.
\]
Finiteness of \(|\nu|\), \eqref{eq:regional-no-tie-mass}, and dominated
convergence imply
\[
\lim_{\varepsilon\downarrow0}
\frac{\Pi(U_\varepsilon)-\Pi(U^*)}{\varepsilon}
=\int_{C_j}\psi\,d\nu>0.
\]
Every sufficiently small positive feasible step raises profit. If the
violating region is \(C_0\), use instead
\(U_\varepsilon=\max(U^*,\varepsilon\psi)\) with
\(\psi\in\mathcal F_+^0\). Its derivative is \(\psi\) on \(C_0\)
and zero on the other unique regions; the same bound and convergence
argument establish the improvement.
\end{proof}

If the candidate already solves the full cutoff-choice problem, an improving
mechanism supplied by the proof must randomize on positive \(H\)-measure.
Otherwise its monotonicity in the stake makes it agree with a threshold,
apart from stake-null types, on \(G\)-almost every belief slice. The endpoint-matched cutoff
\eqref{eq:endpoint-matched-cutoff} then has the same utility floor and
payoff, contradicting optimality within the threshold class.

\subsubsection{A strict gap in the concave-envelope bound}
\label{app:regional-strict-separation}

The following example separates the regional test from the preceding
certificate. The candidate is constant and rent-free. The construction does
not establish such a separation for an optimum with positive implementation
rent.

\begin{prop}[A strict relaxation gap]
\label{prop:regional-strict-separation}
There is an environment on \(\Theta=[1,2]\), \(\mathcal P=(0,1)\) with
positive belief density, strictly positive conditional stake densities, and
\(v(1)=1>0>v(0)=-1\), satisfying
\eqref{eq:regional-regularity} and \eqref{eq:regional-no-tie-mass}, in which
\[
c^*(p)=5/4,\qquad R^*(p)=0,\qquad
\mathcal V_{\mathrm{glob}}=\mathcal V_{\mathrm{thr}}=\frac{1867}{8000}.
\]
Cutoff virtual surplus is nonconcave on a positive-probability interval
of beliefs, and
\begin{align}
\widehat{\mathcal V}&\ge\frac{11977}{50000},
\label{eq:regional-strict-relaxation-gap}\\
\inf_{\lambda\ \mathrm{feasible}}D(\lambda)-\mathcal V_{\mathrm{glob}}
&\ge\frac{1233}{200000}>0.
\nonumber
\end{align}
Here feasible prices satisfy
\eqref{eq:certificate-affine-moments} and
\eqref{eq:certificate-dual-feasibility}. These are lower bounds on the
relaxation value and certificate gap, not their exact values.
\end{prop}

\begin{proof}
We construct a joint type law concentrated mostly in three belief bands.
Two midpoint comparisons verify regional optimality. A small full-support
component preserves those comparisons and the strict relaxation gap.

\emph{Primitives.}
Write \(u=\theta-1\in[0,1]\), \(d=1/4\), and \(k(p)=2p\), so
\(V(p)+\theta=k(p)+u>0\). For \(k>0\) and \(-k<K<1+k\), define
\[
B_{k,K}(u)=(1-u)(k+Ku),\qquad
S_{k,K}(u)=\frac{B_{k,K}(u)}{k+u}.
\]
This survivor equals one at zero and zero at one. Differentiation gives
\[
f_{k,K}(u)=-S'_{k,K}(u)
=\frac{k(1+k-K)+Ku(2k+u)}{(k+u)^2}>0.
\]
If \(K\ge0\), the numerator is minimized at zero, with positive value
\(k(1+k-K)\). If \(K<0\), it is minimized at one, with positive value
\((k+1)(k+K)\). Thus it defines a full-support conditional law.

Let \(e=1/100\), draw \(z\) uniformly from \((-e,e)\), and choose one
of three indices independently with equal probabilities. Set
\[
\begin{array}{c|c|c|c}
i&p_i(z)&k_i(z)=2p_i(z)&K_i(z)\\ \hline
1&1/20+z&1/10+2z&9/20+6z\\
2&1/10+z&1/5+2z&-1/10\\
3&3/20+z&3/10+2z&17/20+6z
\end{array}
\]
Let \(I_i=(p_i(0)-e,p_i(0)+e)\). This component has belief density
\(g_S=(50/3)\mathbf1_{I_1\cup I_2\cup I_3}\). On band \(i\), set
\[
S_S(u,p_i(z))=S_{k_i(z),K_i(z)}(u),\qquad
B_S(u,p_i(z))=B_{k_i(z),K_i(z)}(u).
\]
The bounds \(1+k_1-K_1\ge61/100\), \(1+k_3-K_3\ge41/100\),
and \(k_2+K_2\ge2/25\) verify the density conditions throughout the bands.

For a second component, use uniform beliefs and define
\[
B_A(u,p)=
\begin{cases}
k(p)+u/2-u^2,&0\le u\le d,\\[2pt]
\displaystyle [k(p)+1/16]
\left[1-\left(\frac{u-d}{1-d}\right)^2\right],&d\le u\le1.
\end{cases}
\]
This function is continuously differentiable, strictly concave, and
maximized at \(d\). Its survivor \(S_A=B_A/(k+u)\) has the required
endpoints and positive density: \(N=B_A-(k+u)B_{A,u}\) satisfies
\(N(0,p)=k/2>0\) and \(N_u=-(k+u)B_{A,uu}>0\) on both pieces,
so \(f_A=N/(k+u)^2>0\). The matching first derivatives at \(d\)
make this density continuous and absolutely continuous in the stake.

Mix these two \emph{joint type laws}, with weights \(99/100\) and
\(1/100\):
\begin{align}
g(p)&=\frac{99}{100}g_S(p)+\frac1{100},
\nonumber\\
1-F(1+u\mid p)
&=\frac{(99/100)g_S(p)S_S(u,p)+(1/100)S_A(u,p)}{g(p)}.
\label{eq:regional-separation-primitives}
\end{align}
Define \(S_S=S_A\) outside the open bands, where its weight is zero.
This gives a measurable conditional law at every belief, with positive
density, and \(g\ge1/100\) on \((0,1)\). The prior is
\(\mu=(99/100)(1/10)+(1/100)(1/2)=13/125\). Drawing
\(\omega\mid(\theta,p)\) as Bernoulli \(p\) realizes the specified beliefs.
The final cutoff surplus satisfies
\begin{equation}
g(p)B(1+u,p)=\frac{99}{100}g_S(p)B_S(u,p)+\frac1{100}B_A(u,p).
\label{eq:regional-separation-mixture}
\end{equation}

\emph{Candidate and regularity.}
The constant cutoff \(c^*=5/4\) has canonical floor zero and utility
\(U^*=(u-d)_+\). It is implemented on the full domain by
\(x^*=\mathbf1\{\theta\ge5/4\}\), \(\beta^*=0\), and
\(t_0^*=t_1^*=(5/4)x^*\). Its regions are \(C_0=\{u<d\}\) and
\(C_1=\{u>d\}\).

In the shifted stake coordinate, the revenue measure has interior
density \(B_{uu}\), lower-boundary coefficient \(B_u(0,p)-1\), and
upper-boundary coefficient \(-B_u(1,p)\). Thus the conditional measure
on band \(i\) is
\begin{equation}
\nu_i=(k_i+K_i)\delta_1+(K_i-k_i-1)\delta_0-2K_i\,du.
\label{eq:regional-band-revenue}
\end{equation}
All coefficients are bounded. For the base component,
\(B_{A,uu}=-2\) below \(d\) and
\(-32[k(p)+1/16]/9\) above it; the boundary coefficients are
\(-1/2\) and \(8[k(p)+1/16]/3\). These are bounded as well.
The densities are absolutely continuous in the stake, so these bounds
verify \eqref{eq:regional-regularity}, including the boundary traces.
The final revenue measure is the same weighted mixture of component
measures. No interior stake atom occurs at \(d\), since \(B_{A,u}\)
is continuous there. Neither stake boundary is a menu tie. Hence
\eqref{eq:regional-no-tie-mass} holds.

\emph{Regional inequalities.}
Fix \(z\) and temporarily omit the common weight \(1/3\) on the three
beliefs. Regard \(\nu_i\) as supported at \(p_i(z)\), and split off
\[
\sigma_z(du,dp)=\frac1{10}\,du\,
[2\delta_{p_2(z)}-\delta_{p_1(z)}-\delta_{p_3(z)}](dp).
\]
Because \(p_2(z)\) is the midpoint of \(p_1(z)\) and \(p_3(z)\),
every convex test integrates nonpositively against \(\sigma_z\)
on either common stake interval \(u<d\) or \(u>d\).
The remaining interior densities are
\(-D_1,0,-D_3\), where
\(D_1=4/5+12z>0\) and \(D_3=8/5+12z>0\).
The lower-boundary coefficients are
\(-13/20+4z,-13/10-2z,-9/20+4z\), all negative. The remaining
no-action measure is therefore nonpositive, establishing its inequality
for nonnegative tests.

For the action region, take a convex, stake-decreasing test \(\psi\),
which may have either sign, and set \(A=1/10+2z>0\).
The middle upper-boundary mass is \(A\). For \(i=1,3\),
\(\psi(u,p_i)\ge\psi(1,p_i)\) gives
\[
(k_i+K_i)\psi(1,p_i)-D_i\int_d^1\psi(u,p_i)\,du
\le[k_i+K_i-(1-d)D_i]\psi(1,p_i)
=-\frac A2\psi(1,p_i).
\]
The remaining action integral is at most
\[
A\left[\psi(1,p_2)-\tfrac12\psi(1,p_1)-\tfrac12\psi(1,p_3)\right]\le0.
\]
Adding back \(\sigma_z\) and integrating over \(z\) proves both
regional inequalities for the three-band component.
For the base component the no-action measure is nonpositive, while the
action measure has negative interior density and total mass zero because
\(B_{A,u}(d,p)=0\). Moving its negative interior mass to the upper
stake likewise bounds the action integral by zero. The candidate's regions
are common to both components, so linearity gives
\eqref{eq:regional-dominance} for the final mixture. Theorem
\ref{thm:regional-threshold-optimality} proves unrestricted optimality.

\emph{Profit and strict gap.}
At the candidate, the average of the three band surpluses is
\[
\frac13\sum_i B_S(d,p_i(z))=\frac9{40}+\frac94z.
\]
Symmetry gives expected profit \(9/40\) in this component; the base
gives \(\int_0^1B_A(d,p)\,dp=17/16\). Hence
\[
\Pi^*=\frac{99}{100}\frac9{40}+\frac1{100}\frac{17}{16}
=\frac{1867}{8000}.
\]
On the middle band \(I_2=(9/100,11/100)\), use a pointwise lottery
with probabilities \(3/4\) and \(1/4\) on cutoffs \(1\) and \(2\);
elsewhere use cutoff \(5/4\). Its mean cutoff is always \(5/4\).
This is a witness for the relaxation, not an implementable mechanism.
Its conditional surplus gain is \(3/160\) in the middle-band component,
and \(-p/2-1/16\) in the base component. Thus
\begin{align}
\int_{\mathcal P}[\widehat B(c^*(p),p)-B(c^*(p),p)]\,dG(p)
&\ge\frac{99}{100}\frac13\frac3{160}
-\frac1{100}\int_{9/100}^{11/100}\left(\frac p2+\frac1{16}\right)dp
\nonumber\\
&=\frac{1233}{200000}.
\label{eq:regional-envelope-witness}
\end{align}
Since \(\operatorname{Rent}(c^*)=0\), adding this gain to \(\Pi^*\)
proves the bound on \(\widehat{\mathcal V}\). Nonconcavity also follows
directly: for \(p\in I_2\) and \(0<u<d\),
\[
g(p)\,\partial_{uu}B(1+u,p)
=\frac{99}{100}\frac{50}{3}\frac15-\frac2{100}
=\frac{82}{25}>0.
\]
Finally, the maximum defining \(D(\lambda)\) weakly exceeds the
witness lottery's expected \(B-\lambda z\) at each belief. Its mean
is \(5/4\), and feasible prices have \(\int\lambda\,dG=0\).
Therefore \(D(\lambda)\ge\Pi^*+1233/200000\) for every such price.
The feasible class contains \(\lambda=0\), so its infimum also has
the claimed strictly positive gap.
\end{proof}

The missing incentive restriction has a direct interpretation. With the
zero floor used to implement its mean cutoff, the witness lottery would
give utility \(3u/4\) on the middle band and \((u-1/4)_+\) elsewhere.
For \(0<u<1\), the middle value strictly exceeds the average at the two
neighboring beliefs \(p_1(z)\) and \(p_3(z)\). This violates convexity
in belief and hence global IC. The relaxation prices implementation of the
mean cutoff, not of the lottery itself. The regional inequalities show
that no fully incentive-compatible redesign can outperform the candidate;
they do not attribute every possible repair's cost to floor rent alone.
\end{toappendix}

\begin{toappendix}
\subsection{An Optimal Cutoff with Positive Rent}
\label{app:positive-rent-soft-concavification}

The following example shows that an optimal cutoff can retain upward
curvature and positive implementation rent while strictly outperforming
every concave cutoff.

\begin{prop}[An optimal cutoff can retain positive curvature]
\label{prop:positive-rent-soft-concavification}
Let \(\mathcal P=(0,1)\), let \(G\) be uniform, and let
\(\Theta=[3,4]\). Set \(v(0)=-3\), \(v(1)=3\), so that
\(V(p)=6p-3\). For \(z\in[3,4]\), write \(u=z-3\) and suppose the
conditional survivor function of \(\theta\) is
\begin{align}
1-F(z\mid p)
&=
\frac{(1-u)[6p+(1+5p)u]}{6p+u}.
\label{eq:positive-rent-example-survivor}
\end{align}
Let \(t\) be the unique root in \((1/15,1/14)\) of
\begin{align}
5-80t+78t^2-26t^3+5t^4=0,
\label{eq:positive-rent-example-free-boundary}
\end{align}
and define
\begin{align}
a
&:=
\frac{1-9t+3t^2-t^3}{2(1-t)^3}.
\label{eq:positive-rent-example-slope}
\end{align}
Then the unique optimal cutoff, up to \(G\)-null sets, is
\begin{align}
c^*(p)&=3+a(1-t-p)_+.
\label{eq:positive-rent-example-optimal-cutoff}
\end{align}
It is globally optimal among all BIC/IR mechanisms, strictly outperforms every
concave cutoff, and has least-cost floor and implementation rent
\begin{align}
R^*(p)&=a\bigl(p-(1-t)\bigr)_+,
&
\operatorname{Rent}(c^*)&=\frac{at^2}{2}>0.
\label{eq:positive-rent-example-rent}
\end{align}
It can be implemented with bounded transfers by charging
\(3+a(1-\omega-t)\) when the action is taken and zero otherwise.
\end{prop}

\begin{proof}
Write \(u=z-3\) and \(K=1+5p\). The survivor function in
\eqref{eq:positive-rent-example-survivor} equals one at \(z=3\), equals zero
at \(z=4\), and has density
\begin{align}
f(z\mid p)
&=
\frac{6p^2+Ku(12p+u)}{(6p+u)^2}>0
\qquad
\text{for }(z,p)\in[3,4]\times(0,1).
\label{eq:positive-rent-example-density}
\end{align}
It therefore defines an absolutely continuous conditional distribution with
common support \([3,4]\). All primitives are jointly measurable, and drawing
\(\omega\) conditionally on \((\theta,p)\) as a Bernoulli variable with mean
\(p\) satisfies posterior consistency.

Since \(V(p)+z=6p+u\), multiplying the survivor function by this term gives
\begin{align}
B(z,p)
&=
6p+(1-p)u-(1+5p)u^2.
\label{eq:positive-rent-example-cutoff-surplus}
\end{align}
Hence \(B_{zz}(z,p)=-2(1+5p)<0\), and its unique pointwise maximizer is
\begin{align}
c^0(p)&=3+\frac{1-p}{2(1+5p)}.
\label{eq:positive-rent-example-pointwise-cutoff}
\end{align}
This benchmark is strictly convex because
\[
(c^0)'(p)=-\frac{3}{(1+5p)^2},
\qquad
(c^0)''(p)=\frac{30}{(1+5p)^3}>0.
\]

We next verify the constants in
\eqref{eq:positive-rent-example-free-boundary}--
\eqref{eq:positive-rent-example-slope}. Let
\[
H(s):=5-80s+78s^2-26s^3+5s^4.
\]
Direct calculation gives
\[
H(1/15)=\frac{58}{10125}>0,
\qquad
H(1/14)=-\frac{12511}{38416}<0.
\]
Moreover, \(H'<0\) on \((1/15,1/14)\), since there
\[
H'(s)
=-80+156s-78s^2+20s^3
<-80+\frac{156}{14}+\frac{20}{14^3}<0.
\]
Thus \(t\) exists and is unique. Because \(t<1/14\), the numerator in
\eqref{eq:positive-rent-example-slope} is positive. It is smaller than
\((1-t)^3\), with the difference equal to \(6t\). Consequently,
\(0<a<1/2\), so \(c^*\) lies in \([3,4]\).

The cutoff in \eqref{eq:positive-rent-example-optimal-cutoff} has curvature
measure \(D^2c^*=a\delta_{1-t}\). Because \(1-t>1/2\), the hinge Jensen gap
satisfies \(J_G(1-t)=t^2/2\). Theorem
\ref{thm:finite-rent-curvature} therefore gives precisely the floor and rent
in \eqref{eq:positive-rent-example-rent}.

It remains to prove optimality. Define
\begin{align}
\lambda(p)
&:=
\begin{cases}
1-p-2(1+5p)a(1-t-p), & p\le 1-t,\\
1, & p>1-t.
\end{cases}
\label{eq:positive-rent-example-certificate}
\end{align}
Using \eqref{eq:positive-rent-example-slope}, direct integration yields
\begin{align}
\int_0^1\lambda(p)\,dp
&=\frac{H(t)}{6(t-1)}=0,
&
\int_0^1p\lambda(p)\,dp
&=-\frac{H(t)}{12}=0.
\label{eq:positive-rent-example-moments}
\end{align}
On \([0,1-t]\), \(\lambda\) is a strictly convex quadratic. Its endpoint
values are
\[
\lambda(1-t)=t,
\qquad
\lambda(0)=1-2a(1-t)
=t+\frac{6t}{(1-t)^2},
\]
which both lie in \((0,1)\); indeed, \(t<1/14\) implies
\[
t+\frac{6t}{(1-t)^2}
<\frac1{14}+\frac{84}{169}<1.
\]
Hence \(\lambda\le1\) on \((0,1)\). The first identity in
\eqref{eq:positive-rent-example-moments}, together with \(\lambda=1\) on
\((1-t,1)\), implies that \(\lambda\) is negative somewhere on
\((0,1-t)\). Strict convexity then gives exactly two roots there, with sign
pattern positive, negative, positive across \((0,1)\).

Define the certificate potential
\[
h(s):=\int_s^1(p-s)\lambda(p)\,dp.
\]
The two moment identities also give
\[
h(s)=\int_0^s(s-p)\lambda(p)\,dp.
\]
Let \(L(s):=\int_0^s\lambda(p)\,dp=h'(s)\). The sign pattern of
\(\lambda\), together with \(L(0)=L(1)=0\), implies that \(L\) is first
positive and then negative: it increases before the first root, decreases
between the two roots, and after the second root increases to zero from
below. Since \(h(0)=h(1)=0\), it follows that \(h(s)>0\) for every
\(s\in(0,1)\). The inequality \(\lambda\le1\) and the two representations of
\(h\) give
\[
h(s)
\le
\begin{cases}
\displaystyle\int_0^s(s-p)\,dp=s^2/2, & s\le1/2,\\[6pt]
\displaystyle\int_s^1(p-s)\,dp=(1-s)^2/2, & s\ge1/2.
\end{cases}
\]
The right-hand side is \(J_G(s)\) for the uniform distribution, so
\(0\le h\le J_G\). Because \(\lambda=1\) above \(1-t\),
\[
h(1-t)=\int_{1-t}^1\bigl(p-(1-t)\bigr)\,dp
=\frac{t^2}{2}=J_G(1-t).
\]

For \(p>1-t\), the derivative of \(B(z,p)-\lambda(p)z\) at the lower
cutoff bound \(z=3\) is \((1-p)-1=-p<0\), so strict concavity makes
\(c^*(p)=3\) its unique maximizer. For \(p\le1-t\), the derivative
vanishes at \(z=3+a(1-t-p)=c^*(p)\), because
\[
B_z(c^*(p),p)
=1-p-2(1+5p)a(1-t-p)
=\lambda(p).
\]
Thus the pointwise contact condition in Theorem
\ref{thm:global-optimality-certificate} holds. The cutoff has only the
positive curvature atom \(a\delta_{1-t}\), where
\(h(1-t)=J_G(1-t)\), so the
curvature contact condition also holds. The theorem proves that \(c^*\) is
globally optimal.

The same argument also gives strictness. For any other finite-rent cutoff,
the first term in the gap decomposition
\eqref{eq:certificate-gap-decomposition} is strictly positive whenever that
cutoff differs from \(c^*\) on a set of positive \(G\)-measure, because
\(B(\cdot,p)-\lambda(p)(\cdot)\) is strictly concave. A concave cutoff cannot
agree with \(c^*\) almost everywhere: both are continuous, whereas \(c^*\)
has an upward slope change of size \(a\) at \(1-t\). Every concave cutoff is
finite-rent, since the zero floor is feasible, and therefore gives the
principal a strictly lower payoff.

Finally, consider the direct mechanism that applies the cutoff \(c^*\),
charges \(3+a(1-\omega-t)\) when it takes the action, and charges zero otherwise.
Whatever a true type \((\theta,p)\) reports, her payoff is either
\(\theta-3-a(1-t-p)\) if the action is taken or zero otherwise. Truthful
reporting selects the better outcome, taking the action at ties. The mechanism
is therefore globally BIC/IR and implements \(c^*\). Its lowest-type utility
is \(a\bigl(p-(1-t)\bigr)_+=R^*(p)\), and its realized transfers are bounded.
\end{proof}
\end{toappendix}

\begin{toappendix}
\subsection{Quadratic Cutoff Loss: Optimality and Computation}
\label{app:quadratic-cutoff-loss}
\label{subsec:quadratic-cutoff-loss}
\label{app:quadratic-joint-program}
\label{app:quadratic-numerical-approximation}

This subsection specializes the virtual-surplus loss without changing the
reporting-incentive requirements. We first state the quadratic restriction
and establish existence and uniqueness of the optimal cutoff. We then give
a convenient convex formulation and describe its numerical approximation.
Appendix \ref{app:dependent-capacity-expansion} solves a particular environment
explicitly.

\begin{assump}[Quadratic cutoff loss]
\label{assump:quadratic-cutoff-loss}
There are measurable functions
\(m:\mathcal P\to(\underline\theta,\bar\theta)\) and
\(\kappa:\mathcal P\to(0,\infty)\), with \(\kappa\in L^1(G)\), such that, for
\(G\)-almost every \(p\) and every \(z\in[\underline\theta,\bar\theta]\),
\begin{align}
B(z,p)
=
B(m(p),p)
-\frac{\kappa(p)}{2}[z-m(p)]^2.
\label{eq:quadratic-cutoff-loss}
\end{align}
\end{assump}

The assumption gives an exact quadratic form for the loss
\(\mathcal L\), not a local approximation or an additional cost. For
\(G\)-almost every \(p\), it makes \(m(p)\) the unique pointwise maximizer of
\(B(\cdot,p)\), so \(m=c^0\) \(G\)-almost everywhere in this specialization.
Under Assumption \ref{assump:conditional-density}, it is equivalent to
\begin{align}
W(\theta,p)f(\theta\mid p)
=
\kappa(p)[\theta-m(p)]
\quad\text{for Lebesgue-almost every }\theta\in\Theta.
\label{eq:quadratic-cutoff-loss-assumption}
\end{align}
The equivalence follows from \(B_z=-Wf\), which holds almost everywhere. The
\(h=2\) running-example baseline is the special case
\(\Theta=[-1,1]\), \(\kappa(p)=1\), and \(m(p)=1-p\). More generally, if
\(\theta\mid p\) is uniform on
\([\underline\theta,\bar\theta]\), then
\[
\kappa(p)=\frac{2}{\bar\theta-\underline\theta},
\qquad
m(p)=\frac{\bar\theta-V(p)}{2},
\]
provided \(m(p)\) lies in the interior of \(\Theta\).

Equation \eqref{eq:quadratic-cutoff-loss} makes cutoff virtual surplus strictly
concave. Theorem \ref{thm:threshold-sufficiency} therefore implies that
restricting attention to deterministic thresholds does not reduce the
principal's attainable payoff, whether \(m\) is concave, convex, or has mixed
curvature across beliefs.\footnote{If \(G\) has gaps, implementation rent
depends on \(m\) at unrealized as well as realized reports. Thus statements
below about \(\operatorname{Rent}(m)\) concern the full schedule
\(m:\mathcal P\to\Theta\).}

\paragraph{The optimal cutoff.} Substituting this loss into
\eqref{eq:rent-adjusted-concavification} gives the floor-free program
\begin{align}
\min_{\substack{c:\mathcal P\to[\underline\theta,\bar\theta]\\
c\text{ locally DC}}}
\quad&
\frac12\int_{\mathcal P}\kappa(p)[c(p)-m(p)]^2\,dG(p)
+
\int_{\mathcal P}J_G(t)\,d(D^2c)_+(t),
\label{eq:quadratic-curvature-program}
\end{align}
The objective is \(+\infty\) for a cutoff with infinite implementation rent.
The first term measures the virtual surplus sacrificed by moving away from
the benchmark; the second is the resulting implementation rent. The next
proposition establishes that a best cutoff exists and is globally optimal
among all BIC/IR mechanisms, including randomized ones.

\begin{prop}[Optimal cutoff choice under quadratic loss]
\label{prop:quadratic-global-solution}
Suppose Assumptions \ref{assump:conditional-density},
\ref{assump:belief-support-span}, and \ref{assump:quadratic-cutoff-loss} hold.
The floor-free program \eqref{eq:quadratic-curvature-program} and the
equivalent joint program obtained from
\eqref{eq:rent-adjusted-global-program} are attained and have the same
minimum and the same set of optimal cutoffs. The unrestricted value equals
\(\int_{\mathcal P}B(m(p),p)\,dG(p)\) minus this minimum. The common optimal
cutoff is unique \(G\)-almost everywhere, and every unrestricted optimal
allocation agrees \(H\)-almost everywhere with its induced threshold.
Implementing floors, tilts, and transfers need not be unique.
\end{prop}

\begin{proof}
Assumption \ref{assump:quadratic-cutoff-loss} makes
\(\mathcal L(c(p),p)=\kappa(p)[c(p)-m(p)]^2/2\). Substitution into
\eqref{eq:rent-adjusted-concavification} gives
\eqref{eq:quadratic-curvature-program}. Program \eqref{eq:rent-adjusted-global-program} gives the equivalent attained joint
program. Because
\(B(\cdot,p)\) is strictly concave for \(G\)-almost every \(p\), Theorem
\ref{thm:threshold-sufficiency} makes their common value equal to the
unrestricted value.

To prove cutoff uniqueness, suppose two optimal pairs have cutoffs \(c_1\) and
\(c_2\) that differ on a set of positive \(G\)-measure. Their componentwise
average is feasible. Its floor cost equals the average floor cost, while
strict convexity of the quadratic term makes its cutoff loss strictly smaller
than the average cutoff loss. This contradicts optimality. The strict part of
Theorem \ref{thm:threshold-sufficiency} then implies that every unrestricted
optimal allocation agrees \(H\)-almost everywhere with this threshold.
\end{proof}

\paragraph{Computation.} For computation, it is convenient to choose the
cutoff and floor jointly. Under the proposition's assumptions, the equivalent
formulation is
\begin{align}
\min_{c,R}\quad
&
\frac12\int_{\mathcal P}\kappa(p)[c(p)-m(p)]^2\,dG(p)
+\int_{\mathcal P}R(p)\,dG(p)
\label{eq:quadratic-global-program}\\
\text{subject to}\quad
&
\underline\theta\le c\le\bar\theta,
\qquad R\ge0,
\qquad R\text{ and }R-c\text{ convex on }\mathcal P.
\nonumber
\end{align}
The program is attained and has the same minimum and optimal cutoffs as the
floor-free formulation. Its objective is quadratic in \(c\) and linear in
\(R\). Its feasible set is convex: averaging two feasible cutoff--floor pairs
preserves the bounds, nonnegativity, and both convexity requirements.

For a numerical approximation, represent \(c\) and \(R\) on a belief grid by
piecewise-linear schedules and replace the integrals with positive-weight
quadrature. Convexity of \(R\) and \(R-c\) becomes a collection of linear
restrictions on adjacent slopes, while the objective is quadratic. Each
finite-dimensional approximation can therefore be solved with a standard
convex quadratic-programming solver. Convergence of these finite-grid
solutions to the continuum optimum is not established here. Appendix
\ref{app:constructive-cutoff-selection} gives a separate construction with
certified payoff bounds and convergence to an optimal cutoff.

This formulation need not have a closed-form solution. The next subsection
gives an exact construction under additional conditions, without requiring
quadratic loss. Its quadratic specialization then derives the dependent
example's optimal cutoff. Appendix \ref{app:quadratic-affine-rays} considers
the separate, more restricted choice of how far to move toward a given
affine target.

\end{toappendix}

\begin{toappendix}
\subsection{Constructing an Optimal Clipped-Affine Cutoff}
\label{app:clipped-affine-construction}

The general construction in Appendix \ref{app:constructive-cutoff-selection}
selects an optimum through grid refinement. Under the additional conditions
below, we can determine the optimal shape directly. Two scalar equations
then determine its level and slope, accounting for a binding lower cutoff
bound.
We give the construction for general cutoff virtual surplus and then
specialize it to quadratic loss.

The key economic restriction is that action has positive total surplus even
at the lowest stake. This makes it profitable to reduce a positive floor
and the cutoff together until one reaches its lower bound. We also require
strictly concave cutoff virtual surplus and a convex pointwise benchmark
that stays away from the upper stake bound. These are restrictions on the
environment, not on competing mechanisms or on the curvature of the chosen
cutoff.

\begin{prop}[A clipped-affine construction]
\label{prop:clipped-affine-construction}
Maintain Assumption \ref{assump:conditional-density} and the paper's
full-domain BIC/IR and interim-integrable contract class. Write
\(L:=\underline\theta\), \(H:=\bar\theta\), and
\(\mathcal P=(\ell,r)\). Suppose the following additional conditions hold:
\begin{enumerate}
\item The belief density \(g\) extends continuously to \([\ell,r]\) and
is bounded above and away from zero there. Let \(\mu=\int p\,dG(p)\).
\item The primitive surplus \(B(z,p)=(V(p)+z)[1-F(z\mid p)]\) and its
cutoff derivative \(B_z\) extend continuously to
\([L,H]\times[\ell,r]\). For some \(\gamma>0\),
\(B(z,p)+\gamma z^2/2\) is concave in \(z\) for every \(p\).
At every actual report \(p\in\mathcal P\),
\(V(p)+L>0\) and \(0<B_z(L,p)<1\).
\item The unique pointwise benchmark \(m(p)=\arg\max_{z\in[L,H]}B(z,p)\)
is convex in \(p\), and \(L<m(p)\le M<H\) on \(\mathcal P\), for a
known constant \(M\).
\end{enumerate}
The unique optimal cutoff is
\begin{align}
c^*(p)&=\max\{L,k^*(p)\},
& R^*(p)&=(L-k^*(p))_+,
& k^*(p)&=a^*+s^*(p-\mu).
\label{eq:clipped-affine-output}
\end{align}
Here \(R^*\) is the canonical least-cost floor, and the scalar construction
below determines \((a^*,s^*)\). In particular,
\(L<k^*(\mu)\) and \(k^*\le M\). Completing this pair with the established
transfer formulas gives a globally optimal mechanism, including against
randomized BIC/IR competitors.
\end{prop}

The pointwise inequalities \(V+L>0\), \(0<B_z(L,p)<1\), and \(L<m(p)\)
concern the open report interval; their endpoint limits need not be strict.
The uniform density, strong-concavity, and upper-separation bounds remain
in force. Continuity and strict concavity of \(B\) make \(m\) continuous.
The construction uses the full function \(B\), not just its maximizer \(m\).

\paragraph{Determining the coefficients.}
Trial lines can exceed \(H\), so fix any \(\eta>0\) and extend \(B\) only
for evaluating those trials:
\[
\widetilde B(z,p)=
\begin{cases}
B(z,p),&L\le z\le H,\\
B(H,p)+B_z(H,p)(z-H)-\eta(z-H)^2/2,&z>H.
\end{cases}
\]
This extension is not an admissible cutoff or a change in primitives.
The returned line lies below \(M<H\). For
\(k_{a,s}(p)=a+s(p-\mu)\), define the marginal
\begin{align}
q_{a,s}(p)&=
\begin{cases}
1,&k_{a,s}(p)<L,\\
\widetilde B_z(k_{a,s}(p),p),&k_{a,s}(p)\ge L.
\end{cases}
\label{eq:clipped-affine-marginal}
\end{align}
Below \(L\), raising the line reduces the floor one-for-one without changing
the cutoff. Above \(L\), it changes cutoff virtual surplus. This difference
is why fitting an unrestricted line and then clipping it can give the wrong
answer.

\begin{enumerate}
\item For each queried slope \(s\), find the unique root \(a(s)\) of
\(Q_0(a,s):=\int q_{a,s}(p)\,dG(p)=0\). Initialize bisection with
\[
\begin{split}
a_-(s)&=L-1-\max\{s(\ell-\mu),s(r-\mu)\},\\
a_+(s)&=H+1-\min\{s(\ell-\mu),s(r-\mu)\}.
\end{split}
\]
The residual is positive at the lower endpoint and negative at the upper
endpoint. At a midpoint, a positive residual replaces the lower endpoint;
a negative residual replaces the upper endpoint. Return an exact zero or
the nested-interval limit. For \(s=0\), use the unique root of
\(\int B_z(a,p)\,dG(p)=0\) on \([L,H]\).
\item Find the unique root of
\(Q_1(s):=\int(p-\mu)q_{a(s),s}(p)\,dG(p)=0\), using the same bisection
rule on the bracket
\begin{align}
s_-=-\frac{M-L}{\mu-\ell},\qquad
s_+=\frac{M-L}{r-\mu}.
\label{eq:clipped-affine-slope-bracket}
\end{align}
Its endpoint residuals are strictly positive and strictly negative,
respectively. Set \(s^*\) equal to the root and \(a^*=a(s^*)\), then use
\eqref{eq:clipped-affine-output}. If \(s^*\ne0\), the possible clipping
point is \(t=\mu+(L-a^*)/s^*\); clipping occurs only if \(t\in\mathcal P\).
\end{enumerate}

This specifies two scalar roots, not finite exact arithmetic termination.
It presumes access to the indicated integrals and their exact signs,
including the profiled residual evaluated at the inner root. After \(n\)
outer bisections the slope bracket has width \((s_+-s_-)2^{-n}\).
The intercept root is continuous in slope, so the recovered lines, cutoffs,
and floors converge uniformly on \([\ell,r]\), and their payoffs converge.
With approximate integration and inner roots, bracket updates require
certified signs; a small floating-point residual is not an exact zero.
No error bound for uncertified numerical integration is asserted here.

\begin{proof}
We first reduce feasible cutoff--floor pairs to a concave auxiliary
schedule. We then construct the best line and show that no concave
schedule improves on it. The last step verifies the cutoff bounds and
least-cost floor. The integrated-residual argument follows established
shape-constrained-regression methods.\footnote{See
\citet[Lemma 2.6]{groeneboom2001estimation} for integrated-residual
conditions in convex regression, and
\citet[Lemma A.6]{guntuboyina2015global} for affine projection of a
concave input onto convex functions. We adapt this logic to the present
payoff, including its linear floor-cost branch and open endpoints.}

\emph{1. A lossless reduction.}
For any feasible finite-cost pair \((c,R)\), set
\(e=\min\{R,c-L\}\), \(\widetilde c=c-e\), and
\(\widetilde R=R-e\). Then
\(\widetilde R=\max\{0,R-c+L\}\) is nonnegative and convex,
\(\widetilde R-\widetilde c=R-c\) is convex, and
\(L\le\widetilde c\le H\). The new floor remains integrable. Directly
from the primitive formula for \(B\), the pointwise payoff gain is
\begin{align*}
&B(c-e,p)-(R-e)-[B(c,p)-R]\\
&\qquad=eF(c-e\mid p)
+(V(p)+c)[F(c\mid p)-F(c-e\mid p)]\ge0.
\end{align*}
Both terms are nonnegative because \(e\ge0\) and \(V+c\ge V+L>0\).
Thus it is lossless to set
\(k=c-R\), \(\widetilde c=\max\{L,k\}\), and
\(\widetilde R=(L-k)_+\). Conversely, every finite concave \(k\le H\)
with integrable negative part gives such a feasible pair. Its payoff is
\(\int C(k(p),p)\,dG(p)\), where
\[
C(z,p)=
\begin{cases}
V(p)+z,&z<L,\\
\widetilde B(z,p),&z\ge L.
\end{cases}
\]
This function is concave: at \(L\) its derivative falls from \(1\) to
\(B_z(L,p)\in(0,1)\); above \(L\) it is strongly concave. Only trial
lines use its extension beyond \(H\).

\emph{2. Constructing the best line.}
Consider \(F_0(a,s)=\int C(k_{a,s}(p),p)\,dG(p)\) over all real
coefficients. For fixed nonzero \(s\), crossings of \(L\) have measure
zero. Thus its level derivative is \(Q_0\), continuous and nonincreasing
in \(a\), and strictly decreasing once the line exceeds \(L\) on a set
of positive measure. The stated intercept bracket gives its two strict
signs, so there is exactly one root. For \(s=0\), both one-sided
derivatives at the constant line \(L\) are positive; the unique optimum
lies in \((L,H)\) and solves the stated special equation.

The level root \(a(s)\) is continuous. Indeed, the explicit brackets bound
it locally; every convergent subsequence solves the limiting equation by
dominated convergence, unless the limiting line is identically \(L\).
That exception is impossible: near that line the marginals have positive
lower limit \(B_z(L,p)\), with strictly positive integral. Uniqueness
then identifies every subsequential limit. Differentiating the profiled
objective \(F_0(a(s),s)\) by the envelope inequality gives its continuous
derivative \(Q_1(s)\). The profile is strictly concave: distinct lines
differ except at one point, and each level-optimal line is above \(L\)
on a set of positive measure, where the concavity inequality is strict.
Consequently \(Q_1\) is strictly decreasing.

The affine objective attains its maximum. It is uniformly bounded above;
as \(z\to-\infty\), \(C(z,p)\to-\infty\) linearly, and as
\(z\to+\infty\) it does so quadratically, uniformly in \(p\).
If \(|s|\to\infty\), an affine line is uniformly large in absolute
value, with one sign, on at least one end quarter of \([\ell,r]\).
The density lower bound forces its objective to \(-\infty\), regardless
of its level. If slopes stay bounded and \(|a|\to\infty\), the entire
line has this property. These facts give coercivity and hence an affine
maximizer \(k_0\). Its two marginal moments vanish. We establish the
fixed slope bracket after proving feasibility, so this existence argument
does not presuppose that bracket.

\emph{3. Comparing with every concave schedule.}
Let \(\lambda(p)=q_{a,s}(p)g(p)\) at the affine maximizer. Then
\(\int\lambda=\int p\lambda=0\), and
\(\operatorname{sgn}\lambda=\operatorname{sgn}(m-k_0)\) except at
null crossings. Convexity of \(m-k_0\) makes its negative set an interval.
Unless \(\lambda=0\) almost everywhere, its positive part must occur on
both sides of its negative part. Otherwise, separating the signs at a
point \(t\) would make \(\int(p-t)\lambda(p)\,dp\) strictly positive or
strictly negative, contrary to the two zero moments.

It follows that
\(h(t):=\int_\ell^t(t-p)\lambda(p)\,dp\ge0\).
To see the sign, the cumulative residual first increases, then decreases,
then increases back to zero. It crosses zero at most once in the middle
interval. Thus \(h\) first rises and then falls, with
\(h(\ell)=h(r)=0\). The claim also holds when \(\lambda=0\).

For any finite concave competitor \(k\le H\) with integrable negative
part, write \(\nu=-D^2k\ge0\). The representation used in Theorem
\ref{thm:finite-rent-curvature} writes \(-k\) as an affine function plus
\(\int K_\mu(p,t)\,d\nu(t)\). Integrability implies
\(\int J_G\,d\nu<\infty\). Since \(k_0\) is a bounded line,
\(|\lambda|\le C_0g\) for some finite constant \(C_0\), and
\[
\int\!\int K_\mu(p,t)|\lambda(p)|\,dp\,d\nu(t)
\le C_0\int J_G(t)\,d\nu(t)<\infty.
\]
Fubini and the two zero moments therefore give
\[
\int\lambda(p)[k(p)-k_0(p)]\,dp
=-\int h(t)\,d\nu(t)\le0.
\]
Here \(\int K_\mu(p,t)\lambda(p)\,dp=h(t)\) on either side of
\(\mu\), again by the zero moments. Concavity of \(C\) now implies
\(\int[C(k,p)-C(k_0,p)]\,dG\le0\). This comparison allows infinite
total endpoint curvature; it imposes neither endpoint reports nor finite
endpoint slopes.

\emph{4. Feasibility and recovery.}
If \(k_0>M\) anywhere, the concave competitor \(\min\{k_0,M\}\)
strictly raises payoff on a set of positive measure, because \(C\) is
strictly decreasing above \(m\le M\). This contradicts Step 3.
Thus \(k_0\le M<H\), and the artificial extension is unused in the
returned mechanism. Put \(\rho=g-\lambda\). It vanishes below \(L\)
and is strictly positive where \(k_0>L\), since
\(B_z(k_0,p)\le B_z(L,p)<1\). Its mass is one and its mean is \(\mu\).
As \(\{k_0>L\}\) is an interval, its mean lies strictly inside it.
Hence \(a^*=k_0(\mu)>L\). Together with \(k_0\le M\), this puts its
slope strictly inside \eqref{eq:clipped-affine-slope-bracket} and proves
the claimed endpoint signs for \(Q_1\).

The cutoff \(\max\{L,k_0\}\) has at most one upward kink.
Its floor \((L-k_0)_+\) matches exactly that curvature and is zero at
\(\mu\), so it is the canonical floor from Theorem
\ref{thm:finite-rent-curvature}. Step 1 extends optimality to every
finite-cost feasible cutoff--floor pair. Strict concavity of \(B\),
convexity of the joint feasible set, and full support imply uniqueness of
the cutoff at every report. Proposition
\ref{prop:threshold-implementability} recovers full-domain BIC/IR
transfers. Finally, Theorem \ref{thm:threshold-sufficiency} applies because
\(B(\cdot,p)\) is concave, giving optimality against all maintained
randomized mechanisms.

The proved continuity of \(a(s)\) gives the stated uniform convergence of
lines and of their positive parts. Eventually these lines lie below
\(H\), since the limiting line lies below \(M<H\); continuity of \(B\)
then gives payoff convergence. No assumption of dual-price attainment is
used.
\end{proof}

The economic prescription is to fit an auxiliary line, then use its
below-bound part to determine the utility floor. If it never reaches the
bound, the cutoff is affine and rent-free. If it crosses the bound inside
the report interval, clipping creates an upward kink and positive rent.
The construction selects between these cases rather than imposing either.

\paragraph{Quadratic specialization.}
Under the proposition's assumptions and quadratic loss, the marginal for
\(L\le k(p)\le H\) is \(\kappa(p)[m(p)-k(p)]\). If the fitted line stays in
\([L,H]\), its coefficients \(k(p)=\alpha+sp\) solve
\begin{align}
\begin{pmatrix}W_0&W_1\\W_1&W_2\end{pmatrix}
\begin{pmatrix}\alpha\\s\end{pmatrix}
&=\begin{pmatrix}N_0\\N_1\end{pmatrix},
& W_j&:=\int p^j\kappa(p)\,dG(p),
& N_j&:=\int p^j\kappa(p)m(p)\,dG(p).
\label{eq:clipped-affine-normal-equations}
\end{align}
The determinant is positive because the weighted belief distribution is
nondegenerate. If this line falls below \(L\), use the scalar construction
instead: its marginal there is \(1\), not the quadratic residual.
The following example has no clipping, so these two linear equations
construct its globally optimal cutoff.
\end{toappendix}

\begin{toappendix}
\subsection{Supplementary Capacity Expansion with Dependent Stakes and Beliefs}
\label{app:dependent-capacity-expansion}

This supplementary example replaces the main text's centered independent
stake family with stakes supported on \([0,1]\) and dependent on beliefs.
More optimistic managers tend to value expansion more highly, and the pointwise
benchmark requires positive implementation rent. We first specify the
environment and calculate that rent. Proposition
\ref{prop:clipped-affine-construction} then determines the optimal shape,
and two moment equations give its coefficients. We finish with its payments
and payoff.
The adjustment sacrifices a little virtual surplus to eliminate a larger
rent cost.
We write \(m=c^0\) for this environment's pointwise cutoff.

\paragraph{Primitives and pointwise benchmark.}
Restrict beliefs to \(p\in(1/2,3/4)\), write \(q:=2p-1\), and let \(q\) have
density \(g(q)=\frac83(1-q)\) on \((0,1/2)\). In calculations using \(q\),
we write its induced distribution as \(G\) and express cutoff and floor
schedules in that coordinate. The conditional distribution of stakes on
\([0,1]\) has survivor function
\begin{align}
1-F(\theta\mid p)
&=
\frac{(1-\theta)[\theta+2q(1-q)]}
{2(1-q)(\theta+q)},
\qquad \theta\in[0,1].
\label{eq:running-example-survival}
\end{align}
This survivor function is increasing in \(q\) at every interior \(\theta\):
more optimistic managers also tend to value expansion more highly. It defines
a strictly positive conditional density on \([0,1]\).

The firm's expected payoff from expansion remains \(V(p)=2p-1=q\).
Substituting \eqref{eq:running-example-survival} into
\eqref{eq:cutoff-virtual-surplus} gives
\begin{align}
B(z,p)
&=
\frac{(1-z)[z+2q(1-q)]}{2(1-q)}
\nonumber\\
&=
B(m(p),p)-\frac{[z-m(p)]^2}{2(1-q)},
\qquad
m(p):=q^2-q+\frac12.
\label{eq:running-example-quadratic-loss}
\end{align}
Thus \(m\) is the unique pointwise cutoff and \(\kappa(q)=1/(1-q)\). The
cutoff is strictly convex across beliefs: \(m''(p)=8\). Relative to the independent \(U[0,1]\)-stake counterpart's cutoff
\((1-q)/2\), the dependence lowers the cutoff most at
intermediate beliefs. The firm consequently expands for more preference
types there, causing the cutoff to fall rapidly at low beliefs and then
flatten. This upward change in slope is the positive curvature priced by
implementation rent.

The mean of \(q\) is \(2/9\). Subtracting the tangent to the quadratic
\(m\) at this mean gives its least-cost floor; its expectation is the rent
(Corollary \ref{cor:convex-threshold-inner}):
\begin{align}
R^m(q)
&=\left(q-\frac29\right)^2,
&
\operatorname{Rent}(m)
&=\frac{13}{648}>0.
\label{eq:running-example-benchmark-rent}
\end{align}
The pointwise benchmark must therefore be adjusted. At the same time,
\eqref{eq:running-example-quadratic-loss} is strictly concave in the candidate
cutoff, so Theorem \ref{thm:threshold-sufficiency} makes the best
rent-adjusted threshold optimal in the unrestricted mechanism class.

\begin{proof}[Verification of the primitives and benchmark]
The survivor function in \eqref{eq:running-example-survival} equals one at
\(\theta=0\), equals zero at \(\theta=1\), and has density
\[
f(\theta\mid p)
=
\frac{\theta^2+2q\theta+q(1-2q^2)}
{2(1-q)(\theta+q)^2}.
\]
This density is strictly positive because \(0<q<1/2\) and
\(0\le\theta\le1\).
If \(S(\theta,q)\) denotes the survivor function,
then
\[
\partial_q S(\theta,q)
=
\frac{(1-\theta)\theta[\theta+1-2q+2q^2]}
{2(1-q)^2(\theta+q)^2}>0
\]
for \(\theta\in(0,1)\). This verifies the distributional claims in the
dependent extension.

Next, \(V(p)=q\), so \eqref{eq:cutoff-virtual-surplus} and
\eqref{eq:running-example-survival} give
\[
B(z,p)
=
\frac{(1-z)[z+2q(1-q)]}{2(1-q)}.
\]
Since \(1-2m(p)=2q(1-q)\), the difference-of-squares identity
\[
(1-m)^2-(z-m)^2
=(1-z)[1+z-2m]
\]
proves \eqref{eq:running-example-quadratic-loss}. The density of \(q\)
integrates to one and has moments
\[
\mathbb E[q]=\frac29,
\qquad
\mathbb E[q^2]=\frac5{72}.
\]
Consequently,
\[
\mathbb E[m(q)]-m(\mathbb E[q])
=
\frac{25}{72}-\frac{53}{162}
=
\frac{13}{648},
\]
and the supporting-line formula in Corollary
\ref{cor:convex-threshold-inner} gives
\(R^m(q)=(q-2/9)^2\). This proves
\eqref{eq:running-example-benchmark-rent}.
\end{proof}

\paragraph{Deriving the optimal cutoff.}
Proposition \ref{prop:clipped-affine-construction} applies to these
primitives. Total surplus at the lowest stake is \(V(p)+0=q>0\).
The density of beliefs is continuous and bounded above and away from zero,
as is \(\kappa=1/(1-q)\). Cutoff virtual surplus and its derivative extend
continuously to the closed rectangle, and \(B_{zz}=-\kappa\le-1\).
Moreover, \(B_z(0,p)=m(p)/(1-q)\in(0,1)\), while
\(m\) is convex and \(0<m\le1/2<1\).
The proposition therefore establishes the clipped-affine shape before we
calculate its coefficients. The affine change of coordinate \(q=2p-1\)
preserves the construction.

Here \(\kappa(q)\,dG(q)=\frac83\,dq\), so the quadratic normal equations
\eqref{eq:clipped-affine-normal-equations} become
\[
\int_0^{1/2}[m(q)-A-Bq]\,dq
=
\int_0^{1/2}q[m(q)-A-Bq]\,dq
=0.
\]
Solving gives \(A=11/24\) and \(B=-1/2\). The resulting line lies between
\(5/24\) and \(11/24\), so no clipping or refitting is needed:
\begin{align}
c^*(q)
&=
\frac{11}{24}-\frac q2
=
\frac{23}{24}-p,
\qquad R^*(p)=0.
\label{eq:running-example-optimal-cutoff}
\end{align}
Proposition \ref{prop:clipped-affine-construction} establishes global
optimality, not merely optimality among affine cutoffs. The virtual-surplus
loss used in the payoff comparison below follows from
\[
m(q)-c^*(q)=q^2-\frac q2+\frac1{24},
\qquad
\frac43\int_0^{1/2}[m(q)-c^*(q)]^2\,dq=\frac1{4320}.
\]

\paragraph{Transfers and payoff.}
Because \(c^*\) is affine, its least-cost floor is zero. A concrete
implementation expands exactly when
\(\theta+p\ge 23/24\). Following expansion, the manager pays \(23/24\) if
demand is low and receives \(1/24\) if demand is high; no transfer is made
without expansion. Her truthful utility is
\((\theta+p-23/24)_+\), so the mechanism is globally BIC/IR and leaves zero
floor.

The pointwise benchmark generates virtual surplus \(3/10\) and requires rent
\(13/648\). The optimal adjustment eliminates this rent at a virtual-surplus
cost of only \(1/4320\). The firm's payoff therefore rises by
\[
\frac{13}{648}-\frac{1}{4320}
=
\frac{257}{12960}.
\]
Its payoff rises by approximately \(7.1\)
percent relative to implementing the pointwise benchmark.
Figure~\ref{fig:running-example-adjustment} shows how the firm achieves this
gain: it straightens rather than uniformly shifts the cutoff, expanding for
more preference types at relatively low and high beliefs and for fewer at
intermediate beliefs.
A different comparison isolates the value of responsiveness. The best
constant cutoff, \(c\equiv1/3\), ignores beliefs and yields payoff \(8/27\).
Relative to that rule, the optimal responsive cutoff adds \(1/288\), or about
\(1.17\) percent. Thus the \(7.1\) percent figure measures the gain from
adjusting a costly responsive benchmark; the \(1.17\) percent figure compares
the optimum with the best cutoff that ignores beliefs.

\begin{figure}[H]
\centering
\caption{Optimal cutoff adjustment in the dependent example}
\label{fig:running-example-adjustment}
\begin{tikzpicture}[x=34cm,y=14cm]
  \def\plow{0.5528312164}
  \def\phigh{0.6971687836}

  \path[fill=blue!12]
    plot[domain=0.5:\plow,samples=35]
      (\x,{4*\x*\x-6*\x+2.5})
    -- plot[domain=\plow:0.5,samples=35]
      (\x,{23/24-\x}) -- cycle;
  \path[fill=orange!16]
    plot[domain=\plow:\phigh,samples=55]
      (\x,{4*\x*\x-6*\x+2.5})
    -- plot[domain=\phigh:\plow,samples=55]
      (\x,{23/24-\x}) -- cycle;
  \path[fill=blue!12]
    plot[domain=\phigh:0.75,samples=35]
      (\x,{4*\x*\x-6*\x+2.5})
    -- plot[domain=0.75:\phigh,samples=35]
      (\x,{23/24-\x}) -- cycle;

  \draw[->,black!75] (0.495,0.19) -- (0.762,0.19)
    node[below right,font=\small] {belief \(p\)};
  \draw[->,black!75] (0.5,0.18) -- (0.5,0.535);
  \node[rotate=90,font=\small] at (0.472,0.36)
    {minimum stake for expansion};

  \foreach \x/\lab in {0.5/{\frac12},0.625/{\frac58},0.75/{\frac34}} {
    \draw[black!75] (\x,0.186) -- (\x,0.194)
      node[below=3pt,font=\footnotesize] {$\lab$};
  }
  \foreach \y/\lab in {0.25/{\frac14},0.375/{\frac38},0.5/{\frac12}} {
    \draw[black!75] (0.497,\y) -- (0.503,\y)
      node[left=3pt,font=\footnotesize] {$\lab$};
  }

  \draw[very thick,dashed,black!65]
    plot[domain=0.5:0.75,samples=80]
      (\x,{4*\x*\x-6*\x+2.5});
  \draw[very thick,black]
    plot[domain=0.5:0.75,samples=2]
      (\x,{23/24-\x});

  \draw[very thick,dashed,black!65] (0.600,0.505) -- (0.620,0.505)
    node[right,font=\footnotesize] {pointwise benchmark \(c^0(p)\)};
  \draw[very thick,black] (0.600,0.475) -- (0.620,0.475)
    node[right,font=\footnotesize] {rent-adjusted optimum \(c^*(p)\)};
\end{tikzpicture}
\par\smallskip
\begin{minipage}{\textwidth}
\small
\textbf{Notes:} The dashed curve maximizes virtual surplus separately at each
belief. Dependence makes it especially permissive at intermediate beliefs, so
its slope rises with optimism; this positive curvature requires implementation
rent. The solid affine rule holds the cutoff's response to beliefs constant and
therefore requires no rent. Because a lower cutoff expands for more preference
types, the adjustment adds expansion in the two outer shaded regions and
removes it in the middle.
\end{minipage}
\end{figure}

\begin{center}
\begin{minipage}{\textwidth}
\centering
\begin{tabular}{@{}lrr@{}}
\toprule
Expected value & Pointwise benchmark & Rent-adjusted optimum\\
\midrule
Virtual surplus & 0.3000 & 0.2998\\
Implementation rent & 0.0201 & 0\\
Firm payoff & 0.2799 & 0.2998\\
\bottomrule
\end{tabular}
\par\smallskip
{\small Values are rounded. Firm payoff equals virtual surplus minus
implementation rent.\par}
\end{minipage}
\end{center}

The optimum is affine and rent-free in this example, not by restriction.
Appendix \ref{app:positive-rent-soft-concavification} shows why eliminating
rent need not maximize profit in general.
\end{toappendix}

\begin{toappendix}
\subsection{Affine-Ray Adjustments under Quadratic Loss}
\label{app:quadratic-affine-rays}

The preceding example constructs the globally optimal cutoff. This
subsection answers a more limited question:
given an affine target \(a_0\) and the benchmark rent
\(\operatorname{Rent}(m)\), how far should the principal move toward that
target? The next result solves this one-dimensional choice within the family
\(c_t=(1-t)m+ta_0\). It does not assert that the unrestricted optimum lies in
that family.

\begin{corollary}[Optimal affine-ray adjustment]
\label{cor:quadratic-affine-ray}
Under the assumptions of Proposition \ref{prop:quadratic-global-solution},
suppose \(\operatorname{Rent}(m)<\infty\), let
\(a_0:\mathcal P\to[\underline\theta,\bar\theta]\) be an affine cutoff, and
define
\[
D(a_0)
:=
\int_{\mathcal P}\kappa(p)[m(p)-a_0(p)]^2\,dG(p),
\qquad
c_t:=(1-t)m+ta_0.
\]
Then, for every \(t\in[0,1]\),
\begin{align}
\operatorname{Rent}(c_t)
&=(1-t)\operatorname{Rent}(m),
\label{eq:quadratic-ray-rent}\\
\mathcal J(c_t)-\mathcal J(m)
&=
t\operatorname{Rent}(m)
-\frac{t^2}{2}D(a_0).
\label{eq:quadratic-ray-gain}
\end{align}
If \(0<\operatorname{Rent}(m)<\infty\), the payoff-maximizing step along this
ray is
\begin{align}
t_{\mathrm{ray}}^*
&:=
\begin{cases}
\min\{\operatorname{Rent}(m)/D(a_0),1\},&D(a_0)>0,\\[3pt]
1,&D(a_0)=0,
\end{cases}
\label{eq:quadratic-ray-optimum}
\end{align}
and \(\mathcal J(c_{t_{\mathrm{ray}}^*})>\mathcal J(m)\).
\end{corollary}

\begin{proof}
Fix \(t<1\). If \(R\) is feasible for \(m\), then \((1-t)R\) is nonnegative
and convex, and
\[
(1-t)R-c_t
=
(1-t)(R-m)-ta_0
\]
is convex because \(a_0\) is affine. Taking the infimum over feasible \(R\)
gives
\[
\operatorname{Rent}(c_t)
\le
(1-t)\operatorname{Rent}(m).
\]
Conversely, if \(S\) is feasible for \(c_t\), then \(S/(1-t)\) is nonnegative
and convex, and
\[
\frac{S}{1-t}-m
=
\frac{S-c_t+ta_0}{1-t}
\]
is convex. Taking the infimum over feasible \(S\) gives the reverse inequality
\[
\operatorname{Rent}(c_t)
\ge
(1-t)\operatorname{Rent}(m).
\]
At \(t=1\), \(c_t=a_0\), and the zero floor is feasible because \(-a_0\) is
affine. This proves \eqref{eq:quadratic-ray-rent} on \([0,1]\).

Equation \eqref{eq:quadratic-cutoff-loss} gives
\[
\operatorname{VS}(c_t)-\operatorname{VS}(m)
=
-\frac{t^2}{2}D(a_0).
\]
Combining this equality with \eqref{eq:quadratic-ray-rent} proves
\eqref{eq:quadratic-ray-gain}. If \(r:=\operatorname{Rent}(m)>0\) and
\(D(a_0)>0\), the gain \(rt-t^2D(a_0)/2\) is strictly concave in \(t\) and has
derivative \(r-tD(a_0)\). Its maximizer on \([0,1]\) is therefore
\(\min\{r/D(a_0),1\}\). If \(D(a_0)=0\), the gain is \(rt\) and is maximized
at \(t=1\). In either case, the maximal gain is strictly positive.
\end{proof}

When the benchmark rent is positive, Corollary
\ref{cor:quadratic-affine-ray} gives a closed-form solution to this affine-ray
problem---conditional on the fixed target \(a_0\) and the benchmark rent
\(\operatorname{Rent}(m)\)---not to the global program
\eqref{eq:quadratic-curvature-program}. If the benchmark rent is not available
in elementary closed form, it can be evaluated from the curvature formula
\eqref{eq:curvature-rent-formula}, or from the equivalent floor problem
\eqref{eq:threshold-rent}. The corollary nevertheless constructs a strictly
profitable feasible adjustment, selects the best step along the chosen
direction, and makes the economic tradeoff transparent. Rent falls linearly
with \(t\), whereas virtual surplus falls quadratically.

The affine target need not be chosen arbitrarily. Among affine cutoffs taking
values in \([\underline\theta,\bar\theta]\), choose \(a_0\) to minimize
\(D(a_0)\); this is the constrained \(\kappa\,dG\)-weighted least-squares affine
approximation to \(m\). Because the gain in
\eqref{eq:quadratic-ray-gain} is weakly decreasing in \(D(a_0)\) for every
step \(t\), combining this target with
\eqref{eq:quadratic-ray-optimum} yields the highest optimized payoff across
all admissible affine rays. The globally optimal cutoff still need not lie on
the resulting ray.
\end{toappendix}

\Needspace{13\baselineskip}
\section{Conclusion}
\label{sec:conclusion}

A state-contingent expansion price lets the firm use a manager's private
forecast while screening for her stake. The public-forecast rule gives
the firm a starting point. Under the interior benchmark theorem's
density and support conditions, keeping that rule is optimal exactly
when the minimum stake for expansion is concave in the forecast. Upward bends
carry implementation rent. When the interior benchmark has finite
positive rent, the firm gains by changing the rule.

In the dependent example, the firm eliminates rent at little cost in
virtual surplus. In the uniform example with small stakes, some rent
remains. Compared with keeping the public-forecast rule under private
reporting, changing expansion decisions near break-even saves more in
rent than it costs in virtual surplus.

\clearpage
\singlespacing

\bibliographystyle{aer}
\bibliography{multimech}

\end{document}

%% file: includes/preamble.tex
 \usepackage{graphicx} 
\usepackage[top=1in, bottom=1in, left=1in, right=1in]{geometry}
\usepackage{natbib} 
\usepackage{mdframed}
\usepackage{lipsum}
\usepackage{setspace} 

\usepackage{placeins}
\usepackage{dcolumn} 
\usepackage{titlesec}
\usepackage{titling}

\usepackage{cancel}
\usepackage{amsmath} 
\DeclareMathOperator*{\argmax}{arg\,max}

\usepackage{accents}
\usepackage{tablefootnote}

\usepackage{amscd} 
\usepackage{xpatch}
\usepackage{mathpazo} 
\usepackage[titletoc]{appendix}
\usepackage{sgame} 
\usepackage{comment} 
\usepackage{titlesec} 
\usepackage{amsmath} 
\usepackage{amsthm}
\usepackage{titletoc}
\usepackage{amssymb} 
\usepackage{bbm}
\usepackage[normalem]{ulem}
\usepackage{tikz}
\usepackage{thmtools}
\usepackage{thm-restate}
\usepackage[bibliography=common]{apxproof}
\usepackage{amscd} 
\usepackage{sgame} 
\usepackage[hyphens]{url} 
\usepackage{pdflscape}
\usepackage{graphicx} 
\usepackage{float}
\usepackage{caption}
\usepackage{threeparttable}
\usepackage{booktabs}
\usepackage{booktabs}
\usepackage{titletoc}
\usepackage{array}
\usepackage{mathtools}
\usepackage{ushort}
\usepackage{stackengine}

\newtheorem{theorem}{Theorem}

\newtheorem{example}{Example}

\newtheoremstyle{extstyle}
{20pt}
{0pt}
{}
{}
{\bfseries}
{}
{.5em}
{}
\theoremstyle{extstyle}

\newtheoremstyle{runexampstyle}
{10pt}
{0pt}
{}
{}
{\bfseries}
{.}
{.5em}
{}
\theoremstyle{runexampstyle}
\NewDocumentEnvironment{runexamp}{o}{%
  \par\noindent
  \textbf{Running Example\IfValueT{#1}{ (#1)}.}\enspace
  \ignorespaces
}{%
  \par
}

\mdfdefinestyle{runningexamplebox}{%
  linewidth=0.6pt,
  linecolor=black!50,
  backgroundcolor=black!2,
  innertopmargin=8pt,
  innerbottommargin=8pt,
  innerleftmargin=10pt,
  innerrightmargin=10pt,
  skipabove=12pt,
  skipbelow=12pt,
  nobreak=false,
  splittopskip=8pt,
  splitbottomskip=8pt
}
\theoremstyle{plain}

\defcitealias{kuhn1955hungarian}{Kuhn-Munkres}
\usepackage{relsize,etoolbox}
\AtBeginEnvironment{quote}{\small\it}
\setstackEOL{\!}

\usepackage{hyperref} \hypersetup{
    colorlinks = true,
    linkcolor = red,
    anchorcolor = red,
    citecolor = blue,
    filecolor = red,
    urlcolor = blue
}
\usepackage{eqparbox}

\declaretheorem[name=Proposition]{prop}

\declaretheorem[name=Corollary]{corollary}
\declaretheorem[name=Definition]{definition}
\declaretheorem[name=Assumption]{assump}

\usepackage{optidef}
\robustify\vec
\usepackage{pict2e,picture,graphicx}

\makeatletter
\DeclareRobustCommand{\Arrow}[1][]{%
\check@mathfonts
\if\relax\detokenize{#1}\relax
\settowidth{\dimen@}{$\m@th\rightarrow$}%
\else
\setlength{\dimen@}{#1}%
\fi
\sbox\z@{\usefont{U}{lasy}{m}{n}\symbol{41}}%
\begin{picture}(\dimen@,\ht\z@)
\roundcap
\put(\dimexpr\dimen@-.7\wd\z@,0){\usebox\z@}
\put(0,\fontdimen22\textfont2){\line(1,0){\dimen@}}
\end{picture}%
}
\makeatother

\usepackage{cleveref}
\definecolor{vermillion}{HTML}{D55E00}
\definecolor{bluishgreen}{HTML}{009E73}
\definecolor{skyblue}{HTML}{56B4E9}
\definecolor{pastelblue}{HTML}{B3CDE3}  
\definecolor{pastelpink}{HTML}{FBB4AE}  
\definecolor{pastelgreen}{HTML}{CCEBC5} 
\definecolor{lightpink}{HTML}{FFC0CB}   
\definecolor{mediumpurple}{HTML}{9370DB}
\definecolor{darkblue}{HTML}{00008B}    

\usepackage{pict2e,picture,graphicx}

\newtheoremstyle{revcom}
{25pt}
{10pt}
{\color{blue}\itshape}
{}
{\bfseries\color{black}}
{.}
{.5em}
{}
\theoremstyle{revcom}

\usepackage{enumitem}

\input{includes/online-appendix}

%% file: includes/online-appendix.tex
\makeatletter
\newcommand{\printonlineappendixcontents}{%
  \startcontents[sections]%
  \printcontents[sections]{}{1}{\setcounter{tocdepth}{2}}%
  \ttl@startlists
  \global\let\ttl@startlists\@empty
}
\makeatother

\newcommand{\numberonlineappendixcounter}[1]{%
  \counterwithin*{#1}{section}%
  \csdef{the#1}{\Alph{section}\arabic{#1}}%
  \csdef{theH#1}{appendix.\Alph{section}.\arabic{#1}}%
}

\newif\ifonlineappendixreferences
\onlineappendixreferencestrue

\renewcommand{\appendixprelim}{%
  \clearpage
  \onecolumn
  \renewcommand*{\thesubsection}{\Alph{section}\arabic{subsection}}%
  \renewcommand*{\thesubsubsection}%
    {\Alph{section}\arabic{subsection}.\arabic{subsubsection}}%
  \forcsvlist{\numberonlineappendixcounter}%
    {theorem,proposition,prop,lemma,corollary,definition,assump,claim,%
     example,table,figure}%
  \pagenumbering{arabic}
  \renewcommand*{\thepage}{a\arabic{page}}%
  \setlength{\emergencystretch}{2em}%
  \singlespacing
  \begin{center}
    {\LARGE\textbf{Online Appendix to}\par}
    \vspace{1mm}
    {\Large\textbf{``Contracting for Decision-Relevant Beliefs''}\par}
    \vspace{2mm}
    {\large\textbf{For Online Publication Only}\par}
  \end{center}
  \vspace{5mm}
  \noindent{\large\textbf{Table of Contents}}\par
  \vspace{2mm}
  \hrule
  \vspace{2mm}
  \printonlineappendixcontents
  \vspace{2mm}
  \hrule
  \vspace{4mm}
  \onehalfspacing
}

\renewcommand{\appendixbibliographyprelim}{%
  \stopcontents[sections]%
  \ifonlineappendixreferences\else\let\putbib\relax\fi
}